\documentclass[journal,12pt,onecolumn,draftclsnofoot,]{IEEEtran}

\usepackage{graphicx}%
\usepackage{array}    %
\usepackage{booktabs} %
\usepackage{multirow}%
\usepackage{amsmath,amssymb,amsfonts}%
\usepackage{amsthm}%
\usepackage{mathrsfs}%
\usepackage{xcolor}%
\usepackage{textcomp}%
\usepackage{manyfoot}%
\usepackage{booktabs}%
\usepackage{verbatim}%
\usepackage{orcidlink}
\usepackage{placeins}
\usepackage{cite}

\usepackage{makecell} %

\newcolumntype{M}[1]{>{\centering\arraybackslash}m{#1}}

\newcolumntype{R}[1]{>{\raggedright\arraybackslash}p{#1}}

\usepackage[ruled,vlined,titlenumbered,linesnumbered]{algorithm2e}
\DontPrintSemicolon
\SetKwInOut{Input}{Input}\SetKwInOut{Output}{Output}
\SetKwInOut{Initialize}{Initialize}
\SetKwFunction{and}{\KwSty{and}}
\SetKwFunction{andand}{\KwSty{and}}
\SetKwFunction{oror}{\KwSty{or}}

\usepackage{algpseudocode}%
\usepackage{listings}%
\usepackage{bm}
\usepackage{datetime}

\usepackage{tikz} %
\usepackage{pgfplots}
\usetikzlibrary{patterns, decorations.pathmorphing, spy}
\usepgfplotslibrary{fillbetween} 
\pgfplotsset{compat=newest}

\definecolor{color1}{HTML}{d7191c}
\definecolor{color2}{HTML}{fdae61}
\definecolor{color3}{HTML}{abd9e9}
\definecolor{color4}{HTML}{2c7bb6}

\definecolor{oiBlack}{HTML}{000000}
\definecolor{oiOrange}{HTML}{E69F00}
\definecolor{oiSkyBlue}{HTML}{56B4E9}
\definecolor{oiBluishGreen}{HTML}{009E73}
\definecolor{oiYellow}{HTML}{F0E442}
\definecolor{oiBlue}{HTML}{0072B2}
\definecolor{oiVermillion}{HTML}{D55E00}
\definecolor{oiReddishPurple}{HTML}{CC79A7}

\newcommand{\MyLineWidth}{1.5pt}

\pgfplotsset{
    WgenUB/.style={
        line width=\MyLineWidth, 
        color=oiBlue, 
        solid,
    },
    WgenUBsimple/.style={
        line width=\MyLineWidth, 
        color=oiBlue, 
        dashed,
    },
    WgenUBsimpleImproved/.style={
        line width=\MyLineWidth, 
        color=oiBlue, 
        dashdotted,
    },
    WgenRCUUB/.style={
        line width=\MyLineWidth,
        color=oiBlue,
        dashed,
    },    
    WgenLB/.style={
        line width=\MyLineWidth, 
        color=oiSkyBlue, 
        solid,
    },
    WgenRCULB/.style={
        line width=\MyLineWidth,
        color=oiSkyBlue, 
        dotted,
    },    
    WrandUB/.style={
        line width=\MyLineWidth, 
        color=oiOrange, 
        solid,
    },
    WrandRCUUB/.style={
        line width=\MyLineWidth,
        color=oiOrange,
        dashed,
    },    
    WrandLB/.style={
        line width=\MyLineWidth, 
        color=oiVermillion, 
        solid,
    },
    WrandRCULB/.style={
        line width=\MyLineWidth, 
        color=oiVermillion, 
        dotted,
    },    
    WgenavgLP/.style={
        line width=\MyLineWidth, 
        color=oiBlack, 
        solid,
    },
    WgenwcLP/.style={
        line width=\MyLineWidth, 
        color=oiBlack, 
        solid,
    },
    Wprange/.style={
      only marks,
      mark=triangle,
      mark options={fill=none, line width=1pt, draw=oiBluishGreen}, %
      color=oiBluishGreen,                %
      mark size=3pt,            %
    },
    line3/.style={
        line width=\MyLineWidth,
        color=oiReddishPurple,
        dotted,
    }
}

\tikzset{
    regiontext/.style={
        text width=2cm,
        align=center,
        font=\scriptsize
    },
    blueline/.style={
        line width=0.8pt,
        color=oiBlue,
        very thick,
        solid,
        mark options={solid}
    },
    magentaline/.style={
        line width=0.8pt,
        color=oiReddishPurple,
        thick,
        dashed,
        mark options={solid}
    },
    redline/.style={
        line width=0.8pt,
        color=oiVermillion,
        thick,
        dotted,
        mark options={solid}
    },
    nosolutionfill/.style={
        pattern=north east lines,
        pattern color=gray!50
    },
    hardfill/.style={
        pattern=north west lines,
        pattern color=oiVermillion!20
    },
    easyfill/.style={
        pattern=crosshatch,
        pattern color=oiBluishGreen!20
    },
}

\usepackage{soul}

\usepackage{todonotes}
\usepackage[nolist]{acronym}
\usepackage{mathtools}
\mathtoolsset{showonlyrefs}

\renewcommand{\todo}[1]{}
\newcommand{\tj}[1]{}
\newcommand{\hb}[1]{}
\newcommand{\aw}[1]{}

\newcommand{\F}{\ensuremath{\mathbb F}}
\newcommand{\Fq}{\ensuremath{\mathbb F_{q}}}
\newcommand{\Fqm}{\ensuremath{\mathbb F_{q^m}}}

\newcommand{\mycode}[1]{\ensuremath{\mathcal{#1}}}
\newcommand{\myspace}[1]{\ensuremath{\mathcal{#1}}}

\newcommand{\Grassm}[2]{\mathcal{G}_{#1}(#2)}

\newcommand{\Conv}{%
  \mathop{\scalebox{1.5}{\raisebox{-0.2ex}{$\circledast$}}
  }
}

\newcommand{\conv}{\ensuremath{\circledast}}

\renewcommand{\vec}[1]{{\ensuremath{\bm{#1}}}}
\newcommand{\mat}[1]{{\ensuremath{\bm{#1}}}}
\newcommand{\set}[1]{{\ensuremath{\mathcal{#1}}}}

\renewcommand{\st}{\ensuremath{:}}

\renewcommand{\a}{\vec{a}}

\renewcommand{\c}{\vec{c}}
\newcommand{\e}{\vec{e}}

\newcommand{\n}{\vec{n}}
\renewcommand{\r}{\vec{r}}
\newcommand{\s}{\vec{s}}

\renewcommand{\u}{\vec{u}}
\renewcommand{\v}{\vec{v}}
\newcommand{\w}{\vec{w}}
\newcommand{\x}{\vec{x}}
\newcommand{\y}{\vec{y}}

\newcommand{\A}{\mat{A}}
\newcommand{\B}{\mat{B}}

\newcommand{\G}{\mat{G}}
\newcommand{\I}{\mat{I}}

\renewcommand{\H}{\mat{H}}

\renewcommand{\P}{\mat{P}}

\newcommand{\shots}{\ensuremath{\ell}}
\newcommand{\shotLength}{\ensuremath{\eta}}
\newcommand{\maxShotWeight}{\ensuremath{\mu}}
\newcommand{\wdecomp}[1]{\ensuremath{\mathcal{T}_{#1}}}
\newcommand{\intparts}[1]{\ensuremath{\mathcal{J}_{#1}}}
\DeclareMathOperator{\NM}{NM}
\newcommand{\NMq}[2]{\NM_{#1}{(#2)}}

\newcommand{\errWeightDiff}{\errEx}

\newcommand{\pmfset}[1]{\mathcal{D}(#1)}

\newcommand{\probProf}[2]{\ensuremath{\varphi_{#1}(#2)}}
\newcommand{\probProfmax}[2]{\ensuremath{\varphi_{#1}^{\mathrm{(max)}}(#2)}}
\newcommand{\probProfavg}[2]{\ensuremath{\varphi_{#1}(#2)}}

\newcommand{\probProfRand}[2]{\ensuremath{\phi_{#1}(#2)}}
\newcommand{\probProfRandP}[2]{\ensuremath{\phi_{#1}'(#2)}}
\newcommand{\probProfavgRand}[2]{\ensuremath{\phi_{#1}(#2)}}

\newcommand{\pRCUBgen}{\mathrm{p}_{\mathrm{RCU}}^{(\mathrm{UB}, \mathrm{gen})}}

\newcommand{\supwt}{\ensuremath{v}}
\newcommand{\supwtVec}{\ensuremath{\v}}
\newcommand{\subwt}{\ensuremath{u}}

\DeclareMathOperator{\ucomp}{ucomp}
\DeclareMathOperator{\scomp}{vcomp}

\DeclareMathOperator{\wdcmp}{\psi}

\newcommand{\errorSet}{\ensuremath{\set{E}}}

\newcommand{\inrange}[2]{\ensuremath{\in\{#1,\ldots,#2\}}}

\newcommand{\intersectprob}[2]{\mathsf{P}_{q,#1,#2}^{\cap}}
\newcommand{\subspaceprob}[3]{\mathsf{P}_{q,#1}^{\subseteq}\left(#2 \subseteq #3\right)}

\newcommand{\Wgen}{\ensuremath{W_{\mathrm{gen}}}}
\newcommand{\Wgenwc}{\ensuremath{W_{\mathrm{gen}, \mathrm{wc}}}}
\newcommand{\WgenLB}{\ensuremath{W_{\mathrm{gen}, \mathrm{wc}}^{(\mathrm{LB})}}}
\newcommand{\WgenUB}{\ensuremath{W_{\mathrm{gen},  \mathrm{wc}}^{(\mathrm{UB})}}}
\newcommand{\WgenUBsimple}{\ensuremath{\widetilde{W}_{\mathrm{gen}, \mathrm{wc}}^{(\mathrm{UB})}}}
\newcommand{\WgenUBsimpleImproved}{\ensuremath{\widetilde{W}_{\mathrm{gen}, \mathrm{wc}}^{(\mathrm{UB},\mathrm{improved})}}}
\newcommand{\WgenRCU}{\ensuremath{W_{\mathrm{gen}, \mathrm{RCU}}}}
\newcommand{\WgenRCUUB}{\ensuremath{W_{\mathrm{gen}, \mathrm{RCU}}^{(\mathrm{UB})}}}
\newcommand{\WgenRCULB}{\ensuremath{W_{\mathrm{gen}, \mathrm{RCU}}^{(\mathrm{LB})}}}

\newcommand{\Wrand}{\ensuremath{W_{\mathrm{rand}}}}
\newcommand{\WrandLB}{\ensuremath{W_{\mathrm{rand}}^{(\mathrm{LB})}}}
\newcommand{\WrandUB}{\ensuremath{W_{\mathrm{rand}}^{(\mathrm{UB})}}}

\newcommand{\WrandRCU}{\ensuremath{W_{\mathrm{rand}, \mathrm{RCU}}}}
\newcommand{\WrandRCUUB}{\ensuremath{W_{\mathrm{rand}, \mathrm{RCU}}^{(\mathrm{UB})}}}
\newcommand{\WrandRCULB}{\ensuremath{W_{\mathrm{rand}, \mathrm{RCU}}^{(\mathrm{LB})}}}

\newcommand{\WiterGen}{W_\mathrm{gen}^{(\mathrm{iter})}}
\newcommand{\WiterRand}{W_\mathrm{rand}^{(\mathrm{iter})}}
\newcommand{\WiterEnE}{W_\mathrm{ee-dec}}
\newcommand{\WiterErDec}{W_\mathrm{erasure-dec}}

\newcommand{\WiterWC}{W_\mathrm{iter}}

\newcommand{\bnd}[2]{\ensuremath{#1\mathopen{}\left(#2\right)\mathclose{}}}
\newcommand{\oh}[1]{\bnd{O}{#1}}
\newcommand{\softoh}[1]{\bnd{\widetilde{O}}{#1}}

\newcommand{\ComplSupportDrawingHeu}{\softoh{n^3 m^2 \log_2(q)}} %

\newcommand{\event}[1]{\ensuremath{\mathcal{#1}}}
\newcommand{\uniqueEventGen}{\event{E}_\text{unique}}

\newcommand{\altEventGen}{\event{E}_\mathrm{RCU}}

\newcommand{\ProbOf}[1]{\Pr\left[#1\right]}

\newcommand{\RS}{\ensuremath{\mathcal{E}_{R}}}
\newcommand{\CS}{\ensuremath{\mathcal{E}_{C}}}
\newcommand{\RSv}{\ensuremath{\bm{\mathcal{E}}_{R}}}
\newcommand{\CSv}{\ensuremath{\bm{\mathcal{E}}_{C}}}
\newcommand{\ES}{\ensuremath{{\mathcal{E}}}}
\newcommand{\FS}{\ensuremath{{\mathcal{F}}}}
\newcommand{\GS}{\ensuremath{{\mathcal{U}}}}
\newcommand{\FSv}{\ensuremath{\bm{\mathcal{F}}}}
\newcommand{\ESv}{\ensuremath{\bm{\mathcal{E}}}}
\newcommand{\GSv}{\ensuremath{\bm{\mathcal{U}}}}

\DeclareMathOperator{\sumDim}{\ensuremath{\dim_{\Sigma}}}

\DeclareMathOperator{\ord}{sort}

\DeclareMathOperator{\diag}{diag}
\DeclareMathOperator*{\argmax}{arg\,max}

\newcommand{\ps}[1]{\ensuremath{\alpha_{#1}}}
\newcommand{\psvec}{\ensuremath{\boldsymbol{\alpha}}}
\newcommand{\psm}[1]{\ensuremath{\alpha^{(\mathrm{m})}_{#1}}}
\newcommand{\psmvec}{\ensuremath{\boldsymbol{\alpha}^{(\mathrm{m})}}}

\newcommand{\sord}{\ensuremath{{\v'}}}
\newcommand{\psord}[1]{\ensuremath{\tilde{\alpha}_{#1}}}
\newcommand{\psordvec}{\ensuremath{\tilde{\boldsymbol{\alpha}}}}

\newcommand{\pu}[1]{\ensuremath{\beta_{#1}}}

\newcommand{\pum}[1]{\ensuremath{\beta^{(\mathrm{m})}_{#1}}}
\newcommand{\pumvec}{\ensuremath{\boldsymbol{\beta}^{(\mathrm{m})}}}

\newcommand{\designPr}[1]{\ensuremath{\gamma_{#1}}}
\newcommand{\qFactorSymRand}{\tilde{Q}}

\newcommand{\puwc}[1]{\ensuremath{\tilde{\beta}_{#1}}}

\newcommand{\SG}[1]{\ensuremath{\mathfrak{S}_{\shots,\maxShotWeight}({#1})}}

\newcommand{\defeq}{\coloneqq}
\newcommand{\eqdef}{\eqqcolon}
\newcommand{\sample}{\overset{\$}{\gets}}

\newcommand{\inshots}{\ensuremath{\in\{1,\ldots,\shots\}}}
\newcommand{\inshotsarg}[1]{\ensuremath{\in\{1,\ldots,#1\}}}

\newcommand{\quadbinomq}[2]{\left[\genfrac{}{}{0pt}{}{#1\vphantom{N_N}}{#2\vphantom{N}}\right]_{q}}

\newcommand{\NN}{\ensuremath{\mathbb{Z}_{\geq 0}}}

\DeclareMathOperator{\wt}{wt}

\DeclareMathOperator{\rkq}{rk_{q}}
\DeclareMathOperator{\rkqm}{rk_{q^m}}
\DeclareMathOperator{\dimq}{\dim_{q}}

\newcommand{\SumRankWeight}{\ensuremath{\wt_{\Sigma R}}}
\newcommand{\SumRankWeightN}{\ensuremath{\wt_{\Sigma R}^{(\n)}}}

\newcommand{\SumRankDist}{d_{\ensuremath{\Sigma}R}}
\newcommand{\SumRankDistN}{d_{\ensuremath{\Sigma}R}^{(\n)}}

\newcommand{\wF}{w_F}
\newcommand{\wC}{w_C}
\newcommand{\wR}{w_R}

\newcommand{\errWeight}{\ensuremath{w}}
\newcommand{\decRadius}{\ensuremath{t}}

\newcommand{\errWeightVec}{\ensuremath{\w}}
\newcommand{\wcRP}{\ensuremath{\w_\mathrm{wc}}}

\newcommand{\subwtVec}{\ensuremath{\u}}

\newcommand{\sdInter}{\ensuremath{\epsilon}}
\newcommand{\sdInterMax}{\ensuremath{n-k}}
\newcommand{\sdInterVec}{\ensuremath{\bm{\epsilon}}}

\newcommand{\minEps}{\epsilon_\mathrm{min}}
\newcommand{\errEx}{\xi}

\newcommand{\PermCount}[1]{|\SG{#1}|}
\newcommand{\UniqueDecRad}{\ensuremath{\tau}}

\newcommand{\dmin}{\ensuremath{d_\mathrm{min}}}

\newcommand{\smax}{\ensuremath{\supwt_\mathrm{max}}}

\newcommand{\wRel}{\ensuremath{w_\mathrm{rel}}}
\newcommand{\wGVm}{\ensuremath{w^{-}}}
\newcommand{\wGVp}{\ensuremath{w^{+}}}
\newcommand{\wEasym}{\ensuremath{w_\mathrm{easy}^{-}}}
\newcommand{\wEasyp}{\ensuremath{w_\mathrm{easy}^{+}}}

\DeclareMathOperator{\LRS}{LRS}

\newcommand{\LRSCode}{\ensuremath{\mathcal{C}_{\LRS}}}
\newcommand{\Code}{\ensuremath{\mathcal{C}}}
\DeclareMathOperator{\DEC}{DEC}
\newcommand{\LinearCode}[2]{\ensuremath{{\Code[#1]}_{#2}}}

\newcommand{\boltzmParam}{\lambda}

\begin{acronym}
\acro{BCH}{Bose--Chaudhuri--Hocquenghem}
\acro{BMD}{bounded minimum distance}
\acro{ESP}{error span polynomial}
\acro{ELP}{error locator polynomial}
\acro{SRS}{skew Reed--Solomon}
\acro{ISRS}{interleaved skew Reed--Solomon}
\acro{lclm}{least common left multiple}
\acro{LRS}{linearized Reed--Solomon}
\acro{LLRS}{lifted linearized Reed--Solomon}
\acro{ILRS}{interleaved linearized Reed--Solomon}
\acro{LILRS}{lifted interleaved linearized Reed--Solomon}
\acro{MDS}{maximum distance separable}
\acro{MSRD}{maximum sum-rank distance}
\acro{MSD}{maximum skew distance}
\acro{RS}{Reed--Solomon}
\acro{REF}{row echelon form}
\acro{LEEA}{linearized extended Euclidean algorithm}
\acro{SEEA}{skew extended Euclidean algorithm}
\acro{VSD}{vector-symbol decoding}
\acro{AG}{algebraic-geometry}
\acro{KEM}{key encapsulation mechanism}
\acro{ISD}{information-set-decoding}
\acro{SRD}{sum-rank-metric decoding}
\acro{RCU}{random coding union bound}
\acro{GMD}{generalized minimum distance}
\acro{PMF}{probability mass function}
\acro{ML}{maximum-likelihood}
\acro{FL}{Faure-Loidreau}
\acro{LP}{linear program}
\end{acronym}

\newtheorem{theorem}{Theorem}
\newtheorem{definition}{Definition}
\newtheorem{lemma}{Lemma}
\newtheorem{corollary}{Corollary}
\newtheorem{problem}{Problem}
\newtheorem{remark}{Remark}
\newtheorem{proposition}{Proposition}

\title{Support-Guessing Decoding Algorithms in the Sum-Rank Metric}

\author{
 Thomas Jerkovits~\orcidlink{0000-0002-7538-7639},~\IEEEmembership{Member,~IEEE, }%
 Hannes Bartz~\orcidlink{0000-0001-7767-1513},~\IEEEmembership{Member,~IEEE, }%
 Antonia Wachter-Zeh~\orcidlink{0000-0002-5174-1947},~\IEEEmembership{Member,~IEEE}%
  \thanks{This work was supported by the German Research Council (DFG) as an ANR-DFG project under Grant no. WA 3907/9-1.}
}

\begin{document}

\maketitle

\vspace{-1em}
\begin{abstract}
The sum-rank metric generalizes the Hamming and rank metric by partitioning vectors into blocks and defining the total weight as the sum of the rank weights of these blocks, based on their matrix representation.w
In this work, we explore support-guessing algorithms for decoding sum-rank-metric codes. Support-guessing involves randomly selecting candidate supports and attempting to decode the error under the assumption that it is confined to these supports. While previous works have focused on worst-case scenarios, we analyze the average case and derive an optimal support-guessing distribution in the asymptotic regime. We show that this distribution also performs well for finite code lengths. Our analysis provides exact complexity estimates for unique decoding scenarios and establishes tighter bounds beyond the unique decoding radius.

Additionally, we introduce a randomized decoding algorithm for Linearized Reed--Solomon (LRS) codes. This algorithm extends decoding capabilities beyond the unique decoding radius by leveraging an efficient error-and-erasure decoder. Instead of requiring the entire error support to be confined to the guessed support, the algorithm succeeds as long as there is sufficient overlap between the guessed support and the actual error support. As a result, the proposed method improves the success probability and reduces computational complexity compared to generic decoding algorithms.

Our contributions offer more accurate complexity estimates than previous works, which are essential for understanding the computational challenges involved in decoding sum-rank-metric codes. This improved complexity analysis, along with optimized support-guessing distributions, provides valuable insights for the design and evaluation of code-based cryptosystems using the sum-rank metric. This is particularly important in the context of quantum-resistant cryptography.
\end{abstract}


\section{Introduction}
\label{sec:introduction}

The sum-rank metric, introduced for space-time codes in~\cite{luUnifiedConstructionSpacetime2005}, generalizes both the Hamming and rank metrics by partitioning vectors into several blocks and defining the sum-rank weight as the sum of the rank weights of its blocks. Here, the \emph{rank weight} of a block refers to the rank of its matrix representation, as used in rank-metric codes~\cite{gabidulinTheoryCodesMaximum1985}. 

In recent years, the sum-rank metric has gained significant attention, leading to the development of various code constructions and decoding algorithms~\cite{wachterPartialUnitMemory2011,wachter-zehRankMetricConvolutional2012,wachter-zehConvolutionalCodesRank2015,nappMRDRankMetric2017,martinez-penasSkewLinearizedReed2018,martinez-penasReliableSecureMultishot2019,bartzFastKotterNielsen2024,bartzDecodingInterleavedLinearized2021}. \Ac{LRS} codes, introduced by Martínez-Peñas~\cite{martinez-penasUniversalDynamicLocally2019}, achieve the Singleton-like bound in the sum-rank metric with equality and include Reed--Solomon and Gabidulin codes as special cases. These codes have found applications in multishot network coding~\cite{nobregaMultishotCodesNetwork2010,martinez-penasReliableSecureMultishot2019}, locally repairable codes~\cite{martinez-penasUniversalDynamicLocally2019}, and space-time codes~\cite{luUnifiedConstructionSpacetime2005}.

The sum-rank metric has also attracted interest in the context of code-based cryptography. It offers a balanced approach between the Hamming and rank metric, making it an appealing choice for constructing secure code-based cryptosystems resilient to quantum and classical attacks. 
Generic decoding algorithms in the rank metric, such as the GRS algorithm~\cite{gaboritComplexityRankSyndrome2016}, are more complex than generic decoding algorithms in the Hamming metric, which typically rely on (improvements of) Prange's information set decoding.
Consequently, rank-metric-based cryptography can use smaller parameters than in the Hamming metric while maintaining the same level of security, leading to reduced key sizes. However, many rank-metric cryptosystems depend on highly structured codes, which have been vulnerable to structural attacks and, in some cases, have been broken~\cite{gibsonSeverelyDentingGabidulin1995,gibsonSecurityGabidulinPublic1996,overbeckExtendingGibsonAttacks2006,rashwanSecurityGPTCryptosystem2011,wachter-zehRepairingFaureLoidreauPublicKey2018}.

Although cryptosystems based on the sum-rank metric have yet to be extensively studied, it is hoped that by carefully choosing the block sizes and the number of blocks, the sum-rank metric can be tuned to achieve a desirable balance between security and key size. This flexibility may offer resistance to attacks that exploit vulnerabilities specific to either the Hamming or rank metric, where many attacks effective in one may prove ineffective in the other~\cite{hormannDistinguishingRecoveringGeneralized2023}.

While most code-based cryptosystems, such as the McEliece cryptosystem, focus on unique decoding up to the unique decoding radius, understanding the complexity of decoding in the sum-rank metric beyond the unique decoding radius is also important for cryptographic applications. Cryptosystems like the \ac{FL}~system~\cite{faureNewPublicKeyCryptosystem2006,wachter-zehRepairingFaureLoidreauPublicKey2018,rennerLIGACryptosystemBased2021} rely on the difficulty of decoding structured codes beyond the unique decoding radius, even when the attacker knows the code's structure.

On the other hand, cryptosystems like WAVE~\cite{debris-alazardWaveNewFamily2019}, BIKE~\cite{aragon2020bike}, HQC~\cite{melchor2020hqc}, and RQC~\cite{melchorRankQuasicyclicRQC2020} require solving the syndrome decoding problem for random-like codes. Therefore, analyzing decoding complexities both within and beyond the unique decoding radius is crucial for assessing the security of such systems.

Recently, Puchinger et al.~\cite{puchingerGenericDecodingSumRank2022} proposed the first non-trivial generic decoding algorithm for arbitrary $\mathbb{F}_{q^m}$-linear codes in the sum-rank metric. Their algorithm solves the generic decoding problem for error weights up to $n-k$, where $n$ is the code length and $k$ is the code dimension. They derive upper and lower bounds on the expected decoding complexity for the worst-case rank profile, which refers to the distribution of ranks across the individual blocks of the error vector, as defined by the sum-rank metric. However, in deriving the lower bound, they assume a fixed error, meaning that this lower bound is only valid for unique decoding scenarios, such as in the McEliece setting, where the error weight is up to $(d-1)/2$, with $d$ being the minimum distance of the code.
Beyond the unique decoding radius, alternative solutions to the decoding problem may exist, and the lower bound does not account for these, making it loose in this regime. Additionally, the gap between the upper and lower bounds in~\cite{puchingerGenericDecodingSumRank2022} is significant for some parameters, indicating the potential for improvement. We need a tight lower bound on the decoding complexity to ensure that the system's best-known attack is as hard as brute-forcing all possible keys. This bound effectively serves as an upper limit on the attacker's efficiency, ensuring that any potential attack cannot be easier than brute-force. Therefore, improving the lower bound on the expected decoding complexity is essential for accurately assessing the security level and selecting appropriate parameters when designing such cryptosystems.

This paper addresses these limitations by extending the analysis to the average case (over all rank profiles). We derive a support drawing distribution tailored for the average-case scenario, optimizing the expected decoding complexity in the asymptotic setting. This solution also performs well for finite lengths, as numerical evaluations demonstrate. Our derived theoretical complexity can be computed efficiently and is exact for the case of unique decoding, where exactly one solution exists (e.g., in the McEliece setting). For decoding beyond the unique decoding radius, i.e., for error weights larger than $(d-1)/2$, our exact complexity becomes an upper bound. We introduce a lower bound that accounts for alternative solutions using \ac{RCU} bound arguments, providing a more accurate estimation of the decoding complexity in this regime.

In our previous work~\cite{jerkovitsRandomizedDecodingLinearized2023}, we generalized the randomized decoding algorithm from~\cite{rennerRandomizedDecodingGabidulin2020} in the rank metric to the sum-rank metric. In this paper, we extend that analysis by providing additional insights, proofs, and discussions and by considering the asymptotic average case, similar to our analysis of the generic decoding algorithm.

Conceptually, the randomized decoding algorithm works by randomly guessing part of the error's support. Unlike the generic decoder, which requires the guessed support to fully contain the error's support, the randomized decoding algorithm for \ac{LRS} codes only requires sufficient overlap between the guessed support and the error's support. By using an error-and-erasure decoder (e.g.,~\cite{hormannErrorErasureDecodingLinearized2022}), successful decoding is possible if enough of the error lies within the guessed support. This relaxation improves the success probability and reduces the expected decoding complexity, particularly for errors with weights only slightly above the unique decoding radius.

We analyze the probability of successful decoding for specific distributions of the guessed supports and propose methods to find the optimal distribution. By exploiting the structure of the underlying \ac{LRS} code, our approach achieves better expected computational complexity compared to the generic decoding algorithm introduced in~\cite{puchingerGenericDecodingSumRank2022}.

Additionally, we introduce a new Prange-like algorithm for the generic case in the sum-rank metric, specifically designed for the asymptotic average-case setting where the number of blocks tends to infinity and the sizes of the individual blocks are fixed. Similar to the modified Prange algorithm used in the WAVE scheme~\cite{debris-alazardWaveNewFamily2019} for the Hamming metric, our algorithm can also find a solution to the generic decoding problem for large error weights in the sum-rank metric.

\subsection*{Organization of the Paper}
In Section~\ref{sec:preliminaries}, we provide the necessary preliminaries on the sum-rank metric and \ac{LRS} codes. Section~\ref{sec:problem-description} discusses several decoding problems in the sum-rank metric that are of interest for various applications, highlighting their relevance to cryptography. In Section~\ref{sec:generic-decoding}, we extend the analysis of the generic decoding algorithm to the average case and derive optimal support drawing distributions. Additionally, we provide an improved simple closed-form upper bound that enhances the simple bound introduced in~\cite{puchingerGenericDecodingSumRank2022}. Section~\ref{sec:generic-decoding-prange} introduces the new Prange-like algorithm for the generic case in the sum-rank metric, specifically designed for the asymptotic average-case setting with a large number of blocks. Section~\ref{sec:random_decoding} presents the extended analysis of the randomized decoding algorithm for \ac{LRS} codes, along with its optimization. Finally, in Section~\ref{sec:conclusion}, we conclude the paper and discuss potential directions for future work.

\section{Preliminaries}\label{sec:preliminaries}

This section introduces the notation used throughout the paper, defines the sum-rank metric, and presents the concepts of linear codes and the considered channel model. We also define the error support in the sum-rank metric and \ac{LRS} codes.

\subsection{Notation}
In this paper, we let $n$ and $m$ be integers with $1\leq m$ and $1\leq  n$. Let $q$ be a prime power and denote the finite field of order $q$ by $\Fq$. We denote the extension field of $\Fq$ of degree $m$ as $\Fqm$. We denote the set of all $k$-dimensional subspaces of $\Fq^\mu$ as the Grassmannian $\Grassm{q}{\Fq^\mu}$.

Let $\mathcal{A}$ be a discrete set. The cardinality of $\mathcal{A}$, denoted by $|\mathcal{A}|$, is the number of elements in $\mathcal{A}$. 
We use the notation $a \sample \mathcal{A}$ to denote an element $a$ drawn uniformly at random from a set $\mathcal{A}$.
We use $[0, 1] \subset \mathbb{R}$ to denote the interval of real numbers between 0 and 1, inclusive. We define the set of all valid \acp{PMF} over $\mathcal{A}$ as
\begin{equation}
\pmfset{\mathcal{A}} = \left\{ \ps{\s} \in {[0, 1]}^{|\mathcal{A}|} \subset \mathbb{R}^{|\mathcal{A}|} : \sum_{\s \in \mathcal{A}} \ps{\s} = 1 \right\}. \label{eq:pmfset}
\end{equation}
Here, $\ps{\s} \in {[0, 1]}^{|\mathcal{A}|}$ denotes a vector of length $|\mathcal{A}|$ with real values in the interval $[0, 1]$.

Since each component $\ps{\s}$ can take any real value in $[0, 1]$ and the sum of these components must be 1, the set $\pmfset{\mathcal{A}}$ forms a probability simplex in $\mathbb{R}^{|\mathcal{A}|}$. Therefore, there are uncountably many \acp{PMF} over $\mathcal{A}$.

For a given integer $e$ and $f$ such that $1 \leq e \leq f \leq w$, we use the following notation to denote a \emph{submatrix} of $\A \in \Fq^{v \times w}$ consisting of the selected columns:
\begin{equation*}
    \A_{[e:f]} \coloneqq
    \begin{bmatrix}
        A_{1,e} & \cdots & A_{1,f} \\
        \vdots & \ddots & \vdots \\
        A_{v,e} & \cdots & A_{v,f}
    \end{bmatrix},
\end{equation*}
where $\A_{[e:f]}$ is a submatrix of $\A$ formed by all rows but only columns $e$ to $f$.

\subsection{Sum-Rank Metric}

This section provides a concise overview of the sum-rank weight and sum-rank distance, along with other relevant notations that will be used throughout the paper.
\begin{definition}[Sum-Rank Weight and Sum-Rank Distance]
Let $\x=[\x^{(1)}\mid\ldots\mid\x^{(\shots)}]\in\Fqm^n$ be a vector, that is partitioned into blocks $\x^{(i)}\in\Fqm^{n_i}$ with respect to a length profile $\n=[n_1,\ldots,n_\shots]\in\NN^\shots$.
The sum-rank weight of $\x$ with respect to the length profile $\n$ is then defined as
\begin{equation}
    \SumRankWeightN(\x) \defeq \sum_{i=1}^{\shots} \rkq\left(\x^{(i)}\right),
\end{equation}    
where $\rkq(\x^{(i)}) \defeq \dimq {\langle x_{1}^{(i)},\ldots,x_{n_i}^{(i)} \rangle}$ denotes the dimension of the $\Fq$-span of the entries of $\x^{(i)}$.
The sum-rank distance of two vectors $\x,\y\in\Fqm^n$ is then defined by
\begin{equation}
\SumRankDistN(\x,\y) \defeq \SumRankWeightN(\x-\y).    
\end{equation}
\end{definition}

In some parts of the paper, we restrict our analysis to the case where $n = \shots\shotLength$ with $\n = [\shotLength,\ldots,\shotLength]\in\NN^\shots$. In this case, we omit $\n$ in the notation and simply write $\SumRankWeight(\cdot)$ and $\SumRankDist(\cdot,\cdot)$, respectively. Furthermore, we denote the maximum possible rank weight for each block as
\begin{equation}
\maxShotWeight \defeq \min\{m,\shotLength\}.
\end{equation}

\begin{remark}
    Note that if $n_i = 1$ for all $i\in\{1,\ldots,\shots\}$, we have that the sum-rank weight of a vector $\x\in\Fqm^n$ is the number of nonzero entries of $\x$. This means, that the sum-rank weight coincides with the Hamming weight. For $\shots = 1$ the sum-rank weight coincides with the rank weight. Thus, codes in the Hamming metric and codes in the rank metric can be seen as special cases within the sum-rank metric.
\end{remark}
For a vector $\x=[\x^{(1)},\ldots,\x^{(\shots)}]\in\Fqm^n$ with sum-rank weight $w=\SumRankWeight(\x)$, we are at times interested in the rank profile of $\x$, which describes the rank weights of the individual blocks.
Therefore we define the map $\wdcmp : \Fqm^n \to \{0,\ldots,\maxShotWeight\}^\shots$ as
\begin{equation}
    \wdcmp(\x) =  [\rkq(\x^{(1)}), \ldots, \rkq(\x^{(\shots)})],
\end{equation}
and call $\wdcmp(\x)$ the \emph{rank profile of $\x$}.
Furthermore, we define the set of all possible rank profiles as follows.
\begin{definition}
Let $\errWeight$, $\shots$ and $\maxShotWeight$ be non-negative integers such that $\errWeight \leq \shots \maxShotWeight$. We define the set
\begin{equation}
    \wdecomp{\errWeight,\shots,\maxShotWeight} \defeq \left\{ \errWeightVec \in \{0,\ldots,\maxShotWeight\}^\shots \st \sum_{i=1}^{\shots} \errWeight_i = \errWeight \right\},
\end{equation}
which contains all possible rank profiles of a vector consisting of $\shots$ blocks and sum-rank weight $\errWeight$.%
\footnote{Rank profiles are closely related to the concept known in the literature as \emph{weak integer compositions}, where an integer is expressed as the sum of non-negative integers. In contrast to weak integer compositions, rank profiles have the additional constraint that each part (rank weight) does not exceed a fixed upper bound, which in our case is $\maxShotWeight$.}
\end{definition}

In \cite{puchingerGenericDecodingSumRank2022}, an upper bound on the cardinality of the set $\wdecomp{\errWeight,\shots,\maxShotWeight}$ has been derived as
\begin{equation} \label{eq:wdecomp_bound}
    |\wdecomp{\errWeight,\shots,\maxShotWeight}| \leq \binom{\shots+\errWeight-1}{\shots-1}.
\end{equation}

We now define the set that contains only the ordered rank profiles.
\begin{definition}
Let $\errWeight$, $\shots$, and $\maxShotWeight$ be non-negative integers with $\errWeight \leq \shots \maxShotWeight$. We define the set
\begin{equation}
    \intparts{\errWeight,\shots,\maxShotWeight} \defeq \left\{ \errWeightVec \in \wdecomp{\errWeight,\shots,\maxShotWeight} \st \errWeight_1 \geq \errWeight_2 \geq \dots \geq \errWeight_\shots \right\}.
\end{equation}  
Without the restrictions on the maximal length $\shots$ and the maximum value $\maxShotWeight$, this set would be identical to the set of integer partitions of $\errWeight$. However, with these restrictions, it forms a special case of integer partitions, constrained in length by $\shots$ and in value by $\maxShotWeight$. Note that every element $\w \in \intparts{\errWeight,\shots,\maxShotWeight}$ is also an element in $\wdecomp{\errWeight,\shots,\maxShotWeight}$, but not vice versa.
\end{definition}

Next, for a given $\w \in \wdecomp{\errWeight,\shots,\maxShotWeight}$, we denote by $\SG{\w}$ the set of all possible permutations of $\w$. Hence, each element $\sigma$ in $\SG{\w}$ is an automorphism, i.e., $\sigma : \wdecomp{\errWeight,\shots,\maxShotWeight} \to \wdecomp{\errWeight,\shots,\maxShotWeight}$. 

Additionally, let $\tilde{\w} = [\tilde{w}_1,\ldots,\tilde{w}_\shots]$ be the vector obtained by sorting the entries of $\w$ in decreasing order. We denote this operation by $\ord(\w)$, so that $\tilde{\w} = \ord(\w)$. In other words, $\ord(\w)$ represents the permutation of the entries of $\w$ that yields the sorted vector $\tilde{\w}$, such that
\begin{equation}
    \tilde{w}_1 \geq \dots \geq \tilde{w}_\shots.
\end{equation}

This leads to the following two relations between the two sets $\intparts{\errWeight,\shots,\maxShotWeight}$ and $\intparts{\errWeight,\shots,\maxShotWeight}$:
\begin{align}
\intparts{\errWeight,\shots,\maxShotWeight} &= \{ \ord(\w) \st \forall \w\in\wdecomp{\errWeight,\shots,\maxShotWeight} \}, \\
    \wdecomp{\errWeight,\shots,\maxShotWeight} &= \bigcup\limits_{\w\in\intparts{\errWeight,\shots,\maxShotWeight}}\{ \sigma(\w) \st \forall \sigma \in\SG{\w}  \}.
\end{align}

For a given rank profile $\w\in\wdecomp{\errWeight,\shots,\maxShotWeight}$ we also define the restricted set
\begin{equation}
    \wdecomp{\errWeight,\shots,\maxShotWeight}^{\geq\w} \defeq \{ \w' \in \wdecomp{\errWeight,\shots,\maxShotWeight} \st w_1' \geq w_1, \dots, w_\shots' \geq w_\shots \}.
\end{equation}

For an element $\errWeightVec\in\intparts{\errWeight,\shots,\maxShotWeight}$ denote by $\PermCount{\w}$ the number of possible permutations of $\errWeightVec$ which is given by the multinomial as
\begin{equation}
    \PermCount{\w} = \binom{\shots}{\lambda_1,\lambda_2,\ldots,\lambda_\maxShotWeight} = \frac{\shots!}{\lambda_1 ! \lambda_2 ! \cdots \lambda_\maxShotWeight !},
\end{equation}
where $\lambda_i$ is the number of occurrences of the integer $i$ in $\errWeightVec$ for all $i=0,\ldots,\maxShotWeight$. Hence, from all permutations from all elements in $\intparts{\errWeight,\shots,\maxShotWeight}$ we obtain the set $\wdecomp{\errWeight,\shots,\maxShotWeight}$.

The set of all vectors in $\Fqm^n$ of sum-rank weight $\SumRankWeight(\e)=\errWeight$ is denoted as
\begin{equation}\label{eq:error set definition}
    \errorSet_{q,\shotLength,m,\shots}(\errWeight) \defeq \{ \e \in \Fqm^n \mid \SumRankWeight(\e) = \errWeight \}.
\end{equation}
The cardinality of this set is given by (see~\cite{puchingerGenericDecodingSumRank2022})
\begin{equation}
    |\errorSet_{q,\shotLength,m,\shots}(\errWeight)| = \sum_{\errWeightVec\in\wdecomp{\errWeight,\shots,\maxShotWeight}} \prod_{i=0}^{\shots-1} \NMq{q}{m,\shotLength,\errWeight_i},
    \label{eq:error_set_cardinality}
\end{equation}
where $\errWeightVec=[\errWeight_0,\ldots,\errWeight_{\shots-1}]$ denotes the decomposition of the sum-rank weight $\errWeight$ into $\shots$ non-negative integers $\errWeight_i$ such that $\sum_{i=0}^{\shots-1} \errWeight_i = \errWeight$ and $\errWeight_i \leq \maxShotWeight$ for all $i$. The set of all such decompositions is denoted by $\wdecomp{\errWeight,\shots,\maxShotWeight}$.
The cardinality expression in \eqref{eq:error_set_cardinality} can be efficiently computed using a dynamic programming approach, as described in~\cite{puchingerGenericDecodingSumRank2022}.

The term $\NMq{q}{m,\shotLength,\errWeight_i}$ represents the number of matrices of size $m \times \shotLength$ of rank $\errWeight_i$ over the finite field $\Fq$. It can be calculated using the following formula \cite[Chapter 13]{macwilliamsTheoryErrorCorrectingCodes1977}\cite{miglerWeightRankMatrices2004}
\begin{align}\label{eq:def-cardinality-of-matrix-set}
    \NMq{q}{m,\shotLength,\errWeight_i} &= \prod_{j=0}^{\errWeight_i-1} \frac{(q^m-q^j)(q^\shotLength-q^j)}{q^{\errWeight_i}-q^j} \\
    &= \quadbinomq{\shotLength}{\errWeight_i}\prod_{j=0}^{\errWeight_i-1} (q^m-q^j)\\
    &=  \quadbinomq{m}{\errWeight_i}\prod_{j=0}^{\errWeight_i-1} (q^\shotLength-q^j),
\end{align}
where $\quadbinomq{a}{b}$ denotes the $q$-binomial coefficient, also known as the Gaussian binomial coefficient, which is defined as
\begin{equation}
    \quadbinomq{a}{b} \defeq \prod_{i=1}^{b}\frac{q^{a-b+i} - 1}{q^i - 1} = \prod_{i=0}^{b-1} \frac{(q^{a-i} - 1)}{ (q^{b-i} - 1)}.
\end{equation}

Here, $a$ and $b$ are integers with $a \geq b \geq 0$, and $q$ is a prime power.

The Gaussian binomial coefficient satisfies the following inequality~\cite{koetterCodingErrorsErasures2008}
\begin{equation}
    q^{(a-b)b} \leq \quadbinomq{a}{b} \leq \gamma_q q^{(a-b)b},
\end{equation}
with $\gamma_q$ defined as
\begin{equation}\label{eq:def-qconstant}
    \gamma_q \defeq \prod_{i=1}^{\infty}(1-q^{-i})^{-1}.
\end{equation}

In the following, we present expressions for the probability that two subspaces intersect in a fixed-dimensional space and the probability of one subspace being a subspace of another.

Let $\myspace{A}$ and $\myspace{B}$ be two subspaces of $\Fq^{\mu}$ with dimensions $a$ and $b$, respectively.

We define the conditional probability $\intersectprob{\mu}{a, b}(j)$ as the probability that the intersection of $\myspace{A}$ and $\myspace{B}$ has dimension exactly $j$, given their dimensions $a$ and $b$. This probability is given by (see~\cite{rennerRandomizedDecodingGabidulin2020})
\begin{equation}\label{eq:notation intersect prob}
    \intersectprob{\mu}{a, b}(j) \defeq \Pr[\dim(\myspace{A} \cap \myspace{B}) = j \mid a, b] = \frac{\quadbinomq{\mu-a}{b - j}\quadbinomq{a}{j} q^{(a-j)(b-j)}}{\quadbinomq{\mu}{b}}.
\end{equation}

Next, we define the probability that $\myspace{A}$ is a subspace of $\myspace{B}$, denoted by $\subspaceprob{\mu}{a}{b}$, where $a \leq b$. This probability is given by~\cite{koetterCodingErrorsErasures2008}
\begin{equation}\label{eq:notation subspace prob}
    \subspaceprob{\mu}{a}{b} \defeq \Pr[\myspace{A} \subseteq \myspace{B} \mid a, b] = \frac{\quadbinomq{b}{a}}{\quadbinomq{\mu}{a}}.
\end{equation}

\begin{remark}
The probabilities $\intersectprob{\mu}{a, b}(j)$ and $\subspaceprob{\mu}{a}{b}$ hold whether $\myspace{A}$ and $\myspace{B}$ are both drawn uniformly at random from $\Grassm{q}{\Fq^\mu}$, or one of them is fixed and the other is drawn uniformly at random from $\Grassm{q}{\Fq^\mu}$.
\end{remark}

\begin{remark}
The probability $\subspaceprob{\mu}{a}{b}$ is equal to the intersection probability $\intersectprob{\mu}{a, b}(b)$ (or $\intersectprob{\mu}{a, b}(a)$) when the dimension of the intersection is equal to the dimension of the smaller subspace $\myspace{A}$, i.e.,
\begin{align}
    \subspaceprob{\mu}{a}{b} &= \intersectprob{\mu}{a, b}(b) = \frac{\quadbinomq{\mu-a}{b - a}}{\quadbinomq{\mu}{b}} = \frac{\quadbinomq{b}{a}}{\quadbinomq{\mu}{a}} &  \text{ if } a \leq b,\\
    \subspaceprob{\mu}{b}{a} &= \intersectprob{\mu}{a, b}(a) = \frac{\quadbinomq{a}{b}}{\quadbinomq{\mu}{b}} & \text{ if } b \leq a.
\end{align}
This relationship holds because $\myspace{A}$ is a subspace of $\myspace{B}$ if and only if the dimension of their intersection is the same as that of $\myspace{A}$.
\end{remark}

\subsection{Linear Codes}
A linear code over $\Fqm$ with dimension $k$ and length $n$ is a $k$-dimensional subspace of $\Fqm^n$, denoted as $\LinearCode{\n,k}{\Fqm}$.
The minimum sum-rank distance of $\LinearCode{\n,k}{\Fqm}$ w.r.t to $\n$ is defined as
\begin{equation}
    \dmin = \min_{\substack{\c_1,\c_2\in\LinearCode{\n,k}{\Fqm} \\ \c_1\neq \c_2}} \left\{ \SumRankDistN{(\c_1,\c_2)} \right\}.
\end{equation}
When the minimum distance $\dmin$ of a linear code $\LinearCode{\n,k}{\Fqm}$ is known, we refer to $\LinearCode{\n,k}{\Fqm}$ as an $\LinearCode{n,k,\dmin}{\Fqm}$ code.
We define the \emph{unique decoding radius} $\UniqueDecRad$ of a code with minimum distance $\dmin$ as
\begin{equation*}
    \UniqueDecRad \defeq \left\lfloor \frac{\dmin - 1}{2} \right\rfloor.
\end{equation*}
This is the maximum weight of errors (in the respective metric) that can be uniquely decoded by the code.
In this paper, we concentrate on codes with distance properties in the sum-rank metric. To emphasize that the sum-rank distance is computed with respect to the length profile $\n = [n_1, n_2, \dots, n_\shots]$, we denote such codes as $\LinearCode{\n,k}{\Fqm}$. In the special case where all blocks have the same length $\shotLength$, resulting in a total length of $n=\shotLength\shots$, we use the notation $\LinearCode{n,k}{\Fqm}$.

A matrix $\G \in \Fqm^{k \times n}$ is called a generator matrix of the code $\LinearCode{\n,k}{\Fqm}$ if and only if the rows of $\G$ form a basis for $\LinearCode{n,k}{\Fqm}$. In other words, the linear combinations of the rows of $\G$ generate all the codewords in $\LinearCode{n,k}{\Fqm}$.

Additionally, a matrix $\H \in \Fqm^{(n-k) \times n}$ is called a parity-check matrix of $\LinearCode{n,k}{\Fqm}$ if its rows form a basis for the right kernel (i.e., the null space) of the generator matrix $\G$.

\subsection{Channel Model}\label{sec:pre:channel model}
Consider a codeword $\c \in \LinearCode{n,k}{\Fqm}$ that is corrupted by an error $\e$ of sum-rank weight $\errWeight$. The received word $\y$ is then given by
\begin{equation}
\y = \c + \e.
\end{equation}
The error vector $\e \in \Fqm^n$ with sum-rank weight $\SumRankWeight(\e) = w$ can be partitioned into blocks $\e^{(1)}, \dots, \e^{(\shots)}$ according to the length profile $\n$. 
This partitioned vector can be decomposed as
\begin{equation}\label{eq:errorDecomp}
\e = \a\B = \left[\a^{(1)}\mid\dots\mid\a^{(\shots)}\right] \cdot \diag(\B^{(1)}, \ldots,\B^{(\shots)}),
\end{equation}
where $\a^{(i)}\in\Fqm^{w_i}$, $\B^{(i)}\in\Fq^{w_i \times n_i}$ with $\rkq(\a^{(i)})=\rkq(\B^{(i)})=w_i$ for all $i\in\{1,\ldots,\shots\}$, and $w = \sum_{i=1}^{\shots} w_i$.
It follows that $\e^{(i)} = \a^{(i)}\B^{(i)}$ for $i\in\{1,\ldots,\shots\}$, where the entries of $\a^{(i)}$ form a basis over $\Fq$ of the column space of $\e^{(i)}$ and the rows of $\B^{(i)}$ form a basis over $\Fq$ of its row space.

\subsection{Error Support}
In the following, we define the notions of row support and column support for vectors in the sum-rank metric.
\begin{definition}[Row and Column Support]\label{def:row-column-sum-rank-support}
Let $\e\in\Fqm^n$ be of sum-rank weight $\errWeight$. 
\begin{itemize}
    \item \textbf{Row Support:} The row support of $\e$ is defined as
    \begin{equation}
        \RSv \defeq \RS^{(1)} \times \RS^{(2)} \times \cdots \times \RS^{(\shots)},
    \end{equation}
    where $\RS^{(i)} \subseteq \Fq^{n_i}$ is the $\Fq$-row space of $\B^{(i)}\in\Fq^{\errWeight_i \times n_i}$ and thus of $\e^{(i)}$ as in~\eqref{eq:errorDecomp} for all $i\in\{1,\ldots,\shots\}$.
    
    \item \textbf{Column Support:} The column support of $\e$ is defined as
    \begin{equation}
        \CSv \defeq \CS^{(1)} \times \CS^{(2)} \times \cdots \times \CS^{(\shots)},
    \end{equation}
    where $\CS^{(i)} \subseteq \Fq^{m}$ is the $\Fq$-column space of $\e^{(i)}$. Specifically, each entry of $\e^{(i)} \in \Fqm$ is expanded as a vector in $\Fq^m$ using a fixed basis of $\Fqm$ over $\Fq$. The column support is then defined as the $\Fq$-span of these expanded vectors, for all $i\in\{1,\ldots,\shots\}$.
\end{itemize}    
\end{definition}

Assume $\ESv$ to be either a row or column support of the error as defined in Definition~\ref{def:row-column-sum-rank-support}.
We denote by $\sumDim(\ESv)$ the \emph{sum dimension} of an error support: %
\begin{equation}
    \sumDim(\ESv) \defeq \sum_{i=1}^{\shots} \dim(\ES^{(i)}).
\end{equation}
The intersection of two supports $\ES_1$ and $\ES_2$ is defined as
\begin{equation}
    \ESv_1 \cap \ESv_2 \defeq \left( \ES_1^{(1)} \cap \ES_2^{(1)} \right) \times \cdots \times \left( \ES_1^{(\shots)} \cap \ES_2^{(\shots)} \right).
\end{equation}

Given two supports $\ESv_1$ and $\ESv_2$, we say that $\ESv_2$ is a \emph{row super-support} of $\ESv_1$, denoted by $\ESv_1 \subseteq \ESv_2$, if and only if $\ES_1^{(i)} \subseteq \ES_2^{(i)}$ for all $i\in\{1,\ldots,\shots\}$.

\begin{definition}
    Let $\mu$ be a positive integer and $0\leq w \leq \shots\mu$. For $\w\in\wdecomp{w,\shots, \mu}$, we define the set of all supports as
    \begin{equation}
        \Xi_{q,\mu}(\w) \defeq \left\{ \mathcal{F}_1 \times \cdots \times \mathcal{F}_{\shots} : \mathcal{F}_i \subseteq \Fq^\mu \st \dim(\mathcal{F}_i)=w_i \right\}.
    \end{equation}
\end{definition}

\subsection{Linearized Reed--Solomon Codes}
\Ac{LRS} codes, introduced by Martínez-Peñas in~\cite{martinez-penasSkewLinearizedReed2018}, are a class of codes in the sum-rank metric. An \ac{LRS} code of length $n$, dimension $k$, and length partition $\n = \left[n_1,\ldots,n_\shots\right]$ over $\Fqm$ is denoted by $\LRSCode[\n, k]$.

The key properties of \ac{LRS} codes are:
\begin{itemize}
    \item \ac{LRS} codes achieve the Singleton-like bound in the sum-rank metric, i.e., their minimum sum-rank distance is $\dmin = n - k + 1$~\cite{martinez-penasSkewLinearizedReed2018, byrneFundamentalPropertiesSumRankMetric2021}. Therefore, \ac{LRS} codes are called \ac{MSRD} codes.
    
    \item Consequently, \ac{LRS} codes can uniquely decode errors of sum-rank weight up to
    \begin{equation*}
        \UniqueDecRad = \left\lfloor \frac{n - k}{2} \right\rfloor.
    \end{equation*}
    
    \item \ac{LRS} codes have some restrictions on their parameters: the number of blocks $\shots$ must satisfy $\shots \leq q - 1$, and the length of each block $n_i$ must satisfy $n_i \leq m$ for all $i \in \{1,\dots,\shots\}$ (see~\cite{martinez-penasSkewLinearizedReed2018}).
\end{itemize}

\section{Problem Description}\label{sec:problem-description}

This section presents and categorizes decoding problems in the sum-rank metric, relevant to various coding theory and cryptography applications. Starting with the most generic problem, we introduce progressively specific scenarios.

\subsection{Sum-Rank Syndrome Decoding Problem}\label{subsec:syndrome-decoding}

The \emph{sum-rank syndrome decoding problem} generalizes both the syndrome decoding problem in the Hamming metric and the rank syndrome decoding problem, thereby encompassing a wide range of cryptosystems such as WAVE~\cite{debris-alazardWaveNewFamily2019}, BIKE~\cite{aragon2020bike}, HQC~\cite{melchor2020hqc} (which are based on the Hamming metric), and RQC~\cite{melchorRankQuasicyclicRQC2020} (which is based on the rank metric).

\begin{problem}[Sum-Rank Syndrome Decoding Problem]\label{prob:SRSDecProblemUniformSyndrome}
\hfill
\begin{itemize}
    \item \textbf{Instance:}
    \begin{itemize}
        \item A linear sum-rank metric code $\LinearCode{\n, k}{\Fqm} \subseteq \Fqm^n$ with parity-check matrix $\H \in \Fqm^{(n-k) \times n}$.
        \item A syndrome $\s \in \Fqm^{n-k}$ and an integer $\decRadius > 0$.
    \end{itemize}
    \item \textbf{Objective:} Find an error vector $\e$ such that $\s = \e\H^\top$ and $\SumRankWeight(\e) = \decRadius$.
\end{itemize}
\end{problem}

\noindent The sum-rank syndrome decoding problem is particularly interesting due to its applicability across different metrics and cryptosystems. A solution to this problem is not guaranteed.

The algorithm introduced in~\cite{puchingerGenericDecodingSumRank2022} provides a generic decoding approach that addresses the sum-rank syndrome decoding problem for error weights up to $n - k$. 

\subsection{Decoding Beyond the Unique Radius}\label{subsec:beyond-unique-decoding}

The following problem can be seen as a special case of the sum-rank syndrome decoding problem (see Problem~\ref{prob:SRSDecProblemUniformSyndrome}). In this case, we assume that a specific codeword, corrupted by an error of known weight, has been received.
Under this assumption, the weight of the error, denoted by $\errWeight$, is known, and we set the decoding radius accordingly to $\errWeight$. This ensures that at least one solution exists, although it may not be unique, similar to the general syndrome decoding problem.

We first consider the problem of decoding beyond the unique radius for arbitrary sum-rank metric codes.
\begin{problem}[Beyond Unique Decoding for Sum-Rank Metric Codes]\label{prob:SRSDecProblemY}
\hfill
\begin{itemize}
    \item \textbf{Instance:} 
    \begin{itemize}
        \item Linear sum-rank metric code $\LinearCode{\n, k}{\Fqm} \subseteq \Fqm^n$ with unique decoding radius $\UniqueDecRad$.
        \item Error vector $\e \sample \errorSet_\errWeight$ with $\SumRankWeight(\e)=\errWeight \geq \UniqueDecRad$.
        \item Received vector $\y = \c + \e$ with $\c\in\LinearCode{\n, k}{\Fqm}$.
    \end{itemize}
    \item \textbf{Objective:} Find a codeword $\c \in \LinearCode{\n, k}{\Fqm}$, such that $\SumRankWeight(\y-\c) = \errWeight$.
\end{itemize}
\end{problem}

Problem~\ref{prob:SRSDecProblemY} extends the unique decoding problem by allowing error weights that exceed the unique decoding radius. Unlike unique decoding, multiple codewords may satisfy the decoding condition, and a solution is not guaranteed to be unique. The generic decoder introduced in~\cite{puchingerGenericDecodingSumRank2022} can efficiently address this problem for any linear sum-rank metric code without relying on its structure, handling error weights up to $n - k$.

Next, we specialize this problem for \ac{LRS} codes.
\begin{problem}[Beyond Unique Decoding for \ac{LRS} codes]\label{prob:LRSDecProblem}
\hfill
\begin{itemize}
    \item \textbf{Instance:}
    \begin{itemize}
        \item \Ac{LRS} code $\LRSCode[\n, k] \subseteq \Fqm^n$, $\y\in\Fqm^n$.
        \item Error vector $\e \sample \errorSet_\errWeight$ with $\SumRankWeight(\e)=\errWeight \geq \UniqueDecRad$.
        \item Received vector $\y=\c+\e\in\Fqm^{n}$ with $\c\in\LRSCode$.
    \end{itemize}
    \item \textbf{Objective:} Find a codeword $\c \in \LRSCode$, such that $\SumRankWeight(\y-\c) = \errWeight$.
\end{itemize}
\end{problem}

Problem~\ref{prob:LRSDecProblem} is a specialization of Problem~\ref{prob:SRSDecProblemY} for \ac{LRS} codes. This specialization allows us to exploit the inherent structure of \ac{LRS} codes, potentially leading to more efficient decoding algorithms. We compare the complexity of the generic decoder for Problem~\ref{prob:SRSDecProblemY}, as introduced in~\cite{puchingerGenericDecodingSumRank2022}, with our proposed randomized decoder (see Section~\ref{sec:random_decoding}) tailored for \ac{LRS} codes.

The \ac{FL} system is a rank-metric code-based cryptosystem that uses Gabidulin codes, which are a special case of \ac{LRS} codes. Problem~\ref{prob:LRSDecProblem} is itself a generalization of the decoding problem considered in~\cite{rennerRandomizedDecodingGabidulin2020}, which focused on Gabidulin codes. Our proposed randomized decoder generalizes the decoder introduced in~\cite{rennerRandomizedDecodingGabidulin2020} to the sum-rank metric.
We first introduced this generalization in~\cite{jerkovitsRandomizedDecodingLinearized2023}.

Our proposed decoder and complexity analysis for the sum-rank metric could be valuable for future cryptosystems that are similar to the \ac{FL} system but operate in the sum-rank metric, or for different cryptosystems that rely on problems like Problem~\ref{prob:LRSDecProblem} in the sum-rank metric.

\subsection{Unique Decoding Problem}\label{subsec:unique-decoding}

The \emph{unique decoding problem} is a specific case of Problem~\ref{prob:SRSDecProblemY}, where the error weight does not exceed the unique decoding radius, ensuring that there is exactly one solution within this radius.

\begin{problem}[Unique Decoding Problem]\label{prob:uniqueDecoding}
\hfill
\begin{itemize}
    \item \textbf{Instance:} 
    \begin{itemize}
        \item A linear sum-rank metric code $\LinearCode{\n, k}{\Fqm} \subseteq \Fqm^n$ with unique decoding radius $\UniqueDecRad$.
        \item An error vector $\e \sample \errorSet_\errWeight$ with $\errWeight = \SumRankWeight(\e) \leq \UniqueDecRad$.
        \item A received vector $\y = \c + \e$ where $\c\in\LinearCode{\n, k}{\Fqm}$.
    \end{itemize}
    \item \textbf{Objective:} Determine the unique codeword $\c \in \LinearCode{\n, k}{\Fqm}$ such that $\SumRankWeight(\y-\c) \leq \UniqueDecRad$.
\end{itemize}
\end{problem}

In this case, the decoder is guaranteed to find exactly one solution as long as the error weight is within the unique decoding radius. Efficient polynomial-time decoders are available for several well-known algebraic codes, such as Reed--Solomon codes, \ac{BCH} codes, and Goppa codes in the Hamming metric~\cite{macwilliamsTheoryErrorCorrectingCodes1977}, as well as Gabidulin codes in the rank metric~\cite{gabidulinTheoryCodesMaximum1985}, and \ac{LRS} codes in the sum-rank metric~\cite{martinez-penasSkewLinearizedReed2018}, provided that the code structure is known.

However, in cryptosystems such as McEliece-like cryptosystems, where the code structure is intentionally obfuscated to increase security, decoding becomes much more challenging for unauthorized parties. This further highlights the importance of efficient decoders, particularly in cases where the error weight stays within the unique decoding radius.

\subsection{Channel Model}\label{subsec:channel-model}

In Problem~\ref{prob:SRSDecProblemY}, Problem~\ref{prob:LRSDecProblem} and Problem~\ref{prob:uniqueDecoding}, the error vector $\e$ is drawn uniformly at random from $\errorSet_\errWeight$, the set of all vectors in $\Fqm^n$ with sum-rank weight $\errWeight$. This is done by first determining a rank profile $\w\in\wdecomp{\errWeight,\shots,\maxShotWeight}$ according to the distribution
\begin{equation}\label{eq:probability of errWeightVec for uniform error}
\Pr[\w] = \frac{1}{|\errorSet_{q,\shotLength,m,\shots}(\errWeight)|} \prod_{i=1}^{\shots} \NMq{q}{m,\shotLength,\errWeight_i}.
\end{equation}
Then, for each $i\inrange{1}{\shots}$, an element $\e^{(i)}\in\Fqm^{\shotLength_i}$ of rank weight $\errWeight_i$ is drawn independently and uniformly at random from the set of all elements in $\Fqm^{\shotLength_i}$ of rank weight $\errWeight_i$, as described in~\cite{puchingerGenericDecodingSumRank2022}. The resulting error vector $\e$ satisfies $\SumRankWeight(\e)=\errWeight$ and is uniformly distributed in $\errorSet_\errWeight$. 
The received word $\y$ is then considered to be of the form $\y = \c + \e$ with $\c \in \LinearCode{\n, k}{\Fqm} \subseteq \Fqm^n$.

\section{Generic Decoding in the Sum-Rank Metric}\label{sec:generic-decoding}

In this section, we summarize the generic decoding algorithm for the sum-rank metric introduced by Puchinger et al. \cite{puchingerGenericDecodingSumRank2020, puchingerGenericDecodingSumRank2022}. Their analysis focuses on solving a special instance of Problem~\ref{prob:SRSDecProblemUniformSyndrome}, where the syndrome is chosen such that at least one solution exists. This scenario applies to both Problem~\ref{prob:SRSDecProblemY} and Problem~\ref{prob:uniqueDecoding}. 
In their derivation of the lower bound, they assume that at most one solution exists within the decoding radius, with no alternative solutions possible. As a result, the lower bound applies only to Problem~\ref{prob:uniqueDecoding}.

The decoding algorithm proceeds by guessing possible error supports according to a probability distribution, which is a key design criterion that can be optimized to maximize the decoder's success probability. In each iteration of the decoding loop, a new support is guessed, and the decoder attempts to correct the error based on that support. If the guess is incorrect, the loop continues with another guess. The success probability for each iteration depends on the support-guessing distribution. The worst-case complexity is determined by evaluating all possible error rank profiles, with the worst-case profile, denoted by $\wcRP$, representing the maximum number of operations needed for the decoder to succeed.

\begin{remark}
In~\cite[Remark 17]{puchingerGenericDecodingSumRank2022}, it was shown that an error $\e$ with row support $\RSv$ and column support $\CSv$ can be uniquely recovered if either a row super support $\FSv_{R} \supseteq \RSv$ or a column super support $\FSv_{C} \supseteq \CSv$ is found, both with sum-rank weight $\supwt$, such that $\errWeight \leq \supwt < \dmin$. The sum-rank weight $\supwt$ cannot exceed
\begin{equation}\label{eq:smax}
\supwt_{\max} \defeq \min\left\{n-k, \left\lfloor\frac{m}{\shotLength}(n-k)\right\rfloor\right\}.
\end{equation}

Decoding beyond the unique erasure decoding radius (i.e. error weights larger than $\dmin-1$) is possible up to $\supwt_{\max}$. This situation is analogous to the classical \ac{ISD} algorithm in the Hamming metric, where a support of size $n-k$ is selected. Successful decoding occurs if the corresponding submatrix of the parity-check matrix formed by these columns is of full rank, ensuring a unique solution to the system of linear equations. 
In the Hamming case, the probability that the submatrix is full rank is typically assumed to be close to 1, which is especially true for large field sizes.
However, as shown in \cite{jerkovitsErrorcodePerspectiveMetzner2024}, for the sum-rank metric, the probability of finding a valid decoding support under these conditions can be significantly lower, depending on parameter choices. Within this paper, we similarly assume this probability to be close to 1 in our analysis, though for certain parameter sets, it might be advisable to check whether this assumption holds in practice.
\end{remark}

Algorithm~\ref{alg:genSRdecoder_puchinger} describes the decoding process, which selects the appropriate type of (row or column) support. When $m$ is smaller than $\shotLength$, the algorithm selects the row support; otherwise, it chooses the column support. For the sake of simplicity, we will henceforth refer to the selected support type as simply ``support'' throughout the remainder of this paper, with the understanding that the decoding algorithm makes this choice based on the given parameters.

\begin{algorithm}\label{alg:genSRdecoder_puchinger}
\caption{Generic Sum-Rank Decoder~\cite{puchingerGenericDecodingSumRank2022}}\label{alg:generic_sr_decoder}

\Input{Parameters: $q$, $m$, $n$, $k$, $\shots$, $\errWeight$ and $\supwt$ with $\errWeight \leq \supwt \leq \smax$ \\
       Received vector $\y\in\Fqm^n$\\
       Parity-check matrix $\H\in\Fqm^{(n-k)\times n}$ of an $\Fqm$-linear sum-rank \\
       metric code $\mycode{C}$\\    
}
\Output{
    Vector $\c' \in \mycode{C}$ such that $\SumRankWeight(\y-\c') = \errWeight$
}
$\e' \gets \bm{0}$ \;
$\shotLength \gets n/\shots$ \;
$\mu \gets \min\{m,\shotLength\}$ \;
\While{$\H(\r-\e')^\top \neq \bm{0}$ \oror $\SumRankWeight(\e') \neq \errWeight$}{
    $\FSv \gets $ Draw random support $\FSv\subseteq \Fq^\maxShotWeight \times \dots \times \Fq^\maxShotWeight$ of sum dimension $\supwt$\label{alg:generic_sr_decoder:drawsupport} \;
    \If{$\shotLength < m$}{
        $\e' \gets$ Column erasure decoding w.r.t. $\FSv$, $\H$, $\y$~(cf.~\cite[Theorem 13]{puchingerGenericDecodingSumRank2022})
    }
    \Else{
        $\e' \gets$ Row erasure decoding w.r.t. $\FSv$, $\H$, $\y$~(cf.~\cite[Theorem 14]{puchingerGenericDecodingSumRank2022})
    }
}
\Return $\y-\e'$
\end{algorithm}

\subsection{Improved Simple Bound on the Worst-Case Success Probability}

Recall that we denoted the worst-case rank profile as $\wcRP$. In~\cite[Theorem 16]{puchingerGenericDecodingSumRank2022}, the authors derive lower and upper bounds on the expected run time $\Wgen$ of Algorithm~\ref{alg:generic_sr_decoder} to decode an error pattern with rank profile $\wcRP$.

Given an error rank profile $\errWeightVec \in \wdecomp{\errWeight,\shots,\maxShotWeight}$ and a super-support rank profile $\v \in \wdecomp{\supwt,\shots,\maxShotWeight}$, we compute the probability $\probProf{q,\maxShotWeight}{\v, \errWeightVec}$ that the error support lies within the guessed super support as~\cite{puchingerGenericDecodingSumRank2022}
\begin{equation}\label{eq:success probability for generic given profile of v and w}
    \probProf{q,\maxShotWeight}{\v, \errWeightVec} \defeq \prod_{i=1}^{\shots}\subspaceprob{\maxShotWeight}{\errWeightVec_i}{\v_i},
\end{equation}
where $\subspaceprob{\maxShotWeight}{\errWeightVec_i}{\v_i}$ is the probability as in~\eqref{eq:notation subspace prob}.

Using this probability, we define $\probProfmax{q,\maxShotWeight,\supwt}{\errWeightVec}$ as the maximum probability over all super-support rank profiles $\v \in \wdecomp{\supwt,\shots,\maxShotWeight}$
\begin{equation}\label{eq:probProfmax}
    \probProfmax{q,\maxShotWeight,\supwt}{\errWeightVec} \defeq \max_{\v\in\wdecomp{\supwt,\shots,\maxShotWeight}} \probProf{q,\maxShotWeight}{\v, \errWeightVec}.
\end{equation}
Further let $\scomp_{\maxShotWeight}(\errWeightVec, \supwt)$ be a function $\scomp_{\maxShotWeight} : \wdecomp{\errWeight,\shots,\maxShotWeight} \times \NN \to \wdecomp{\supwt,\shots,\maxShotWeight}$ that for a given integer $\supwt\in\NN$ returns a rank profile such that
\begin{equation}
    \probProfmax{q,\maxShotWeight,\supwt}{\errWeightVec} = \probProf{q,\maxShotWeight}{\scomp_{\maxShotWeight}(\errWeightVec, \supwt), \errWeightVec}.
\end{equation}
The function $\scomp_{\maxShotWeight}(\errWeightVec, \supwt)$ can be implemented efficiently, see~\cite{puchingerGenericDecodingSumRank2022}.

With these definitions in place, Puchinger et al.~\cite[Theorem 16]{puchingerGenericDecodingSumRank2022} provide the following bounds on the expected runtime $\Wgen$ of Algorithm~\ref{alg:generic_sr_decoder} for Problem~\ref{prob:uniqueDecoding}
\begin{align}
    \WgenLB &= \WiterGen {|\wdecomp{\errWeight,\shots,\maxShotWeight}|}^{-1} Q_{\errWeight,\shots,\maxShotWeight} \label{eq:WgenLB}, \\
    \WgenUB &= \WiterGen Q_{\errWeight,\shots,\maxShotWeight} \label{eq:WgenUB}, \\
    \WgenUBsimple &= \WiterGen \binom{\shots+\errWeight-1}{\shots-1} \gamma_q^\shots q^{\errWeight(\maxShotWeight-\frac{\supwt}{\shots})} \label{eq:WgenUBsimple},
\end{align}
with
\begin{equation}\label{eq:def of Q}
     Q_{\errWeight,\shots,\maxShotWeight} \defeq \sum_{\errWeightVec\in\wdecomp{\errWeight,\shots,\maxShotWeight}} {\probProfmax{q,\maxShotWeight,\supwt}{\errWeightVec}}^{-1},
\end{equation}
where $\WiterGen$ represents the complexity of one iteration of the generic decoding algorithm.
To be more precise, $\WiterGen$ is the sum of the complexities of two main components: drawing an error super support and performing the row/column erasure decoding, as outlined in Algorithm~\ref{alg:generic_sr_decoder}. The complexity of drawing an error super support is in the order of $\softoh{n^3 m^2 \log_2(q)}$ bit operations. Meanwhile, the row/column-erasure decoding involves $\oh{(n-k)^3 m^3}$ operations over $\mathbb{F}_q$. 
While these complexities are asymptotic approximations and neglect constant factors, they provide useful estimates for large input sizes. For finite lengths, the exact complexities depend on the specific implementation details of the underlying algorithms. Therefore, for the purpose of plotting and practical considerations, we approximate $\WiterGen \approx n^3 m^3$, combining the dominant terms of both components.

The previous upper bound $\WgenUBsimple$ on the expected runtime of Algorithm~\ref{alg:generic_sr_decoder} can be rather loose when the parameters are closer to the Hamming metric, i.e., when $\shots \to n$ for fixed $\shotLength$ and/or $\shotLength \to 1$ (see Figure~\ref{fig:genNewSimple1}). To address this, we introduce a new, tighter upper bound in Theorem~\ref{thm:improvedUB}.

\begin{theorem}\label{thm:improvedUB}
    Let $\c$ be a codeword of a sum-rank-metric code $\LinearCode{\n, k}{\Fqm}$ with minimum sum-rank distance $\dmin$. Additionally, let $\e$ be an error of sum-rank weight $\errWeight < \dmin$ with a rank profile corresponding to the worst-case rank profile $\wcRP$. Then Algorithm~\ref{alg:generic_sr_decoder} in the context of Problem~\ref{prob:SRSDecProblemUniformSyndrome} returns a solution w.r.t. an error $\e' \in \Fqm^n$ with the weight $\errWeight$. Each iteration of Algorithm~\ref{alg:generic_sr_decoder} has complexity $\WiterGen$. The overall expected worst-case runtime, also referred to as the complexity $\Wgenwc$ of Algorithm~\ref{alg:generic_sr_decoder} is upper bounded by
    \begin{align}\label{eq:improvedUB}
        \Wgenwc&\leq \WgenUBsimpleImproved,
    \end{align}
    with
    \begin{align}
        \WgenUBsimpleImproved &\defeq \WiterGen \binom{\shots+\errWeight-1}{\shots-1} q^{\errWeight(\maxShotWeight-\frac{\supwt}{\shots})} \cdot \min\left(\gamma_{q}^{\shots},  {\left(\frac{1 - q^{-\maxShotWeight}}{1-q^{-1}}\right)}^\errWeight\right).
    \end{align}
\end{theorem}
\begin{proof}
To prove the theorem, we will show that
\begin{equation}
    \Wgenwc\leq \WiterGen \binom{\shots+\errWeight-1}{\shots-1} q^{\errWeight(\maxShotWeight-\frac{\supwt}{\shots})} {\left(\frac{1 - q^{-\maxShotWeight}}{1-q^{-1}}\right)}^\errWeight.
\end{equation}
Starting from the bound in~\eqref{eq:WgenUB}, it suffices to show that
\begin{equation}
     Q_{\errWeight,\shots,\maxShotWeight} \leq \binom{\shots+\errWeight-1}{\shots-1} q^{\errWeight(\maxShotWeight-\frac{\supwt}{\shots})} {\left(\frac{1 - q^{-\maxShotWeight}}{1-q^{-1}}\right)}^\errWeight.
\end{equation}
By the definition of $Q_{\errWeight,\shots,\maxShotWeight}$ in~\eqref{eq:def of Q}, we can bound $Q_{\errWeight,\shots,\maxShotWeight}$ as 
\begin{equation}
    Q_{\errWeight,\shots,\maxShotWeight} \leq \underbrace{|\wdecomp{\errWeight,\shots,\maxShotWeight}|}_{\mathclap{ \leq \binom{\shots+\errWeight-1}{\shots-1} \text{ (cf.~\cite{puchingerGenericDecodingSumRank2022})}}} \cdot \max_{\errWeightVec\in\wdecomp{\errWeight,\shots,\maxShotWeight}}  {\probProfmax{q,\maxShotWeight,\supwt}{\errWeightVec}}^{-1},  
\end{equation}
where 
\begin{equation}\label{eq:some equation in proof}
    \max_{\errWeightVec\in\wdecomp{\errWeight,\shots,\maxShotWeight}} {\probProfmax{q,\maxShotWeight,\supwt}{\errWeightVec}}^{-1} = \max_{\errWeightVec\in\wdecomp{\errWeight,\shots,\maxShotWeight}} \left\{{\probProf{q,\maxShotWeight}{\supwtVec', \errWeightVec}}^{-1} \st \supwtVec' = \scomp_{\maxShotWeight}(\errWeightVec, \supwt) \right\}.
\end{equation}
Next, we can bound ${\probProf{q,\maxShotWeight}{\supwt', \errWeightVec}}^{-1}$ as
\begin{equation}\label{eq:proof bound step 2}
    {\probProf{q,\maxShotWeight}{\supwtVec', \errWeightVec}}^{-1} = \prod_{i=1}^{\errWeight} \frac{\quadbinomq{\maxShotWeight}{\errWeight_i}}{\quadbinomq{\supwt_i'}{\errWeight_i}} \leq {\left(\frac{1-q^{-\maxShotWeight}}{1-q^{-1}}\right)}^{\errWeight} \cdot \prod_{i=1}^{\errWeight} q^{\errWeight_i(\maxShotWeight-\supwt_i')} ,
\end{equation}
where the inequality follows from
\begin{align}
    \frac{\quadbinomq{a}{b}}{\quadbinomq{c}{b}} = \frac{q^{b(a-b)}}{q^{b(c-b)}}\cdot\prod_{i=1}^{\errWeight}\frac{(1-q^{-a})(1-q^{-a+1})\cdots(1-q^{-a+b-1})}{(1-q^{-c})(1-q^{-c+1})\cdots(1-q^{-c+b-1})} \leq q^{b(a-c)} \cdot \frac{1-q^{-a}}{1-q^{-1}}.
\end{align}
Substituting~\eqref{eq:proof bound step 2} into~\eqref{eq:some equation in proof} yields
\begin{equation}
    Q_{\errWeight,\shots,\maxShotWeight} \leq \binom{\shots+\errWeight-1}{\shots-1}\cdot {\left(\frac{1-q^{-\maxShotWeight}}{1-q^{-1}}\right)}^{\errWeight} 
    \cdot \max_{\errWeightVec\in\wdecomp{\errWeight,\shots,\maxShotWeight}} \left\{ q^{\sum_{i=1}^{\shots} \errWeight_i (\maxShotWeight-\supwt_i')} \st  \supwtVec' = \scomp_{\maxShotWeight}(\errWeightVec, \supwt)  \right\}.
\end{equation}
The remainder of the proof follows from the proof in~\cite[Proposition 21]{puchingerGenericDecodingSumRank2022}, which leads to the desired result
\begin{equation}
    \Wgenwc\leq \WgenUBsimpleImproved.
\end{equation}
\end{proof}

\begin{figure}[ht]
  \centering
 \begin{tikzpicture}
      \begin{axis}[
          width=\linewidth, %
          height=8cm,
          grid=major, 
          grid style={dotted,gray!50},
          xlabel={$\shots$}, %
          ylabel={$\log_2(\text{Complexity})$},
          xlabel style={font=\scriptsize}, %
          ylabel style={font=\scriptsize}, %
          xmin=1,
          xmax=60, %
          ymin=0,
          ymax=200,
          xtick={1,2,3,4,5,6,10,12,15,20,30,60}, %
          xticklabels={1,2,3,4,5,6,10,12,15,20,30,60}, %
          xticklabel style={font=\scriptsize}, %
          yticklabel style={font=\scriptsize}, %
          legend style={
              at={(0.80,0.99)}, 
              anchor=north east, 
              legend columns=-1, 
              font=\scriptsize
          }, %
          legend cell align=left,
          clip=false, %
      ]

          \fill [pattern=north east lines, pattern color=gray!20] 
              (axis cs:1,0) rectangle (axis cs:3,200);

          \addplot[WgenUBsimple] 
              table[x=L,y=wold,col sep=space] {./data/wbounds_n60_k30_q2_m20_t9_s10.txt}; 
          \addlegendentry{$\WgenUBsimple$}
          
          \addplot[WgenUBsimpleImproved] 
              table[x=L,y=wbest,col sep=space] {./data/wbounds_n60_k30_q2_m20_t9_s10.txt}; 
          \addlegendentry{$\WgenUBsimpleImproved$}
          
          \addplot[WgenUB] 
              table[x=L,y=wup,col sep=space] {./data/wbounds_n60_k30_q2_m20_t9_s10.txt}; 
          \addlegendentry{$\WgenUB$}
          
          \addplot[WgenLB] 
              table[x=L, y=wlow, col sep=space] {./data/wbounds_n60_k30_q2_m20_t9_s10.txt};
          \addlegendentry{$\WgenLB$}

          \addplot [black, thick] coordinates {(3,0) (3,200)};
          
          \node (tr) at (axis cs:6,200) {}; %
          \node (br) at (axis cs:1,75) {};  %

          \draw[black, thick, dotted] (axis cs:1,75) rectangle (axis cs:6,200);

          \node[rotate=90, anchor=west, font=\scriptsize] at (axis cs:2,5) {Column Support};
          \node[rotate=90, anchor=west, font=\scriptsize] at (axis cs:4,5) {Row Support};

      \end{axis}
      
      \begin{scope}[shift={(0cm,8.2cm)}] %
          \begin{axis}[
              width=\linewidth,
              height=8cm,
              grid=major,
              grid style={dotted,gray!50},
              xlabel={$\shots$},
              ylabel={$\log_2(\text{Complexity})$},
              xlabel style={font=\scriptsize},
              ylabel style={font=\scriptsize},
              xmin=1, xmax=6,
              ymin=75, ymax=200,
              xtick={1,2,3,4,5,6},
              ytick={100,125,150,175},
              xticklabel style={font=\scriptsize},
              yticklabel style={font=\scriptsize},
              legend style={draw=none}, %
              legend cell align=left,
              clip=true, %
          ]

              \fill [pattern=north east lines, pattern color=gray!20] 
                  (axis cs:1,75) rectangle (axis cs:3,200);

              \addplot [black, thick] coordinates {(3,75) (3,200)};

              \addplot[WgenUBsimple] 
                  table[x=L,y=wold,col sep=space] {./data/wbounds_n60_k30_q2_m20_t9_s10.txt};
              
              \addplot[WgenUBsimpleImproved]
                  table[x=L,y=wbest,col sep=space] {./data/wbounds_n60_k30_q2_m20_t9_s10.txt};
              
              \addplot[WgenUB]
                  table[x=L,y=wup,col sep=space] {./data/wbounds_n60_k30_q2_m20_t9_s10.txt};
              
              \addplot[WgenLB]
                  table[x=L, y=wlow, col sep=space] {./data/wbounds_n60_k30_q2_m20_t9_s10.txt};

              \node (inset_tr) at (axis cs:6,75) {}; %
              \node (inset_br) at (axis cs:1,75) {};  %

          \end{axis}
      \end{scope}

      \draw[dotted, thick] 
          (tr) -- (inset_tr); %

      \draw[dotted, thick] 
          (br) -- (inset_br); %

    \end{tikzpicture}
    \caption{
        Illustration of the improved upper bound on the average complexity of Algorithm~\ref{alg:generic_sr_decoder} for $q=2$, $m=20$, $n=60$, $k=30$, $t=9$, $s=10$. At $\shots=3$, the algorithm transitions from guessing the column support to guessing the row support.
    }
    \label{fig:genNewSimple1}
\end{figure}
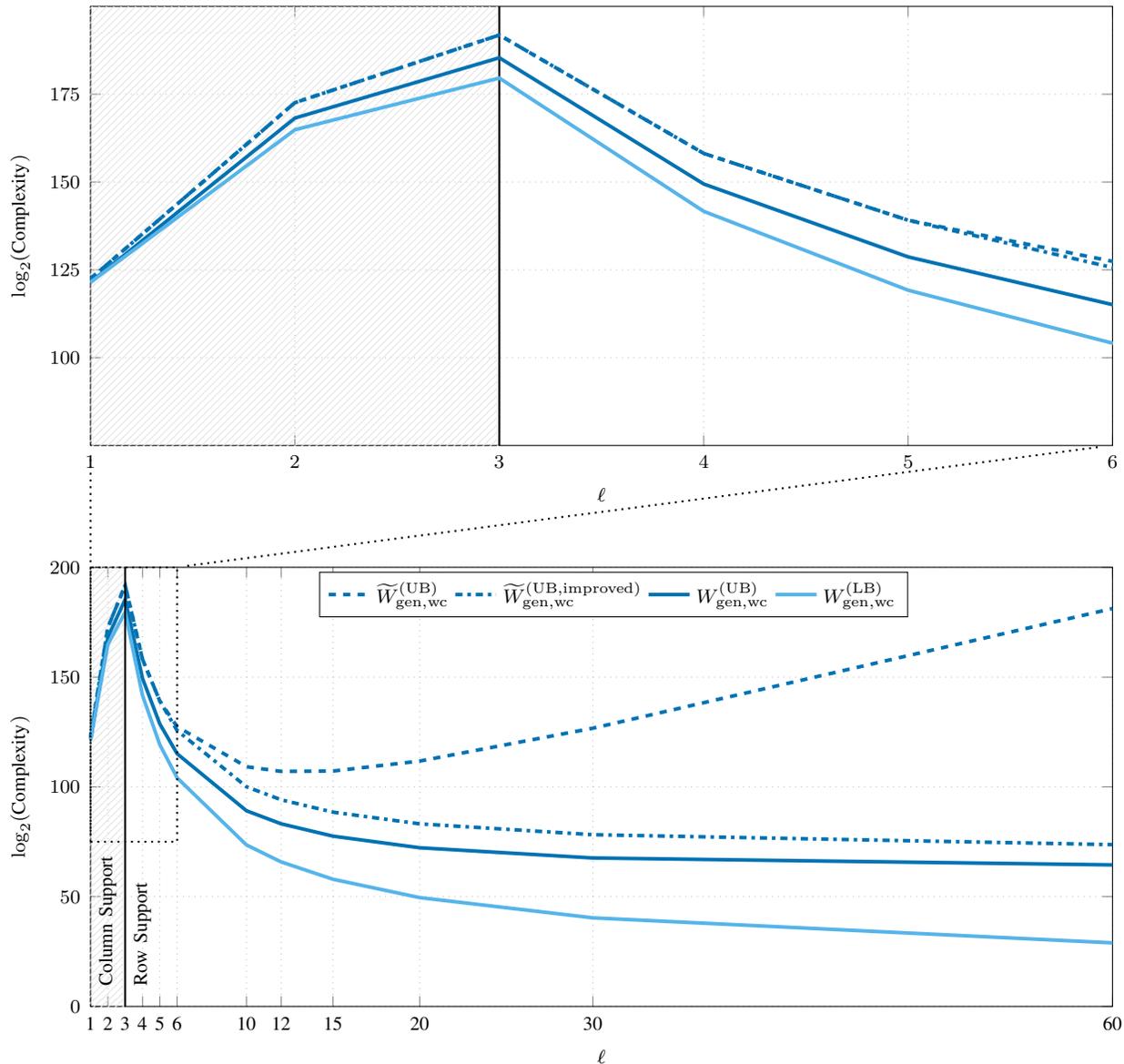

Figure~\ref{fig:genNewSimple1} illustrates the existing bounds from Puchinger et al.~\cite{puchingerGenericDecodingSumRank2022} alongside the new improved bound $\WgenUBsimpleImproved$ given by~\eqref{eq:improvedUB}. The figure shows the complexity as a function of the number of blocks $\shots$ for a fixed code length $n$, where the code length is defined as $n = \shots \shotLength$. This representation allows for the analysis of how the bounds and the improved bound behave as the number of blocks varies while keeping the code length and rate constant.

The existing bounds include the lower bound $\WgenLB$, the upper bound $\WgenUB$, and the simplified upper bound $\WgenUBsimple$, as defined in~\eqref{eq:WgenLB}, \eqref{eq:WgenUB}, and \eqref{eq:WgenUBsimple}, respectively. For a fair comparison, we used the same parameters as those presented in one of the figures from the original paper. The improved bound $\WgenUBsimpleImproved$ is significantly tighter and closer to the upper bound $\WgenUB$, particularly in the region near the Hamming metric, which corresponds to cases where $\shots \to n$ and $\shotLength \to 1$. This suggests that the new simplified bound provides a better approximation of the algorithm's complexity compared to the previous simplified upper bound $\WgenUBsimple$, especially when the sum-rank metric closely resembles the Hamming metric.

\subsection{Success Probability Analysis for the Average Case}\label{sec:generic_decoding_beyond}

We now consider the channel model described in Section~\ref{subsec:channel-model} and begin by deriving the success probability for the case of unique decoding, where exactly one solution exists. Additionally, we derive an upper bound on the success probability for decoding beyond the unique decoding radius, using \ac{RCU} arguments. This upper bound accounts for alternative solutions that the decoder in Algorithm~\ref{alg:generic_sr_decoder} may return in this scenario. 

In this analysis, we assume that the support drawing distribution is known. In the subsequent section, we will use the probabilities derived here to formalize an optimization problem with respect to the support drawing distribution.

In Line~\ref{alg:generic_sr_decoder:drawsupport} of Algorithm~\ref{alg:generic_sr_decoder}, we need to draw a suitable super support $\FSv = \FSv_1 \times \dots \times \FSv_\shots$, where $\FSv \subseteq \Fq^\maxShotWeight \times \dots \times \Fq^\maxShotWeight$, each $\FSv_i$ has dimension $\supwt_i$ for $i\inrange{1}{\shots}$, and $\sum_{i=1}^\shots \supwt_i = \supwt$. The distribution from which these super supports are drawn is a critical design parameter of the algorithm and requires careful optimization to maximize the algorithm's performance.

To draw the super support $\FSv$, we first draw a suitable rank profile $\v = [\supwt_1, \supwt_2, \ldots, \supwt_\shots] \in \wdecomp{\supwt,\shots,\maxShotWeight}$. Let $S$ be a discrete random variable over $\wdecomp{\supwt,\shots,\maxShotWeight}$, and denote the probability distribution of $S$ as $\ps{\v}$, i.e.,
\begin{equation}
    \ps{\v} \defeq \Pr[S=\v].
\end{equation}
Moreover, let $\psvec$ denote the probability vector for $S$, such that $\psvec = [\ps{\v_1},\ldots,\ps{\v_{|\wdecomp{\supwt,\shots,\maxShotWeight}|}}]$, where $\v_1, \ldots, \v_{|\wdecomp{\supwt,\shots,\maxShotWeight}|} \in \wdecomp{\supwt,\shots,\maxShotWeight}$ and $\psvec \in \pmfset{\wdecomp{\supwt,\shots,\maxShotWeight}}$.

After drawing the rank profile $\v$ according to the distribution $\ps{\v}$, the next step is to construct the super support $\FSv$. We draw $\FSv$ uniformly from the set $\Xi_{q,\maxShotWeight}(\v)$, which contains all valid super supports for the given rank profile $\v$. Each block $\FSv_i$ is drawn independently from the set of subspaces of $\F_q^{\maxShotWeight}$ with dimension $\supwt_i$.
By controlling the distribution $\ps{\v}$, we influence the distribution of the super support $\FSv$ and aim to minimize the expected complexity of the decoding process.

The following theorem gives the success probability for uniquely decoding a solution using a single iteration of Algorithm \ref{alg:generic_sr_decoder} under the average-case setting. 

\begin{theorem}\label{thm:genDecSP_unique}
Let $\mycode{C}\subseteq \Fqm^n$ be a linear sum-rank metric code with length $n$ and dimension $k$. Let $\c \in \mycode{C}$ be a codeword and consider a channel model as described in Section~\ref{sec:pre:channel model}, where the error $\e$ is drawn uniformly at random from $\errorSet_{q,\shotLength,m,\shots}(\errWeight)$, as defined in~\eqref{eq:error set definition}. Let $\SumRankWeight(\e) = \errWeight$, and assume that $\y=\c+\e$. Let $\smax$ denote the maximum sum-rank weight of the guessed super support, as defined in~\eqref{eq:smax}. Define the event $\uniqueEventGen$ as the event that Algorithm~\ref{alg:generic_sr_decoder} outputs $\c$ in a single iteration for the scenario of Problem~\ref{prob:SRSDecProblemY}. The probability of $\uniqueEventGen$ is given by
\begin{align}\label{eq:genDecSP_unique}
\ProbOf{\uniqueEventGen} &= \frac{1}{|\errorSet_{q,\shotLength,m,\shots}(\errWeight)|}\sum_{\supwt'=\errWeight}^{\smax} \sum_{\errWeightVec \in \wdecomp{\errWeight,\shots,\maxShotWeight}}  \probProfavg{q,\maxShotWeight,\supwt'}{\errWeightVec}\cdot  \prod_{i=1}^{\shots} \NMq{q}{m,\shotLength,\errWeight_i},
\end{align}
where $\probProfavg{q,\maxShotWeight,\supwt'}{\errWeightVec}$ is the average probability defined in~\eqref{eq:probProfavg}.
\end{theorem}
\begin{proof}
The average probability of decoding success can be expressed as a sum over all possible super space dimensions $\supwt'$ and error weight decompositions $\errWeightVec$, weighted by the probability of each error weight decomposition
\begin{align}
\ProbOf{\uniqueEventGen}&= \sum_{\supwt'=\errWeight}^{\smax} \sum_{\errWeightVec \in \wdecomp{\errWeight,\shots,\maxShotWeight}} \ProbOf{\errWeightVec} \cdot \probProfavg{q,\maxShotWeight,\supwt'}{\errWeightVec}.
\end{align}
Substituting the expression for $\ProbOf{\errWeightVec}$ from \eqref{eq:probability of errWeightVec for uniform error}, we arrive at the expression given in~\eqref{eq:genDecSP_unique}
\begin{align}
\ProbOf{\uniqueEventGen} &= \frac{1}{|\errorSet_{q,\shotLength,m,\shots}(\errWeight)|}\sum_{\supwt'=\errWeight}^{\smax} \sum_{\errWeightVec \in \wdecomp{\errWeight,\shots,\maxShotWeight}}  \probProfavg{q,\maxShotWeight,\supwt'}{\errWeightVec} \cdot  \prod_{i=1}^{\shots} \NMq{q}{m,\shotLength,\errWeight_i},
\end{align}
completing the proof.
\end{proof}

An upper bound on the probability of obtaining alternative solutions when using random linear sum-rank-metric codes is provided by the upcoming theorem.
For the analysis, we define $\probProfavg{q,\maxShotWeight,\supwt}{\errWeightVec}$ as the average probability over all super-support rank profiles
\begin{equation}\label{eq:probProfavg}
    \probProfavg{q,\maxShotWeight,\supwt}{\errWeightVec} \defeq \sum_{\supwtVec \in \wdecomp{\supwt,\shots,\maxShotWeight}} \ps{\supwtVec} \cdot \probProf{q,\maxShotWeight,\supwt}{\supwtVec, \errWeightVec},
\end{equation}
where $\ps{\v}$ is the probability distribution over the super-support rank profiles. 

\begin{theorem}[Random Coding Union Bound]\label{thm:rc_bound_gen}
Let $\mycode{C}$ be a random code of length $n$ and cardinality $|\mycode{C}|=q^{mk}$ over $\Fqm$, where each codeword is drawn uniformly at random from the ambient space $\Fqm^n$. Suppose that the received word $\y \in \Fqm^n$ is a noisy version of a codeword $\c \in \mycode{C}$, corrupted by an error vector $\e \in \Fqm^n$ of sum-rank weight $\errWeight$, i.e., $\y = \c + \e$. Let $\supwt$ be an integer satisfying $\errWeight \leq \supwt \leq \smax$.

The probability $\ProbOf{\altEventGen}$ of one iteration of Algorithm~\ref{alg:generic_sr_decoder} to output an alternative solution $\c'$ with $\c' \neq \c$ is upper bounded by
\begin{equation}
\ProbOf{\altEventGen} \leq \pRCUBgen,
\end{equation}
where
\begin{align}
    \pRCUBgen &\defeq q^{m(k-n)} \sum_{\errWeightVec \in \wdecomp{\errWeight,\shots,\maxShotWeight}} \probProfavg{q,\maxShotWeight,\supwt}{\errWeightVec} \prod_{i=1}^{\shots} \NMq{q}{m, \shotLength, \errWeight_i},
\end{align}
and $\probProfavg{q,\maxShotWeight,\supwt}{\errWeightVec}$ is the average probability defined in~\eqref{eq:probProfavg}.
\end{theorem}
\begin{proof}
    By assumption, each codeword in the codebook $\mycode{C}$ is drawn uniformly at random over $\Fqm^n$. Let $\c_j \in \mycode{C}$ with $\c_j \neq \c$ be one such alternative codeword with $j\inshotsarg{q^{mk}-1}$, and define $\mathcal{X}_j$ as the event that Algorithm~\ref{alg:generic_sr_decoder} can decode this codeword. Then
    \begin{equation}
        \ProbOf{\mathcal{X}_j} = \sum_{\substack{\e'\in\Fqm^n \\ \SumRankWeight(\e')=\errWeight}} \frac{1}{q^{mn}} \cdot \probProfavg{q,\maxShotWeight,\supwt}{\wdcmp(\e')}.
    \end{equation}
    Since $\probProfavg{q,\maxShotWeight,\supwt}{\wdcmp(\e')}$ only depends on the rank profile of $\e'$, we can change the sum to be over all rank profiles $\errWeightVec \in \wdecomp{\errWeight,\shots,\maxShotWeight}$ and multiply by the number of error vectors that have the same rank profile
    \begin{equation}
        \ProbOf{\mathcal{X}_j} = \sum_{\errWeightVec \in \wdecomp{\errWeight,\shots,\maxShotWeight}} \frac{1}{q^{mn}} \cdot \probProfavg{q,\maxShotWeight,\supwt}{\errWeightVec} \prod_{i=1}^{\shots} \NMq{q}{m, \shotLength, \errWeight_i}.
    \end{equation}
    The total probability of successful decoding is given by the union of the events $\mathcal{X}_1,\ldots,\mathcal{X}_{q^{mk}-1}$, which can be upper bounded by
    \begin{equation}
        \ProbOf{\bigcup_{j=1}^{q^{mk}-1} \mathcal{X}_j} \leq \sum_{j=1}^{q^{mk}-1} \ProbOf{\mathcal{X}_j} \leq q^{mk} \ProbOf{\mathcal{X}_j} = \pRCUBgen.
    \end{equation}
    Substituting the expression for $\ProbOf{\mathcal{X}_j}$ yields the desired upper bound on the success probability.
\end{proof}

Combining Theorem~\ref{thm:rc_bound_gen} and Theorem~\ref{thm:genDecSP_unique}, we can derive bounds on the success probability of Algorithm~\ref{alg:generic_sr_decoder} for one iteration to return at least one solution. We state these bounds in the following lemma.

\begin{theorem}
\label{thm:success_probability_bounds}
Let $\mycode{C}$ be a random code of length $n$ and size $q^{mk}$ over $\Fqm$, where each codeword is drawn uniformly at random from the ambient space $\Fqm^n$. Suppose that the received word $\y \in \Fqm^n$ is a noisy version of a codeword $\c \in \mycode{C}$, corrupted by an error vector $\e \in \Fqm^n$ with $\SumRankWeight(\e) = \errWeight$, i.e., $\y = \c + \e$. The success probability of Algorithm~\ref{alg:generic_sr_decoder} to output at least one solution satisfies
\begin{equation}
    \Pr[\text{success}] \geq \frac{1}{|\errorSet_{q,\shotLength,m,\shots}(\errWeight)|}\sum_{\supwt'=\errWeight}^{\smax} \sum_{\errWeightVec \in \wdecomp{\errWeight,\shots,\maxShotWeight}}  \probProfavg{q,\maxShotWeight,\supwt'}{\errWeightVec}\cdot  \prod_{i=1}^{\shots} \NMq{q}{m,\shotLength,\errWeight_i},
\end{equation}
and
\begin{equation}
    \Pr[\text{success}] \leq \left( \frac{1}{|\errorSet_{q,\shotLength,m,\shots}(\errWeight)|} + q^{m(k-n)} \right) \sum_{\supwt'=\errWeight}^{\smax} \sum_{\errWeightVec \in \wdecomp{\errWeight,\shots,\maxShotWeight}}  \probProfavg{q,\maxShotWeight,\supwt'}{\errWeightVec}\cdot  \prod_{i=1}^{\shots} \NMq{q}{m,\shotLength,\errWeight_i}.
\end{equation}
\end{theorem}
\begin{proof}
    First, we prove the lower bound on the success probability. Recall that we assume the received word $\y \in \Fqm^n$ is a noisy version of a codeword $\c \in \mycode{C}$, corrupted by an error vector $\e \in \Fqm^n$ with $\SumRankWeight(\e) = \errWeight$, i.e., $\y = \c + \e$. This implies that the codeword $\c$ is always within the decoding radius of the received word $\y$. Using the expression for $\Pr[\uniqueEventGen]$ from Theorem~\ref{thm:genDecSP_unique}, we have
    \begin{equation}
        \Pr[\text{success}] \geq \Pr[\uniqueEventGen] = \frac{1}{|\errorSet_{q,\shotLength,m,\shots}(\errWeight)|}\sum_{\supwt'=\errWeight}^{\smax} \sum_{\errWeightVec \in \wdecomp{\errWeight,\shots,\maxShotWeight}}  \probProfavg{q,\maxShotWeight,\supwt'}{\errWeightVec}\cdot  \prod_{i=1}^{\shots} \NMq{q}{m,\shotLength,\errWeight_i}.
    \end{equation}
    
    Next, we prove the upper bound on the success probability. Using union bound arguments and the expressions for $\Pr[\altEventGen]$ and $\Pr[\uniqueEventGen]$ from Theorem~\ref{thm:rc_bound_gen} and Theorem~\ref{thm:genDecSP_unique}, respectively, we obtain
    \begin{align}
        \Pr[\text{success}] &\leq \Pr[\uniqueEventGen] + \Pr[\altEventGen] \\
        &= \frac{1}{|\errorSet_{q,\shotLength,m,\shots}(\errWeight)|}\sum_{\supwt'=\errWeight}^{\smax} \sum_{\errWeightVec \in \wdecomp{\errWeight,\shots,\maxShotWeight}}  \probProfavg{q,\maxShotWeight,\supwt'}{\errWeightVec}\cdot  \prod_{i=1}^{\shots} \NMq{q}{m,\shotLength,\errWeight_i} \\
        &\quad + q^{m(k-n)} \sum_{\supwt'=\errWeight}^{\smax} \sum_{\errWeightVec \in \wdecomp{\errWeight,\shots,\maxShotWeight}}  \probProfavg{q,\maxShotWeight,\supwt'}{\errWeightVec}\cdot  \prod_{i=1}^{\shots} \NMq{q}{m,\shotLength,\errWeight_i} \\
        &= \left( \frac{1}{|\errorSet_{q,\shotLength,m,\shots}(\errWeight)|} + q^{m(k-n)} \right) \sum_{\supwt'=\errWeight}^{\smax} \sum_{\errWeightVec \in \wdecomp{\errWeight,\shots,\maxShotWeight}}  \probProfavg{q,\maxShotWeight,\supwt'}{\errWeightVec}\cdot  \prod_{i=1}^{\shots} \NMq{q}{m,\shotLength,\errWeight_i},
    \end{align}
    which concludes the lemma.
\end{proof}

From Theorem~\ref{thm:success_probability_bounds} we get that to find an optimal distribution $\ps{\v}$ to draw $\v$ from $\wdecomp{\supwt,\shots,\maxShotWeight}$ we need to maximize the term
\begin{equation}
\max_{\psvec \in \pmfset{\wdecomp{\supwt,\shots,\maxShotWeight}}} \sum_{\errWeightVec \in \wdecomp{\errWeight,\shots,\maxShotWeight}} \sum_{\v \in \wdecomp{\supwt,\shots,\maxShotWeight}} \ps{\v} \cdot \probProf{q,\maxShotWeight}{\v, \errWeightVec} \cdot \prod_{i=1}^{\shots} \NMq{q}{m, \shotLength, \errWeight_i},
\end{equation}
where $\pmfset{\wdecomp{\supwt,\shots,\maxShotWeight}}$ is the set of all valid \acp{PMF} over $\wdecomp{\supwt,\shots,\maxShotWeight}$ as defined in \eqref{eq:pmfset}.

\subsection{Optimizing the Support-Drawing Distribution via Linear Programming}\label{sec:generic-lp-solution-inefficient}

The process of drawing a super support from a known distribution can be further broken down. Instead of drawing a rank profile $\v \in \wdecomp{\supwt,\shots,\maxShotWeight}$ according to $\ps{\v}$, we can draw an ordered rank profile $\sord \in \intparts{\supwt,\shots,\maxShotWeight}$ according to a distribution $\psord{\sord}$, where $\sord = \ord(\v)$ is obtained by sorting the elements of $\v$ in non-increasing order. This simplification is possible due to symmetry, as the probability of drawing a particular rank profile remains the same for all permutations of that profile.

After drawing the ordered rank profile $\sord$, we perform a uniformly random permutation to obtain the final rank profile $\v$. The relation between the two probability distributions is given by
\begin{equation}\label{eq:sortedpvsunsortedp}
    \ps{\v} = \frac{\psord{\ord(\v)}}{ \PermCount{\ord(\v)}} =  \frac{\psord{\ord(\v)}}{ \PermCount{\v}}.
\end{equation}

By reducing the problem to optimizing the distribution $\psord{\sord}$ of ordered rank profiles, we have reduced the number of unknowns since we have $|\intparts{\supwt,\shots,\maxShotWeight}| \leq |\wdecomp{\supwt,\shots,\maxShotWeight}|$.

In summary, the process of drawing a suitable super support $\FSv$ can be broken down into three steps:
\begin{itemize}
    \item[1)] Draw an ordered rank profile $\sord \in \intparts{\supwt,\shots,\maxShotWeight}$ according to a distribution $\psord{\sord}$, which is the criterion we need to optimize, and then apply a uniformly random permutation to obtain the rank profile $\v$.
    \item[2)] For each $i\inrange{1}{\shots}$, draw $\FSv_i$ from the set of all spaces of dimension $\supwt_i$, independently for all blocks.
    \item[3)] Combine the individual blocks $\FSv_i$ to form the overall super support $\FSv = \FSv_1 \times \dots \times \FSv_\shots$.
\end{itemize}

By making use of~\eqref{eq:sortedpvsunsortedp} we can reduce the number of unknowns and instead maximize
\begin{equation}
    \begin{aligned}
        \max_{\psordvec \in \pmfset{\intparts{\supwt,\shots,\maxShotWeight}}} \sum_{\errWeightVec \in \intparts{\errWeight,\shots,\maxShotWeight}} & \sum_{\v \in \wdecomp{\supwt,\shots,\maxShotWeight}^{\geq \errWeightVec} } \frac{\PermCount{\errWeightVec}}{\PermCount{\v}} \psord{\ord(\v)}\probProf{q,\maxShotWeight}{\v, \errWeightVec} \prod_{i=1}^{\shots} \NMq{q}{m, \shotLength, \errWeight_i} \\
        = \max_{\psord{\v'} \in \pmfset{\intparts{\supwt,\shots,\maxShotWeight}}} & \sum_{\v' \in \intparts{\supwt,\shots,\maxShotWeight}} \psord{\v'} \cdot f(\v') \label{eq:lp-objective},
    \end{aligned}
\end{equation}
where
\begin{equation}
    f(\v') \coloneqq \sum_{\errWeightVec \in \intparts{\errWeight,\shots,\maxShotWeight}}  \frac{\PermCount{\errWeightVec}}{\PermCount{\v'}} \left(\sum_{\v''\in\SG{\v'}}\probProf{q,\maxShotWeight,\supwt}{\v'', \errWeightVec} \right) \prod_{i=1}^{\shots} \NMq{q}{m, \shotLength, \errWeight_i}.
\end{equation}

This optimization problem can be solved via \ac{LP} methods, where the objective function is~\eqref{eq:lp-objective} with $|\intparts{\supwt,\shots,\maxShotWeight}|$ unknowns. Although we have reduced the number of unknowns by restricting the optimization to the ordered set $\intparts{\supwt,\shots,\maxShotWeight}$, it is important to note that the cardinality of this set, and consequently the number of unknowns, can still grow super-polynomially with the parameters $\supwt$, $\shots$, and $\maxShotWeight$. Furthermore, computing the coefficients of the constraints requires summing over the set $\intparts{\errWeight,\shots,\maxShotWeight}$, which can be computationally demanding due to its potentially large cardinality.
Even if we successfully derive the optimal support-drawing distribution through this process, implementing an efficient algorithm to sample from this distribution poses another significant challenge. This limitation motivates the need for alternative approaches to simplify the optimization problem and develop more practical sampling algorithms.

\subsection{Efficient Optimization of the Support-Drawing Distribution in Generic Decoding}\label{sec:new heuristic approach for generic decoding for rcu bound}

In this section, we propose an efficient method to optimize the support-drawing distribution, addressing the computational challenges discussed earlier. By assuming independence between the sum-rank metric blocks, we greatly simplify the problem. Instead of drawing a complete rank profile vector $\v \in \wdecomp{\errWeight,\shots,\maxShotWeight}$ with a fixed total rank $\supwt$, we independently draw the rank $\supwt_i$ for each of the $\shots$ blocks. As a result, the sum rank $\supwt = \sum_{i=1}^\shots \supwt_i$ becomes a random variable. To prevent it from becoming unbounded, we constrain its expected value, $\mathbb{E}[\supwt]$, to match a predetermined relative sum-rank weight $\supwt / \shots$. This assumption reduces the complexity of the optimization problem by focusing on the distributions for individual blocks, and it enables efficient sampling from the optimized distribution, overcoming the practical limitations of the previous approach.

Although this heuristic approach may not always yield the optimal solution that accounts for the dependencies between the ranks of the guessed supports across different blocks, it still provides a good approximation. We demonstrate this numerically in Section~\ref{sec:generic-numerical-results} by comparing the performance of the heuristic solution with solutions that consider these dependencies, for parameters where the more complex optimization method is feasible.

Let $\psm{i}$ denote the marginal probability of drawing a super support $\myspace{F}_i$ with dimension $\supwt_i$, where $0 \leq \supwt_i \leq \maxShotWeight$ for $i\inshots$, and let $\psmvec = [\psm{0}, \ldots, \psm{\maxShotWeight}]$ represent the marginal probability vector. Assuming that the dimension of each subspace $\myspace{F}_i$ is drawn independently according to  $\psmvec$, the probability of a given rank profile $\v = [\supwt_1, \ldots, \supwt_\shots] \in \wdecomp{\supwt,\shots,\maxShotWeight}$ is given by
\begin{equation}\label{eq:super space marginal assumption}
    \ps{\v} = \prod_{i=1}^{\shots} \psm{\supwt_i}.
\end{equation}

We define the following two quantities
\begin{equation}\label{eq:B_def_0}
\tilde{B}_{q, m, \shotLength}(\errWeight, \supwt, \shots) \defeq \sum_{\w \in \wdecomp{\errWeight,\shots,\maxShotWeight}} \sum_{\v \in \wdecomp{\supwt,\shots,\maxShotWeight}} \ps{\v} \NMq{q}{m, \shotLength, w_i} \subspaceprob{\maxShotWeight}{w_i}{\supwt_i},
\end{equation}
and
\begin{equation}\label{eq:B_def}
B_{q, m, \shotLength}(\errWeight, \supwt, \shots) \defeq \sum_{\w \in \wdecomp{\errWeight,\shots,\maxShotWeight}} \sum_{\v \in \wdecomp{\supwt,\shots,\maxShotWeight}}\prod_{i=1}^{\shots} \psm{\supwt_i} \NMq{q}{m, \shotLength, w_i} \subspaceprob{\maxShotWeight}{w_i}{\supwt_i},
\end{equation}
where~\eqref{eq:B_def} is a special case of~\eqref{eq:B_def_0} using our independence assumption.
In Appendix~\ref{sec:appendix-1}, we show that \eqref{eq:B_def} is efficiently computable in polynomial time.
Using the definition of $B_{q, m, \shotLength}(\errWeight, \supwt, \shots)$ from~\eqref{eq:B_def} and the relaxation in~\eqref{eq:super space marginal assumption}, we can restate Theorem~\ref{thm:genDecSP_unique}, Theorem~\ref{thm:rc_bound_gen}, and Theorem~\ref{thm:success_probability_bounds} in the following corollaries, respectively.

\begin{corollary}\label{cor:rc_bound_gen_relaxed}
The probability $\Pr[\altEventGen]$ of having an alternative solution in Algorithm~\ref{alg:generic_sr_decoder} for a random linear code of length $n$ and cardinality $M=q^{mk}$ over $\Fqm$ can be upper bounded as
\begin{equation}
    \Pr[\altEventGen] \leq q^{m(k-n)} \sum_{\supwt=\errWeight}^{\supwt_{\max}} \tilde{B}_{q, m, \shotLength}(\errWeight, \supwt, \shots).
\end{equation}
\end{corollary}

\begin{corollary}\label{cor:genDecSP_unique_relaxed}
The probability $\Pr[\uniqueEventGen]$ that Algorithm~\ref{alg:generic_sr_decoder} outputs a unique solution $\c$ for a random linear code of length $n$ and cardinality $M=q^{mk}$ over $\Fqm$ is given by
\begin{equation}
    \Pr[\uniqueEventGen] = \frac{1}{|\errorSet_{q,\shotLength,m,\shots}(\errWeight)|} \sum_{\supwt=\errWeight}^{\smax} \tilde{B}_{q, m, \shotLength}(\errWeight, \supwt, \shots).
\end{equation}
\end{corollary}

\begin{corollary}\label{cor:success_probability_bounds_relaxed}
The success probability of Algorithm~\ref{alg:generic_sr_decoder} to output at least one solution satisfies
\begin{equation}
    \Pr[\text{success}] \geq \frac{1}{|\errorSet_{q,\shotLength,m,\shots}(\errWeight)|}\sum_{\supwt=\errWeight}^{\smax}\tilde{B}_{q, m, \shotLength}(\errWeight, \supwt, \shots),
\end{equation}
and
\begin{equation}
    \Pr[\text{success}] \leq \left( \frac{1}{|\errorSet_{q,\shotLength,m,\shots}(\errWeight)|} + q^{m(k-n)} \right) \sum_{\supwt=\errWeight}^{\supwt_{\max}} \tilde{B}_{q, m, \shotLength}(\errWeight, \supwt, \shots).
\end{equation}
\end{corollary}

From Corollary~\ref{cor:success_probability_bounds_relaxed}, under our independence assumption, the success probability is proportional to the term
\begin{equation}\label{eq:maxP_iid}
\sum_{\supwt=\errWeight}^{\supwt_{\max}} \tilde{B}_{q, m, \shotLength}(\errWeight, \supwt, \shots) = \sum_{\supwt=\errWeight}^{\supwt_{\max}} B_{q, m, \shotLength}(\errWeight, \supwt, \shots),
\end{equation}
which we aim to maximize over all possible $\psmvec \in \pmfset{\{0, \dots, \maxShotWeight\}}$.

In the following, we further upper bound the expression in~\eqref{eq:maxP_iid} and propose a method to maximize this upper bound, to obtain a valid solution for $\psmvec \in \pmfset{\{0, \dots, \maxShotWeight\}}$. That is
\begin{align}
    \sum_{\supwt=\errWeight}^{\smax}B_{q, m, \shotLength}(\errWeight, \supwt, \shots) &= \sum_{\supwt=\decRadius}^{\smax}\sum_{\w \in \wdecomp{\decRadius,\shots,\maxShotWeight}} \sum_{\v \in \wdecomp{\supwt,\shots,\maxShotWeight}}\prod_{i=1}^{\shots} \psm{\supwt_i}  \NMq{q}{m, \shotLength, w_i} \subspaceprob{\maxShotWeight}{w_i}{\supwt_i} \\
    &\leq \sum_{\w \in {\{0,\ldots,\maxShotWeight\}}^\shots} \sum_{\v \in {\{0,\ldots,\maxShotWeight\}}^\shots}\prod_{i=1}^{\shots} \psm{\supwt_i}  \NMq{q}{m, \shotLength, w_i} \subspaceprob{\maxShotWeight}{w_i}{\supwt_i} \\
    &= {\left(\sum_{w=0}^{\maxShotWeight} \sum_{\supwt=0}^{\maxShotWeight}\psm{\supwt}  \NMq{q}{m, \shotLength, w} \subspaceprob{\maxShotWeight}{w}{v} \right)}^{\shots},\label{eq:iid_upper_max}
\end{align}
where the last equality follows from the fact that we're summing over all possible $\shots$-tuples of $w_i$ and $v_i$, and for each tuple, we're computing the product of functions that depend only on $w_i$ and $v_i$. Since each $(w_i, v_i)$ pair is independent and the summations are over the same finite ranges, the combined sum over all vectors can be expressed as the $\shots$-th power of a single sum over $w$ and $v$. This is due to the distributive property of multiplication over addition and the independence of each component in the tuples, allowing us to factor the multiple sums and products into a single term raised to the power of $\shots$.

To maximize the right-hand side of \eqref{eq:iid_upper_max}, it suffices to maximize the expression

\begin{equation}\label{eq:maximize_term}
    \sum_{\errWeight'=0}^{\maxShotWeight} \sum_{\supwt'=0}^{\maxShotWeight} \psm{\supwt'} \NMq{q}{m, \shotLength, \errWeight'} \subspaceprob{\maxShotWeight}{\errWeight'}{\supwt'}.
\end{equation}

This expression is closely related to the average probability that a randomly drawn super space $\FS_i$ contains the error space $\ES_i$ in a single block, averaged over all possible rank weights $\errWeight'$. The average single-block success probability is given as
\begin{equation}\label{eq:single_block_prob}
    \sum_{\errWeight'=0}^{\maxShotWeight} \Pr[\errWeight'] \sum_{\supwt'=0}^{\maxShotWeight} \psm{\supwt'} \subspaceprob{\maxShotWeight}{\errWeight'}{\supwt'},
\end{equation}
where $\Pr[\errWeight']$ is the marginal probability of an error of rank weight $\errWeight'$ occurring in a single block.

In the asymptotic setting, where $\shotLength$ and $m$ are fixed and the number of blocks $\shots \to \infty$, the assumption of independence between blocks becomes valid due to the law of large numbers and the concept of typical sequences from statistical mechanics. In this regime, the empirical distribution of error weights in the blocks converges to the marginal distribution $\Pr[\errWeight']$, which can be approximated by the Boltzmann distribution~\cite{couveeBoundsSphereSizes2024}

\begin{equation}\label{eq:boltzmann_approx}
    \Pr[\errWeight'] = \frac{\NMq{q}{m,\shotLength, \errWeight'}}{\sum_{\errWeight''=0}^{\maxShotWeight} \NMq{q}{m,\shotLength, \errWeight''} e^{-\boltzmParam \errWeight''}} \cdot  e^{-\boltzmParam \errWeight'},
\end{equation}
where $\boltzmParam$ is the unique solution to the weight constraint
\begin{equation}\label{eq:weight_constraint}
    \mathbb{E}[\errWeight'] = \sum_{\errWeight''=0}^{\maxShotWeight}\errWeight'' \cdot \Pr[\errWeight'']  = \frac{\errWeight}{\shots}.
\end{equation}
By substituting \eqref{eq:boltzmann_approx} into \eqref{eq:single_block_prob}, we obtain the single-block success probability under the asymptotic error-weight distribution.

Maximizing the single-block success probability in \eqref{eq:single_block_prob} effectively maximizes the overall success probability in the asymptotic regime. Although \eqref{eq:maximize_term} represents an upper bound on the success probability, this upper bound becomes tight as $\shots \to \infty$ due to the convergence properties established by the law of large numbers. Therefore, optimizing this upper bound is justified because it aligns with maximizing the actual success probability in the asymptotic setting.

This connection reveals that optimizing \eqref{eq:maximize_term} to obtain an optimal marginal distribution $\psm{\supwt'}$ for the guessed super-support dimensions is beneficial for maximizing~\eqref{eq:maxP_iid}.

To optimize~\eqref{eq:maximize_term}, our approach focuses on the marginal distribution $\psmvec \in \pmfset{\{0,\ldots,\maxShotWeight\}}$ rather than directly optimizing $\psvec$, aiming to approximate the optimal average rank profile for the super support. Since directly optimizing $\psmvec$ results in a distribution independent of the number of blocks $\shots$, we impose the constraint $\psm{i} = \frac{x_i}{\shots}$, where $x_i \in \NN$ represents the number of occurrences of rank $i$ across the $\shots$ blocks.

We then maximize the objective in~\eqref{eq:iid_upper_max} using linear integer programming with appropriate constraints and non-negativity conditions. This method assumes independence of rank weights across the $\shots$ blocks, which holds asymptotically as $\shots \to \infty$ for fixed $\shotLength$ and $m$.

By applying this method, we obtain a solution $\x = {[x_0, \ldots, x_{\maxShotWeight}]}$, from which we construct the ordered rank profile $\hat{\v}\in\intparts{\errWeight,\shots,\maxShotWeight}$ as
\begin{equation}\label{eq:heur_svec_avgrcu}
    \hat{\supwtVec} = [\underbrace{\maxShotWeight,\ldots,\maxShotWeight}_{x_{\maxShotWeight}\text{ times}},\underbrace{\maxShotWeight-1,\ldots,\maxShotWeight-1}_{x_{\maxShotWeight-1}\text{ times}},\ldots,\underbrace{1,\ldots,1}_{x_{1}\text{ times}},\underbrace{0,\ldots,0}_{x_{0}\text{ times}}],
\end{equation}
where each element $i \in \{0,\ldots,\maxShotWeight\}$ appears exactly $x_i$ times in the vector $\hat{\v}$.
We then have that
\begin{equation}\label{eq:heu-distr-gen}
    \psord{\v'}^{(\text{heu})} \defeq 
    \begin{cases}
        1 & \text{if } \v' = \hat{\v}, \\
        0 & \text{otherwise}.
    \end{cases}
\end{equation}
Using the relation in~\eqref{eq:sortedpvsunsortedp}, we obtain the overall probability for $\v$, i.e., $\ps{\v}^{(\text{heu})}$.
For this specific \ac{PMF} we can write~\eqref{eq:B_def_0} as
\begin{equation}\label{eq:B_def_0_heu}
\tilde{B}_{q, m, \shotLength}(\errWeight, \supwt, \shots) = \sum_{\w \in \wdecomp{\errWeight,\shots,\maxShotWeight}} \prod_{i=1}^{\shots} \NMq{q}{m, \shotLength, w_i} \subspaceprob{\maxShotWeight}{w_i}{\hat{\supwt}_i}.
\end{equation}
As shown in Appendix~\ref{sec:appendix-1}, this expression can be efficiently computed in polynomial time since it is a special case of~\eqref{eq:B_def} for a fixed vector $\supwtVec = \hat{\supwtVec}$.

Thus, the bounds on the overall expected runtime provided in the following theorem, which are general for any support-guessing distribution, can be efficiently computed for our specific support-guessing distribution given by~\eqref{eq:heur_svec_avgrcu} and~\eqref{eq:heu-distr-gen}.

\begin{theorem}
\label{thm:complexity_bounds}
Under the same assumptions as in Corollary~\ref{cor:success_probability_bounds_relaxed}, the overall expected runtime $\WgenRCU$ of Algorithm~\ref{alg:generic_sr_decoder} to output at least one solution is bounded by
\begin{equation}
    \WgenRCULB \leq \WgenRCU \leq \WgenRCUUB,
\end{equation}
with
\begin{equation}\label{eq:WgenRCULB}
    \WgenRCULB \defeq \WiterErDec \left(\left( \frac{1}{|\errorSet_{q,\shotLength,m,\shots}(\errWeight)|} + q^{m(k-n)} \right) \sum_{\supwt=\errWeight}^{\supwt_{\max}} \tilde{B}_{q, m, \shotLength}(\errWeight, \supwt, \shots)\right)^{-1},
\end{equation}
and
\begin{equation}\label{eq:WgenRCUUB}
    \WgenRCUUB \defeq \WiterErDec \left(\frac{1}{|\errorSet_{q,\shotLength,m,\shots}(\errWeight)|}\sum_{\supwt=\errWeight}^{\smax}\tilde{B}_{q, m, \shotLength}(\errWeight, \supwt, \shots)\right)^{-1},
\end{equation}
where $\WiterErDec$ denotes the cost of one iteration of Algorithm~\ref{alg:generic_sr_decoder} and $\WiterErDec \in \oh{(n-k)^3 m^3}$ operations over $\Fq$ and we neglect the complexity of drawing from $\ps{\v}$.
\end{theorem}
\begin{proof}
    For the lower bound on the complexity, we consider the worst-case scenario for the complexity of each iteration, denoted by $\WiterErDec$. The expected number of iterations until success is the reciprocal of the success probability. Using the upper bound on the success probability from Corollary~\ref{cor:success_probability_bounds_relaxed}, the lower bound on the overall expected runtime satisfies
    \begin{equation}
        \WgenRCULB = \WiterErDec \left(\left( \frac{1}{|\errorSet_{q,\shotLength,m,\shots}(\errWeight)|} + q^{m(k-n)} \right) \sum_{\supwt=\errWeight}^{\supwt_{\max}} \tilde{B}_{q, m, \shotLength}(\errWeight, \supwt, \shots)\right)^{-1}.
    \end{equation}

    For the upper bound on the complexity, we use the cost of one iteration $\WiterErDec$, which is $\oh{(n-k)^3 m^3}$ operations over $\Fq$ according to~\cite[Theorem 13 and Theorem 14]{puchingerGenericDecodingSumRank2022}. Using the heuristic probability distribution for the guessing super support as in~\eqref{eq:heur_svec_avgrcu}, we can neglect the complexity of drawing the rank profile of the guessing support. The expected number of iterations until success is the reciprocal of the success probability. Using the lower bound on the success probability from Corollary~\ref{cor:success_probability_bounds_relaxed}, the upper bound on the overall expected runtime satisfies
    \begin{equation}
        \WgenRCUUB = \WiterErDec \left(\frac{1}{|\errorSet_{q,\shotLength,m,\shots}(\errWeight)|}\sum_{\supwt=\errWeight}^{\smax}B_{q, m, \shotLength}(\errWeight, \supwtVec, \shots)\right)^{-1},
    \end{equation}
    which concludes the theorem.
\end{proof}

Note that the assumption made in Theorem~\ref{thm:complexity_bounds}, which neglects the complexity of drawing from $\ps{\v}$, is valid since, in our solution, we only need to permute the support-guessing rank profile $\hat{\supwtVec}$ uniformly at random.

\subsection{Numerical Results for the Generic Decoding Algorithm}\label{sec:generic-numerical-results}

In this section, we compare the complexity analysis of using a support-drawing distribution derived from the method described in Section~\ref{sec:new heuristic approach for generic decoding for rcu bound} for the average case against the worst-case bounds from~\cite{puchingerGenericDecodingSumRank2022}. We evaluate the average complexity over all error patterns for a specific sum-rank weight and plot the logarithmic complexity (base 2) versus the number of blocks $\shots$, while keeping the code parameters and field size $q^m$ constant. The length of each individual block $\eta$ is adjusted as $\shots$ varies.

Figure~\ref{fig:complexity_comparison_1} shows the complexity for generic decoding beyond the unique decoding radius with parameters $q=2$, $m=20$, $n=60$, $k=30$, $t=9$, and $v=10$ while in Figure~\ref{fig:complexity_comparison_2}, we increase $v$ to $\smax$.
We include the upper bound $\WgenRCUUB$~\eqref{eq:WgenRCUUB} and the lower bound $\WgenRCULB$~\eqref{eq:WgenRCULB} for the expected complexity of Algorithm~\ref{alg:generic_sr_decoder}. The lower bound reflects the effect of decoding beyond the unique decoding radius and accounts for alternative solutions.

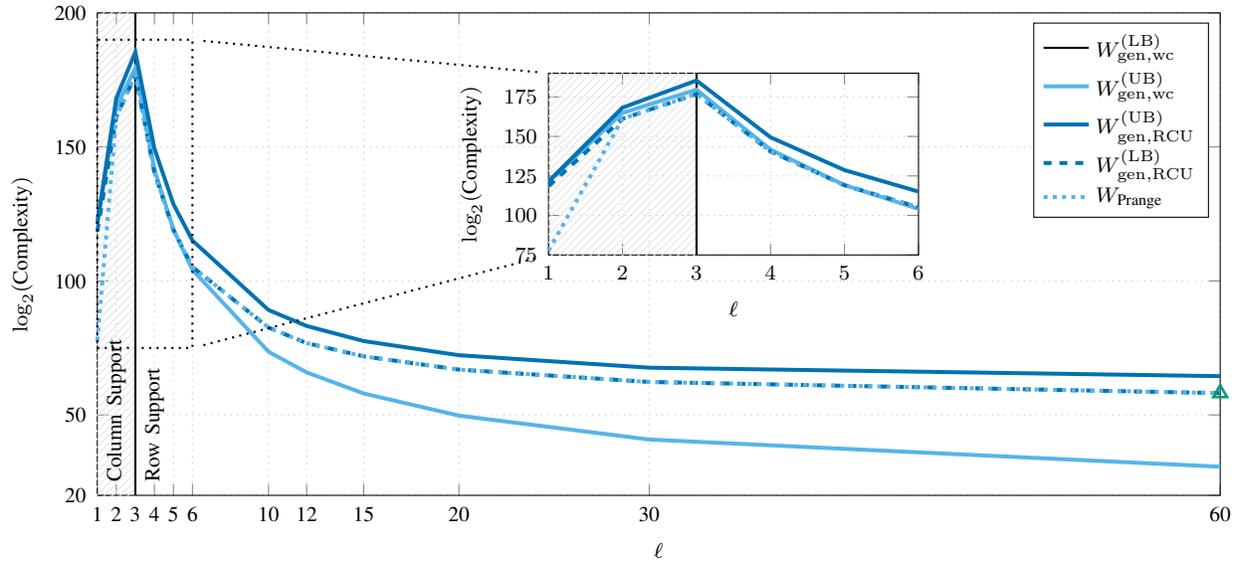
\begin{figure}[htbp]
  \centering
        \begin{tikzpicture}
        \begin{axis}[
            width=\linewidth, %
            height=8cm,
            grid=major, 
            grid style={dotted,gray!50},
            xlabel={$\shots$}, %
            ylabel={$\log_2(\text{Complexity})$},
            xlabel style={font=\scriptsize}, %
            ylabel style={font=\scriptsize}, %
            xmin=1,
            xmax=60, %
            ymin=20,
            ymax=200,
            xtick={1,2,3,4,5,6,10,12,15,20,30,60}, %
            xticklabels={1,2,3,4,5,6,10,12,15,20,30,60}, %
            ytick={20,50,100,150,200}, %
            yticklabels={20,50,100,150,200}, %
            xticklabel style={font=\scriptsize}, %
            yticklabel style={font=\scriptsize}, %
            legend style={
                at={(0.99,0.98)}, 
                anchor=north east, 
                legend columns=1, 
                font=\scriptsize
            }, %
            legend cell align=left,
            clip=false, %
        ]

            \fill [pattern=north east lines, pattern color=gray!20] 
                (axis cs:1,20) rectangle (axis cs:3,200);

            \addplot [black, thick] coordinates {(3,20) (3,200)};
            
            \addplot[WgenLB] 
                table[x=L, y=WF, col sep=comma] {data/results/W_QwcLB_gen_N60_K30_q2_m20_t9_s10.txt};
            \addlegendentry{$\WgenLB$}

            \addplot[WgenUB]
                table[x=L, y=WF, col sep=comma] {data/results/W_QwcUB_gen_N60_K30_q2_m20_t9_s10.txt};
            \addlegendentry{$\WgenUB$}    

            \addplot[WgenRCUUB]
                table[x=L, y=WF, col sep=comma] {data/results/W_heuvecavg_gen_N60_K30_q2_m20_t9_s10.txt};
            \addlegendentry{$\WgenRCUUB$}

            \addplot[WgenRCULB]
                table[x=L, y=WF, col sep=comma] {data/results/W_rcveclb_gen_N60_K30_q2_m20_t9_s10.txt};
            \addlegendentry{$\WgenRCULB$}

          \addplot[Wprange]
              coordinates {(60, 58.1477619214834)};
          \addlegendentry{$W_\text{Prange}$}

            \node (tr) at (axis cs:6,190) {}; %
            \node (br) at (axis cs:6,75) {};  %

            \draw[black, thick, dotted] (axis cs:1,75) rectangle (axis cs:6,190);

            \node[rotate=90, anchor=west, font=\scriptsize] at (axis cs:2,21) {Column Support};
            \node[rotate=90, anchor=west, font=\scriptsize] at (axis cs:4,21) {Row Support};

        \end{axis}
        
        \begin{scope}[shift={(6cm,3.2cm)}] %
            \begin{axis}[
                width=6.5cm,
                height=4cm,
                grid=major,
                grid style={dotted,gray!50},
                xlabel={$\shots$},
                ylabel={$\log_2(\text{Complexity})$},
                xlabel style={font=\scriptsize},
                ylabel style={font=\scriptsize},
                xmin=1, xmax=6,
                ymin=75, ymax=190,
                xtick={1,2,3,4,5,6},
                ytick={75,100,125,150,175,200},
                xticklabel style={font=\scriptsize},
                yticklabel style={font=\scriptsize},
                legend style={draw=none}, %
                legend cell align=left,
                clip=true, %
            ]

                \fill [pattern=north east lines, pattern color=gray!20] 
                    (axis cs:1,75) rectangle (axis cs:3,190);

                \addplot [black, thick] coordinates {(3,75) (3,190)};

                \addplot[WgenLB] 
                    table[x=L, y=WF, col sep=comma] {data/results/W_QwcLB_gen_N60_K30_q2_m20_t9_s10.txt};
                
                \addplot[WgenUB]
                    table[x=L, y=WF, col sep=comma] {data/results/W_QwcUB_gen_N60_K30_q2_m20_t9_s10.txt};
                
                \addplot[WgenRCUUB]
                    table[x=L, y=WF, col sep=comma] {data/results/W_heuvecavg_gen_N60_K30_q2_m20_t9_s10.txt};
                
                \addplot[WgenRCULB]
                    table[x=L, y=WF, col sep=comma] {data/results/W_rcveclb_gen_N60_K30_q2_m20_t9_s10.txt};

                \node (inset_tr) at (axis cs:1,190) {}; %
                \node (inset_br) at (axis cs:1,75) {};  %

            \end{axis}
        \end{scope}

        \draw[dotted, thick] 
            (tr) -- (inset_tr); %

        \draw[dotted, thick] 
            (br) -- (inset_br); %

    \end{tikzpicture}
    \vspace{-3em} %
    \caption{Complexity comparison for generic decoding with parameters: $q=2$, $m=20$, $n=60$, $k=30$, $t=9$, and $v=10$.}
    \label{fig:complexity_comparison_1}
\end{figure}

\begin{figure}[htbp]
  \centering
        \begin{tikzpicture}
      \begin{axis}[
          width=\linewidth, %
          height=8cm,
          grid=major, 
          grid style={dotted,gray!50},
          xlabel={$\shots$}, %
          ylabel={$\log_2(\text{Complexity})$},
          xlabel style={font=\scriptsize}, %
          ylabel style={font=\scriptsize}, %
          xmin=1,
          xmax=60,
          ymin=20,
          ymax=130,
          xtick={1,2,3,4,5,6,10,12,15,20,30,60}, %
          xticklabels={1,2,3,4,5,6,10,12,15,20,30,60}, %
          xticklabel style={font=\scriptsize}, %
          yticklabel style={font=\scriptsize}, %
          legend style={
              at={(0.99,0.98)}, 
              anchor=north east, 
              legend columns=1, 
              font=\scriptsize
          }, %
          legend cell align=left,
          clip=false, %
      ]

          \fill [pattern=north east lines, pattern color=gray!20] 
              (axis cs:1,20) rectangle (axis cs:3,130);

          \addplot[WgenLB] 
              table[x=L, y=WF, col sep=comma] {data/generic_results/W_QwcLB_gen_N60_K30_q2_m20_t9_s30.txt};
          \addlegendentry{$\WgenLB$}
          
          \addplot[WgenUB]
              table[x=L, y=WF, col sep=comma] {data/generic_results/W_QwcUB_gen_N60_K30_q2_m20_t9_s30.txt};
          \addlegendentry{$\WgenUB$}
          
          \addplot[WgenRCUUB]
              table[x=L, y=WF, col sep=comma] {data/generic_results/W_heuvecavg_gen_N60_K30_q2_m20_t9_s30.txt};
          \addlegendentry{$\WgenRCUUB$}
              
          \addplot[WgenRCULB]
              table[x=L, y=WF, col sep=comma] {data/generic_results/W_rcveclb_gen_N60_K30_q2_m20_t9_s30.txt};
          \addlegendentry{$\WgenRCULB$}
          
          \addplot[Wprange]
              coordinates {(60, 37.6994570379119)};
          \addlegendentry{$W_\text{Prange}$}

          \addplot [black, thick] coordinates {(3,20) (3,130)};

          \node[rotate=90, anchor=west, font=\scriptsize] at (axis cs:2,21) {Column Support};
          \node[rotate=90, anchor=west, font=\scriptsize] at (axis cs:4,21) {Row Support};

          \node (tr) at (axis cs:6,130) {}; %
          \node (br) at (axis cs:6,40) {};  %

          \draw[black, thick, dotted] (axis cs:1,40) rectangle (axis cs:6,130);

      \end{axis}    

      \begin{axis}[
          width=\linewidth, %
          height=8cm,
          grid=none, %
          axis y line=none, %
          axis x line=top, %
          xmin=1,
          xmax=60,
          ymin=20, %
          ymax=130,
          xtick={1,2,3,4,5,6,10,12,15,20,30,60}, %
          xticklabels={10,20,30,30,30,30,30,30,30,30,30,30}, %
          xticklabel style={
              font=\tiny, %
              /pgf/number format/fixed,
              /pgf/number format/precision=0,
              anchor=south,
              rotate=45, %
          }, %
          xlabel={{$v$}}, %
          xlabel style={
              font=\scriptsize,
              at={(axis description cs:0.5,1.05)}, %
              anchor=south,
          },
          axis line style={-}, %
      ]

      \end{axis}
      
      \begin{scope}[shift={(4.5cm,3.4cm)}] %
          \begin{axis}[
              width=6.5cm,
              height=4cm,
              grid=major,
              grid style={dotted,gray!50},
              xlabel={$\shots$},
              ylabel={$\log_2(\text{Complexity})$},
              xlabel style={font=\scriptsize},
              ylabel style={font=\scriptsize},
              xmin=1, xmax=6,
              ymin=40, ymax=130, %
              xtick={1,2,3,4,5,6},
              ytick={40,60,80,100,120},
              xticklabel style={font=\scriptsize},
              yticklabel style={font=\scriptsize},
              legend style={draw=none}, %
              legend cell align=left,
              clip=true, %
          ]

              \fill [pattern=north east lines, pattern color=gray!20] 
                  (axis cs:1,40) rectangle (axis cs:3,130);

              \addplot [black, thick] coordinates {(3,40) (3,130)};

              \addplot[WgenLB] 
                  table[x=L, y=WF, col sep=comma] {data/generic_results/W_QwcLB_gen_N60_K30_q2_m20_t9_s30.txt};
              
              \addplot[WgenUB]
                  table[x=L, y=WF, col sep=comma] {data/generic_results/W_QwcUB_gen_N60_K30_q2_m20_t9_s30.txt};
              
              \addplot[WgenRCUUB]
                  table[x=L, y=WF, col sep=comma] {data/generic_results/W_heuvecavg_gen_N60_K30_q2_m20_t9_s30.txt};
              
              \addplot[WgenRCULB]
                  table[x=L, y=WF, col sep=comma] {data/generic_results/W_rcveclb_gen_N60_K30_q2_m20_t9_s30.txt};

              \node (inset_tr) at (axis cs:1,130) {}; %
              \node (inset_br) at (axis cs:1,40) {};  %

          \end{axis}
      \end{scope}

      \draw[dotted, thick] 
          (tr) -- (inset_tr); %

      \draw[dotted, thick] 
          (br) -- (inset_br); %

    \end{tikzpicture}
    \vspace{-3em} %
    \caption{Complexity comparison for generic decoding with parameters: $q=2$, $m=20$, $n=60$, $k=30$, $t=9$ and $v=\smax$.}
    \label{fig:complexity_comparison_2}
\end{figure}

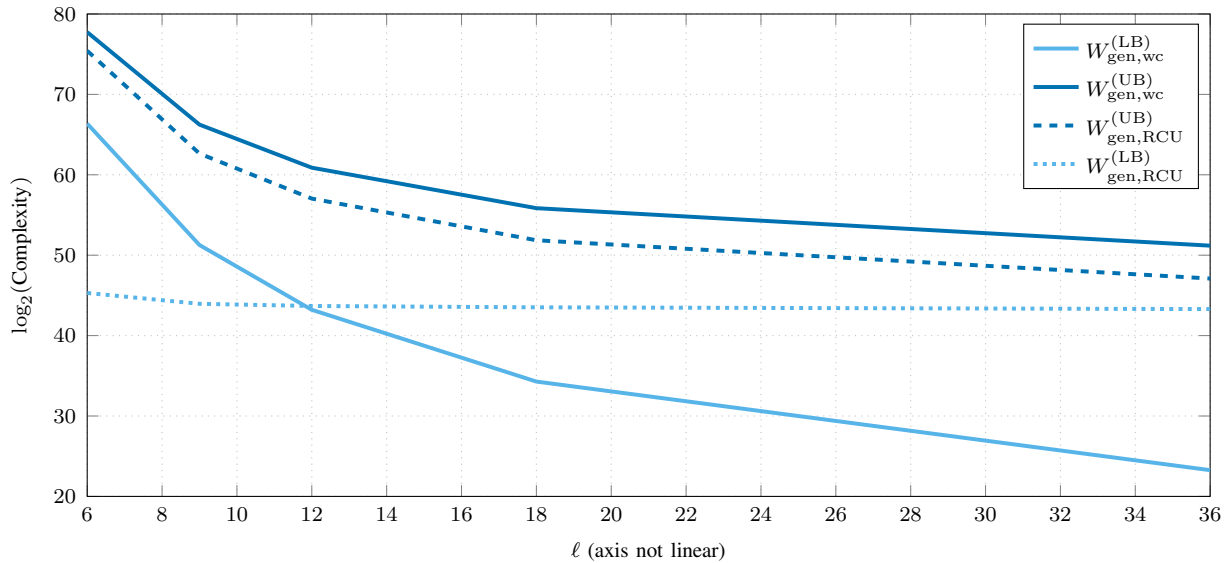
\begin{figure}[htbp]
  \centering
    \begin{tikzpicture}
  \begin{axis}[
      width=\linewidth, %
      height=8cm,
      grid=major, 
      grid style={dotted,gray!50},
      xlabel={$\shots$ (axis not linear)}, %
      ylabel={$\log_2(\text{Complexity})$},
      xlabel style={font=\scriptsize}, %
      ylabel style={font=\scriptsize}, %
      xmin=6,
      xmax=36,
      ymax=80,
      ymin=20,
      xticklabel style={font=\scriptsize}, %
      yticklabel style={font=\scriptsize}, %
      legend style={at={(0.99,0.98)}, anchor=north east, %
                    legend columns=1, %
                    font=\scriptsize}, %
      legend cell align=left,
    ]

    \addplot[WgenLB] 
    table[x=L, y=WF, col sep=comma] {data/generic_results/W_QwcLB_gen_N36_K22_q2_m6_t10_s10.txt};
    \addlegendentry{$\WgenLB$}
    
    \addplot[WgenUB]
    table[x=L, y=WF, col sep=comma] {data/generic_results/W_QwcUB_gen_N36_K22_q2_m6_t10_s10.txt};
    \addlegendentry{$\WgenUB$}
    
    \addplot[WgenRCUUB]
    table[x=L, y=WF, col sep=comma] {data/generic_results/W_heuvecavg_gen_N36_K22_q2_m6_t10_s10.txt};
    \addlegendentry{$\WgenRCUUB$}
        
    \addplot[WgenRCULB]
    table[x=L, y=WF, col sep=comma] {data/generic_results/W_rcveclb_gen_N36_K22_q2_m6_t10_s10.txt};
    \addlegendentry{$\WgenRCULB$}

  \end{axis}    
\end{tikzpicture}
    \vspace{-3em} %
    \caption{Complexity comparison for generic decoding with parameters: $q=2$, $m=6$, $n=36$, $k=22$, $t=10$ and $v=10$.}
    \label{fig:complexity_comparison_3}
\end{figure}

The figures show that the effect of alternative solutions, indicated by the divergence between the upper and lower bounds, becomes more prominent near the rank metric (i.e., when $\shots$ approaches 1) for the chosen parameters. However, this behavior can vary depending on the parameters, as seen in Figure~\ref{fig:complexity_comparison_3}, which uses the parameters $q=2$, $m=6$, $n=36$, $k=22$, $t=10$, and $v=10$. In this case, a significant difference between the lower and upper bounds persists even when the number of blocks does not correspond to the rank metric.

With the increase of $v$ to $\smax$ in Figure~\ref{fig:complexity_comparison_2}, the upper and lower bounds for the worst-case scenario become even looser. The upper bound, in particular, deviates significantly from the complexity of the Prange algorithm, denoted as $W_\text{Prange}$, and given by~\cite{prangeUseInformationSets1962}
\begin{equation}
    W_\text{Prange} = \WiterGen \frac{\binom{n}{\errWeight}}{\binom{v}{\errWeight}},
\end{equation}
especially at $\shots=60$, which corresponds to the special case of the Hamming metric. In this case, not only does the complexity of our solution coincide with the Prange algorithm, but the algorithm itself is exactly the Prange algorithm. 

Additionally, in Figure~\ref{fig:complexity_comparison_2}, for $\shots < 12$, the complexity of our solution falls below the lower bound of the worst-case scenario. This demonstrates that worst-case bounds may not always provide accurate estimates for real-world scenarios. For instance, when selecting parameters for cryptosystems based on sum-rank metric codes, relying solely on worst-case bounds may lead to underestimating the actual complexity of practical attacks in certain regimes. The lower bound $\WgenRCULB$ provides a closer approximation to the actual complexity in practice than the worst-case bounds from~\cite{puchingerGenericDecodingSumRank2022}. 

We performed extensive computations using \ac{LP} to account for dependencies between the blocks, as described in Section~\ref{sec:generic-lp-solution-inefficient}, for the parameters used in Figures~\ref{fig:complexity_comparison_1} and~\ref{fig:complexity_comparison_2}. For these parameters this approach is feasible for $\shots$ up to $10$ and we obtained the exact same support-guessing distribution corresponding to the rank profiles as our efficient solution. This is particularly interesting, as our efficient method is expected to yield tighter results for larger block sizes, i.e., as $\shots \to n$. Thus, our findings indicate that for these parameters, the efficient approach remains effective even at lower block counts.

Table~\ref{tab:bounds_conditions} summarizes the individual bounds and their applicable decoding scenarios.

\begin{table}[htbp]
\caption{Overview of the bounds and their applicable scenarios.}
\label{tab:bounds_conditions}
\centering
\renewcommand{\arraystretch}{1.4}
\begin{tabular}{M{0.8cm} R{1cm} c R{9cm}}
 \toprule
 \multicolumn{2}{c}{\textbf{Bound}} & \textbf{Applies to} & \textbf{Conditions} \\
 \midrule
 \multirow[c]{2}{*}{\rotatebox[origin=c]{90}{Worst case}} & $\WgenLB$ & Problem~\ref{prob:uniqueDecoding} & Exactly one solution exists for the worst-case rank profile channel \\[0pt]
 \cmidrule{2-4}
 & $\WgenUB$ & Problem~\ref{prob:SRSDecProblemUniformSyndrome} & At least one solution exists for the worst-case rank profile channel \\[0pt]
 \midrule
 \multirow[c]{2}{*}{\rotatebox[origin=c]{90}{Average case}} & $\WgenRCUUB$ & Problem~\ref{prob:SRSDecProblemUniformSyndrome} & At least one solution exists; exact when exactly one solution exists \\[0pt]
 \cmidrule{2-4}
 & $\WgenRCULB$ & Problem~\ref{prob:SRSDecProblemUniformSyndrome} & At least one solution exists; accounts for alternative solutions \\[0pt]
 \bottomrule
\end{tabular}
\end{table}
\FloatBarrier

\section{Generic Decoding for Large Error Weights}\label{sec:generic-decoding-prange}

In the previous section, we focused on decoding for low error weights, specifically when the error weight $\errWeight$ satisfies $\errWeight \leq \dmin - 1$, as erasure decoding is not possible beyond this threshold, according to~\cite[Theorem 13 and Theorem 14]{puchingerGenericDecodingSumRank2022}.
We explored unique decoding for Problem~\ref{prob:uniqueDecoding} up to $\errWeight \leq \dmin - 1$ and going beyond unique decoding is possible for $\errWeight \leq \min\{n-k, \frac{m}{\shotLength}(n-k)\}$ (see~\cite{puchingerGenericDecodingSumRank2022}).

In this section, we introduce a generic decoding algorithm (Algorithm~\ref{alg:prangeSR}) that aims to solve Problem~\ref{prob:SRSDecProblemUniformSyndrome}. The proposed algorithm is inspired by the Prange-like algorithm for the Hamming metric, as presented in~\cite{debris-alazardWaveNewFamily2019}. Our analysis of the algorithm focuses on the asymptotic case $\shots\to\infty$ and its average performance. The following proposition establishes the range of relative weights for which solutions to Problem~\ref{prob:SRSDecProblemUniformSyndrome} can be found efficiently using Algorithm~\ref{alg:prangeSR}.

\begin{algorithm}
\caption{PrangeSumRank}\label{alg:prangeSR}
\SetKwInOut{Input}{Input}\SetKwInOut{Output}{Output}
\Input{$\H\in\Fqm^{\shotLength(\shots-\kappa)\times \shotLength\shots)}$, $\s\in\Fqm^{\shotLength(\shots-\kappa)}$ and $w\in\NN$}
\Output{$\e\H^\top = \s$ and $\SumRankWeight(\e)=w$}
$\maxShotWeight \gets \min\{m,\eta\}$ \;
$\e \gets \bm{0} \in \Fqm^{\shotLength\shots}$ \;
\While{$\SumRankWeight^{(\n)}(\e)\neq \errWeight$}{
    $\H' \gets \bm{0} \in \Fqm^{\shotLength(\shots-\kappa)\times \shotLength\shots}$ \;
    \While{$\rkqm(\H_{[1:\shotLength(\shots-\kappa)]}') \neq \shotLength(\shots-\kappa)$}{
        $\P \sample$ Any $\shots \times \shots$ permutation matrix \; \label{alg:prangeSR:lineA}
        $\P' \gets \P \otimes \I_{\shotLength}$ \;
        $\H' \gets \H \P'$ \; \label{alg:prangeSR:lineB}
    }
    $\A \gets \H_{[1:\shotLength(\shots-\kappa)]}'$ \;
    $\B \gets \H_{[\shotLength(\shots-\kappa)+1:\shotLength\shots]}'$ \;
    $\errWeight_1 \sample \{0,\ldots,\kappa\maxShotWeight\}$ \;
    $\e' \sample \{\x \in \Fqm^{\kappa\shots} \st \SumRankWeight^{[n_{\kappa+1},\ldots,n_\shots]}(\e') = \errWeight_1 \}$ \;
    $\e \gets ((\s-\e'\B)\A^{-\top}, \e') (\P')^\top$ \; \label{alg:prangeSR:lastLine}
}
\Return $\e$
\end{algorithm}

\begin{proposition}
Consider a $\Fqm$-linear sum-rank-metric code of length $n=\shotLength\shots$, dimension $k=\shotLength\kappa$ with parity check matrix $\H\in\Fqm^{(n-k)\times n}$ and $R\defeq\frac{k}{n}=\frac{\kappa}{\shots}$ where $\kappa\in\NN$ and $0 \leq \kappa \leq \shots$. Let the sum-rank weight be defined with respect to the length partition of constant block length, i.e., $\n = [n_1,\ldots,n_\shots] = [\shotLength, \ldots, \shotLength]$. Define
\begin{equation}
    \bar{a} \defeq \frac{\sum_{i=0}^{\maxShotWeight} i \cdot \NMq{q}{m, \shotLength, i}}{q^{m\shotLength}},
\end{equation}
as the average rank weight of a single block if drawn uniformly at random. Then, for the relative weight $\wRel \defeq \errWeight/n$ in the interval $[\wEasym,\wEasyp]$, where
\begin{align}
    \wEasym &\defeq \frac{1-R}{\shotLength} \cdot \bar{a}, \\
    \wEasyp &\defeq \frac{1-R}{\shotLength} \cdot \bar{a} + \frac{R\maxShotWeight}{\shotLength},
\end{align}
a solution to Problem~\ref{prob:SRSDecProblemUniformSyndrome} can be found in probabilistic polynomial time using the Prange-like Algorithm~\ref{alg:prangeSR}.
\end{proposition}

\begin{proof}
To solve Problem~\ref{prob:SRSDecProblemUniformSyndrome}, we want to find, for a given syndrome $\s$, an error $\e$ of sum-rank weight $\errWeight$ such that $\e\H^\top = \s$. The matrix $\H$ is a full-rank matrix and therefore contains an invertible submatrix $\A\in\Fqm^{(n-k)\times(n-k)}$. Without loss of generality, assume that this matrix is formed by the first $n-k=\eta(\shots-\kappa)$ positions, i.e. we have
\begin{equation}\label{eq:syndromeEquationPrange}
    \H = \left[\A \mid \B \right].
\end{equation}
Assume $\e = [\e'', \e'] \in \Fqm^{\shotLength\shots}$ with $\e'' \in \Fqm^{\shotLength(\shots-\kappa)}$ and $\e' \in \Fqm^{\shotLength\kappa}$. Then, by \eqref{eq:syndromeEquationPrange}, we have
\begin{equation}
    \e'' = (\s - \e'\B^\top){(\A^{-1})}^\top.
\end{equation}
The idea is to arbitrarily choose $\e'$ of length $k = \shotLength\kappa$. Then, on average, the expected (partial) sum-rank weight of the remaining $\shots-\kappa$ blocks is
\begin{equation}
    \mathrm{E}\left[\SumRankWeight^{([n_{\kappa+1},\ldots,n_\shots])}(\e'')\right] = \bar{a} \cdot (\shots - \kappa).
\end{equation}
The average probability of the sum-rank weight of $\e$ is then
\begin{equation}
    \mathrm{E}\left[\SumRankWeight^{(\n)}(\e)\right] = \underbrace{\mathrm{E}\left[\SumRankWeight^{([n_{1},\ldots,n_\kappa])}(\e')\right]}_{\eqdef \bar{\errWeight}_1} + \mathrm{E}\left[\SumRankWeight^{([n_{\kappa+1},\ldots,n_\shots])}(\e'')\right] = \bar{\errWeight}_1 + \bar{a} \cdot (\shots - \kappa),
\end{equation}
where $\bar{\errWeight}_1$ is determined by the distribution of $\e'$, which we can choose freely. Nonetheless, we have $0 \leq \bar{\errWeight}_1 \leq \maxShotWeight\kappa$, and therefore
\begin{equation}
    \underbrace{\bar{a} \cdot (\shots - \kappa)}_{= n\wEasym} \leq \mathrm{E}\left[\SumRankWeight^{(\n)}(\e)\right] \leq \underbrace{\maxShotWeight\kappa + \bar{a} \cdot (\shots - \kappa)}_{= n\wEasyp}.
\end{equation}
From this, we deduce that any weight in the interval $\errWeight \in [\wEasym n, \wEasyp n]$ can be reached probabilistically in polynomial time using a distribution for $\e'$ with $\bar{w}_1 = \errWeight - \wEasym n$ such that $\mathrm{E}\left[\SumRankWeight^{(\n)}(\e)\right] = \errWeight$ and which is sufficiently concentrated around its expectation.
Algorithm~\ref{alg:prangeSR} implements this approach, where in Line~\ref{alg:prangeSR:lineA} to Line~\ref{alg:prangeSR:lineB}, the parity-check matrix $\H$ is permuted block-wise among the $\shots$ blocks, i.e., the permutation is applied to the block indices but not within the blocks. Here, $\otimes$ denotes the Kronecker product, which is used to construct the block-wise permutation matrix. This permutation is reversed in Line~\ref{alg:prangeSR:lastLine}.
\end{proof}

The proposition above provides the interval $[\wEasym,\wEasyp]$ for which a solution to Problem~\ref{prob:SRSDecProblemUniformSyndrome} can be found in probabilistic polynomial time using Algorithm~\ref{alg:prangeSR}. Combining this with the Gilbert-Varshamov bound for the sum-rank metric, we have the following summary of the relative weight intervals:
\begin{itemize}
    \item $\wRel \in [\wGVm,\wGVp]$: A solution to Problem~\ref{prob:SRSDecProblemUniformSyndrome} is likely to exist (Gilbert-Varshamov bound for the sum-rank metric, see~\cite{byrneFundamentalPropertiesSumRankMetric2021})
    \item $\wRel \in [\wEasym,\wEasyp]$: A solution to Problem~\ref{prob:SRSDecProblemUniformSyndrome} can be found in probabilistic polynomial time using a Prange-like algorithm, as stated in the proposition above.
\end{itemize}

Figure~\ref{fig:gv_plot_1} and Figure~\ref{fig:gv_plot_2} show the regions of hardness for finding a solution using Algorithm \ref{alg:prangeSR} and the bounds on the relative weight intervals for successful decoding plotted against the code rate $R$ for different parameters. These results apply asymptotically ($\shots \to \infty$, for fixed $m$ and $\maxShotWeight$) and on average. The "no solution", "hard", and "easy" regions indicate the difficulty of finding a solution for different code rates.

In Figure~\ref{fig:gv_plot_1}, the parameters are set to $m = \shotLength = 2$, $q = 2$, while in Figure~\ref{fig:gv_plot_2}, the parameters are $m = \shotLength = 6$, $q = 2$. Comparing the two figures, we observe that the "hard" region for large relative weights becomes smaller as the values of $m$ and $\shotLength$ increase. This indicates that it is easier for Algorithm \ref{alg:prangeSR} to decode errors of large relative weight when $m$ and $\shotLength$ are larger.

\begin{figure}[ht!]
  \begin{center}
    \begin{tikzpicture}
    \begin{axis}[
        width=\linewidth, %
        height=8cm,
        grid=major, 
        grid style={dotted,gray!80},
        xlabel={$R$}, %
        ylabel={$\wRel$},
        xmax=1,
        xmin=0,
        ymax=1,
        ymin=0,
        tick label style={font=\scriptsize},
        legend style={font=\scriptsize, at={(0.98,0.98)}, anchor=north east}, 
        label style={inner sep=0, font=\small},
    ]
        \addplot[blueline, name path=GV] table[x=R,y=w,col sep=comma] {./data/prange_GV\string_SR\string_q2\string_m2\string_eta2.txt}; 
        \addplot[magentaline, name path=wM] table[x=R,y=w,col sep=comma] {./data/prange_wM\string_SR\string_q2\string_m2\string_eta2.txt}; 
        \addplot[redline, name path=wP] table[x=R,y=w,col sep=comma] {./data/prange_wP\string_SR\string_q2\string_m2\string_eta2.txt}; 
        
        \legend{$\wGVm$, $\wEasym$, $\wEasyp$}    
        
        \addplot[hardfill] fill between[of=GV and wM];
        \addplot[easyfill] fill between[of=wM and wP];
        
        \path[name path=axis-top] (axis cs:0,1) -- (axis cs:1,1);
        \addplot[hardfill] fill between[of=wP and axis-top];
        
        \path[name path=axis-bottom] (axis cs:0,0) -- (axis cs:1,0);
        \addplot[nosolutionfill] fill between[of=GV and axis-bottom];

        \node[regiontext] at (0.17,0.06) {no solution};
        \node[regiontext] at (0.2,0.35) {hard};
        \node[regiontext] at (0.5,0.55) {easy};
        \node[regiontext] at (0.2,0.9) {hard};
        
    \end{axis}
\end{tikzpicture}
    \caption{Regions of hardness for Algorithm \ref{alg:prangeSR} and bounds on the relative weight intervals for successful decoding vs code rate $R=k/n$ for parameters: $m = \shotLength = 2$, $q = 2$ ($\shots \to \infty$, average-case).}\label{fig:gv_plot_1}
  \end{center}
\end{figure}
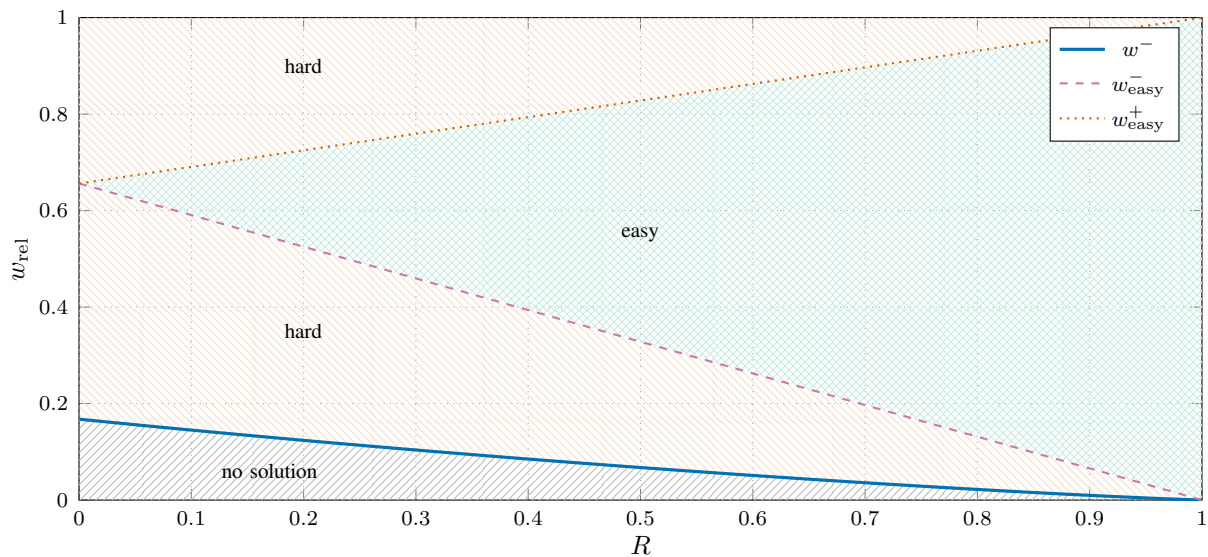
\FloatBarrier

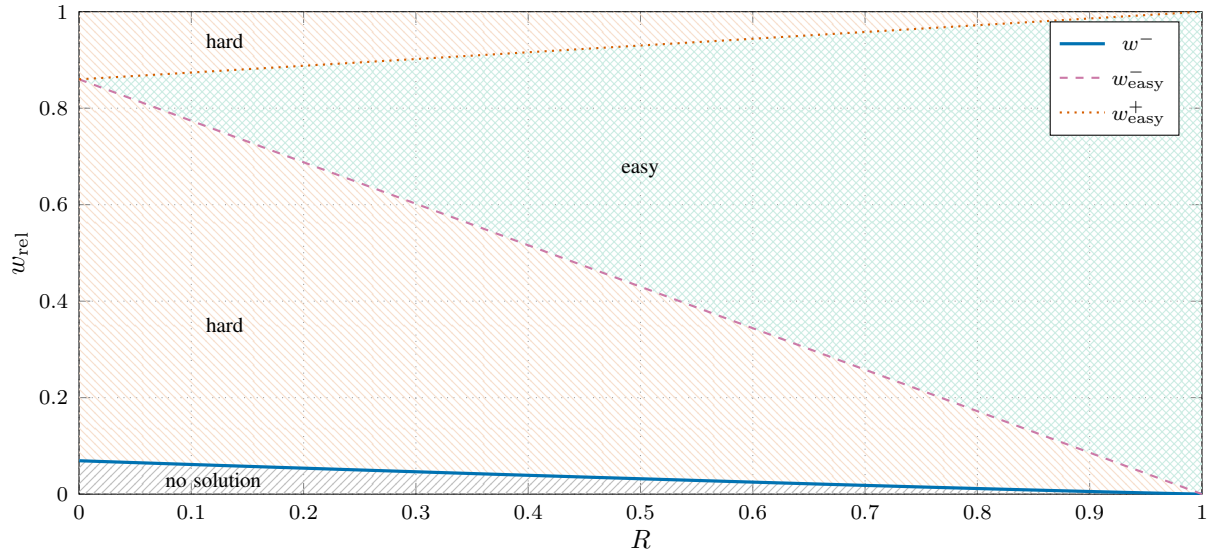
\begin{figure}[ht!]
  \begin{center}
    \begin{tikzpicture}
    \begin{axis}[
        width=\linewidth, %
        height=8cm,
        grid=major, 
        grid style={dotted,gray!80},
        xlabel={$R$}, %
        ylabel={$\wRel$},
        xmax=1,
        xmin=0,
        ymax=1,
        ymin=0,
        tick label style={font=\scriptsize},
        legend style={font=\scriptsize, at={(0.98,0.98)}, anchor=north east}, 
        label style={inner sep=0, font=\small},
    ]
        \addplot[blueline, name path=GV] table[x=R,y=w,col sep=comma] {./data/prange_GV\string_SR\string_q2\string_m6\string_eta6.txt}; 
        \addplot[magentaline, name path=wM] table[x=R,y=w,col sep=comma] {./data/prange_wM\string_SR\string_q2\string_m6\string_eta6.txt}; 
        \addplot[redline, name path=wP] table[x=R,y=w,col sep=comma] {./data/prange_wP\string_SR\string_q2\string_m6\string_eta6.txt}; 
        
        \legend{$\wGVm$, $\wEasym$, $\wEasyp$}    
        
        \addplot[hardfill] fill between[of=GV and wM];
        \addplot[easyfill] fill between[of=wM and wP];
        
        \path[name path=axis-top] (axis cs:0,1) -- (axis cs:1,1);
        \addplot[hardfill] fill between[of=wP and axis-top];
        
        \path[name path=axis-bottom] (axis cs:0,0) -- (axis cs:1,0);
        \addplot[nosolutionfill] fill between[of=GV and axis-bottom];

        \node[regiontext] at (0.12,0.03) {no solution};
        \node[regiontext] at (0.13,0.35) {hard};
        \node[regiontext] at (0.5,0.67) {easy};
        \node[regiontext] at (0.13,0.94) {hard};
        
    \end{axis}
\end{tikzpicture}
    \caption{Regions of hardness for Algorithm \ref{alg:prangeSR} and bounds on the relative weight intervals for successful decoding vs code rate $R=k/n$ for parameters: $m = \shotLength = 6$, $q = 2$ ($\shots \to \infty$, average-case).}\label{fig:gv_plot_2}
  \end{center}
\end{figure}
\FloatBarrier

\section{Randomized Decoding of Linearized Reed--Solomon Codes}\label{sec:random_decoding}

In this section, we consider the probabilistic decoding algorithm introduced in~\cite{jerkovitsRandomizedDecodingLinearized2023} to solve Problem~\ref{prob:LRSDecProblem}. In~\cite{jerkovitsRandomizedDecodingLinearized2023}, the complexity of this decoder was analyzed for the worst-case rank profile. This algorithm generalizes the randomized decoder for Gabidulin codes from~\cite{rennerRandomizedDecodingGabidulin2020} to \ac{LRS} codes.

We revisit the decoder from~\cite{jerkovitsRandomizedDecodingLinearized2023} and present two new contributions:
\begin{itemize}
    \item We adapt the methods from~\cite{puchingerGenericDecodingSumRank2022} to efficiently compute worst-case bounds and sample the support rank profile for the randomized decoder for \ac{LRS} codes.
    \item We extend the average-case analysis from Section~\ref{sec:generic-decoding} to the randomized decoder, deriving an objective function to optimize the support-drawing distribution, motivated by the asymptotic setting.
\end{itemize}

The proposed decoder relies on two key aspects. 
First, we consider an underlying \ac{LRS} code. Recall that \ac{LRS} codes are \ac{MSRD} codes with a minimum sum-rank distance $\dmin = n-k+1$. Efficient algorithms exist to decode \ac{LRS} codes up to the unique decoding radius $\UniqueDecRad = \frac{n-k}{2}$. 
We consider \ac{LRS} codes of length $n$ partitioned into constant block lengths $\n = [n_1,\ldots,n_\shots] = [\eta,\ldots,\eta]$ and dimension $k$ over $\Fqm$, denoted by $\LRSCode$. Additionally, \ac{LRS} codes are restricted to $\shots \leq q-1$ and $n_i \leq m$ for all $i\inshots$.

Second, we focus on decoding beyond the unique decoding radius, where the sum-rank weight of the error $\errWeightVec$ exceeds $\UniqueDecRad$, i.e., $\SumRankWeight(\errWeightVec)=\errWeight > \UniqueDecRad$. The error excess beyond the unique decoding radius is defined as
\begin{equation}\label{eq:error_excess}
    \errEx \defeq \errWeight - \UniqueDecRad.
\end{equation}
Note that while $2\errEx$ is always an integer, $\errEx$ itself does not necessarily need to be an integer.

As discussed in Section~\ref{sec:problem-description}, for errors with weight $\errWeight \leq \UniqueDecRad$, Problem~\ref{prob:LRSDecProblem} has at most one solution. However, for $\errWeight > \UniqueDecRad$, multiple solutions may exist. The number of solutions can vary, being either polynomially or exponentially bounded in terms of the code parameters and depending on the structure of the code. This behavior has been studied for \ac{LRS} codes in~\cite{puchingerBoundsListDecoding2021}. Following the reasoning in~\cite{rennerRandomizedDecodingGabidulin2020}, we analyze the complexity of finding at least one solution. If the code is list-decodable, this process can be repeated to obtain a list of solutions.

\subsection{Erasures in the Sum-Rank Metric}
We consider an error of the form as described in Section~\ref{sec:pre:channel model}. The error $\e$ can be further decomposed into a sum of three types of error vectors
\begin{equation}
\e = \e_{\mathrm{F}} + \e_{\mathrm{R}} + \e_{\mathrm{C}},
\end{equation}
where $\e_{\mathrm{F}}$ represents \emph{full errors}, $\e_{\mathrm{R}}$ represents \emph{row erasures}, and $\e_{\mathrm{C}}$ represents \emph{column erasures}. The sum-rank weights of these error vectors are denoted by $\wF$, $\wR$, and $\wC$, respectively, such that $\SumRankWeight(\e_{\mathrm{F}}) = \wF$, $\SumRankWeight(\e_{\mathrm{R}}) = \wR$, and $\SumRankWeight(\e_{\mathrm{C}}) = \wC$ (see~\cite{hormannErrorErasureDecodingLinearized2022}).

Each of the three error vectors can be decomposed as in~\eqref{eq:errorDecomp}
\begin{align}
\e_\mathrm{F} &= \a_\mathrm{F}\B_\mathrm{F} \quad \text{with} \quad \a_\mathrm{F} \in \Fqm^{\wF} \text{ and } \B_\mathrm{F} \in \Fq^{\wF \times n}, \\
\e_\mathrm{R} &= \a_\mathrm{R}\B_\mathrm{R} \quad \text{with} \quad \a_\mathrm{R} \in \Fqm^{\wR} \text{ and } \B_\mathrm{R} \in \Fq^{\wR \times n}, \\
\e_\mathrm{C} &= \a_\mathrm{C}\B_\mathrm{C} \quad \text{with} \quad \a_\mathrm{C} \in \Fqm^{\wC} \text{ and } \B_\mathrm{C} \in \Fq^{\wC \times n}.
\end{align}
For full errors, neither $\a_{\mathrm{F}}$ nor $\B_{\mathrm{F}}$ are known. For row erasures, $\a_{\mathrm{R}}$ is known but $\B_{\mathrm{R}}$ is unknown. For column erasures, $\a_{\mathrm{C}}$ is unknown but $\B_{\mathrm{C}}$ is known.

An efficient algorithm for \ac{LRS} codes was proposed in~\cite{hormannErrorErasureDecodingLinearized2022}, capable of correcting combinations of \emph{full errors}, \emph{row erasures}, and \emph{column erasures} up to
\begin{equation}\label{eq:deccond}
2\wF + \wC + \wR \leq n-k,
\end{equation}
with a complexity of $\oh{n^2}$ operations over $\Fqm$. We estimate this complexity, denoted as $\WiterEnE$, over $\Fq$ as $\oh{m^2 n^2}$ and approximate it, similar to the arguments in Section~\ref{sec:generic-decoding}, as
\begin{equation}\label{eq:defWiterEnE}
    \WiterEnE \approx m^2 n^2.
\end{equation}

We denote this error-erasure decoder by $\DEC(\y,\a_\mathrm{R},\B_\mathrm{C})$, which takes as input the received word $\y = \c + \e$, along with a basis $\a_\mathrm{R}$ of the column support of $\e_\mathrm{R}$ (row erasures) and/or a basis $\B_\mathrm{C}$ of the row support of $\e_\mathrm{C}$ (column erasures). The decoder outputs a valid codeword $\hat{c}$ if the condition in~\eqref{eq:deccond} is satisfied; otherwise, it returns $\emptyset$.

\subsection{Randomized Decoding Algorithm}\label{sec:randomized decoding algorithm}

We consider an error $\e$ with row support $\RSv$ and column support $\CSv$. Unlike the generic decoding algorithm, where we guess a super support that must completely contain the actual error support to succeed, the randomized approach aims to guess only parts of the error supports and utilizes an error-and-erasure decoder to succeed with a smaller number of guesses. The steps are as follows:

For each block $i \in \shots$, we have $\RS^{(i)} \subseteq \Fq^\shotLength$ and $\CS^{(i)} \subseteq \Fq^m$. As shown in~\cite{rennerRandomizedDecodingGabidulin2020}, for Gabidulin codes, guessing a combination of row and column supports does not improve success, and it is more effective to guess from a smaller ambient space. Since this result applies block-wise in the sum-rank metric, we set $\maxShotWeight = \min\{m, \eta\}$ and guess from $\RSv$ if $\maxShotWeight = \shotLength$, otherwise from $\CSv$. 

This ensures the guessed support is always a subspace of $\Fq^\maxShotWeight$. For simplicity, we use $\ESv$ to denote the error support, whether from $\RSv$ or $\CSv$.

To guess parts of $\ESv$, we first draw a corresponding rank profile $\subwtVec = [\subwt_1,\ldots,\subwt_\shots] \in \NN^\shots$ according to some \ac{PMF} denoted as $\pu{\subwtVec}\defeq\Pr[\subwtVec]$. 
Then, a support $\GSv$ is drawn uniformly at random from $\Xi_{q,\maxShotWeight}(\subwtVec)$ with $u\defeq \sumDim(\GSv)$. 
Define $\sdInter$ as the sum dimension of the intersection space between the guessed space $\GSv$ and the actual error support $\ESv$, i.e.,
\begin{equation}
    \sdInter \defeq \sumDim(\GSv \cap \ESv).
\end{equation}
The number of full errors $\wF$ is reduced by $\sdInter$, so $\wF = \errWeight - \sdInter$, while the number of column or row erasures increases by $\subwt$, corresponding to the guessed parts of $\GSv$. For these guessed parts, we assume knowledge of the column support but not the row support (or vice versa), effectively trading errors for erasures~\footnote{This approach is reminiscent of the generalized minimum distance (GMD) decoding strategy introduced by Forney~\cite{forneyGeneralizedMinimumDistance1966} for Hamming metric codes, and later extended to the rank metric in~\cite{bossertVerfahrenUndKommunikationsvorrichtung2003} by Bossert \textit{et al.}.}.

The error-and-erasure decoder takes as input a vector containing $\errWeight - \sdInter$ full errors and $\subwt$ erasures.
From the decoding condition in~\eqref{eq:deccond}, we have
\begin{equation}\label{eq:deltaBounds}
    2(\errWeight-\sdInter)+\subwt \leq n-k,
\end{equation}
which implies that for successful decoding, we need to have
\begin{equation}\label{eq:minEpsDef}
    \sdInter \geq \minEps \defeq \errWeight + \frac{\subwt - (n-k)}{2} = \errEx + \frac{\subwt}{2}.
\end{equation}
If the intersection between the guessed spaces and the actual error support is sufficiently large, an error-erasure decoder can successfully decode. From~\eqref{eq:minEpsDef}, the valid range for $\subwt$ is $\subwt \in \{2\xi, \ldots, n-k\}$, with the lower bound ensuring $\sdInter \geq 0$ and the upper bound corresponding to the most favorable case, where $\sdInter = \errWeight$. In the latter case, $\subwt \leq n-k$, which is the maximum erasure decoding capability for \ac{LRS} codes.
The performance of the randomized decoder depends on the choice of the \ac{PMF} $\pu{\subwtVec}$ used to draw the rank profile $\subwtVec$. The optimal choice of $\pu{\subwtVec}$ and the analysis of the algorithm’s success probability will be discussed in the following sections.

Algorithm~\ref{alg:randomized_sr_decoder} outlines the proposed approach. It can be easily generalized to variable block lengths and other sum-rank metric codes that support efficient error-and-erasure decoders.

\begin{algorithm}[ht!]
\caption{Randomized Sum-Rank Metric Decoder for \ac{LRS} codes}\label{alg:randomized_sr_decoder}

\Input{
    Parameters: $q$, $m$, $\shotLength$, $\shots$, $\errWeight$ and $\subwt$ with $2\errEx \leq \subwt \leq \sdInterMax$ \\
    \enspace Received vector $\y\in\Fqm^n$\\
    \enspace \Ac{LRS} code $\LRSCode$ of length $n=\shots\eta$ and dimension $k$  \\
    \enspace Error-erasure decoder $\DEC(\cdot,\cdot,\cdot)$ for $\LRSCode$ \\    
}
\Output{
    Vector $\c' \in \LRSCode$ such that $\SumRankWeight(\y-\c') = \errWeight$
}
$\c' \gets \emptyset$ \;
$\maxShotWeight \gets \min\{m,\shotLength\}$ \;
\While{$\c' = \emptyset$ \oror $\SumRankWeight(\y-\c') \neq \errWeight$}{
    $\subwtVec \gets $ Draw rank profile for the guess space according to $\pu{\subwtVec}$ \;
    $\GSv \sample \Xi_{q,\maxShotWeight}(\subwtVec)$\label{alg:randomized_sr_decoder:drawsupport} \;
    \If{$\shotLength < m$}{
        $\B_\mathrm{C} \gets $ Basis of $\GSv$ \;
        $\c' \gets \DEC(\y,\emptyset,\B_\mathrm{C})$ \makebox[3cm][l]{\scriptsize\tcc*{Error-erasure decoding with row erasures}}
    }
    \Else{
        $\a_\mathrm{R} \gets $ Basis of $\GSv$ \;
        $\c' \gets \DEC(\y,\a_\mathrm{R},\emptyset)$ \makebox[3cm][l]{\scriptsize\tcc*{Error-erasure decoding with column erasures}}
    }
}
\Return $\c'$
\end{algorithm}
\FloatBarrier

The following lemma provides a useful result that contributes to the derivation of the average complexity of the randomized approach in Algorithm~\ref{alg:randomized_sr_decoder}.

\begin{lemma}\label{lemma:success probability of randomized given w and u}
    For a fixed error $\e = [\e_1, \ldots, \e_\shots] \in\Fqm^n$  with rank profile $\errWeightVec = [\errWeight_1, \ldots, \errWeight_\shots] \in \NN^\shots$ and a given rank profile $\subwtVec = [\subwt_1,\ldots,\subwt_\shots] \in \NN^\shots$, let $\ESv$ and $\GSv$ be the error space and guessed space, respectively, where $\GSv$ is chosen uniformly from $\Xi_{q,\maxShotWeight}(\subwtVec)$. Further, let $S_j$ denote the event that $\sumDim(\ESv \cap \GSv) = j$.
    The probability of $S_j$ given $\e$ and $\subwtVec$ is then
    \begin{equation}\label{eq:success probability of randomized given w and u}
        \Pr[ S_{j} | \e, \subwtVec] = \Pr[ S_{j} | \errWeightVec, \subwtVec] =  \left(\,\Conv_{i=1}^{\shots} \intersectprob{\maxShotWeight}{\errWeight_i, \subwt_i} \right)(j),
    \end{equation}
    with
    \begin{equation}
     \left(\,\Conv_{i=1}^{\shots} \intersectprob{\maxShotWeight}{\errWeight_i, \subwt_i} \right)(j) \defeq \left( \intersectprob{\maxShotWeight}{\errWeight_1, \subwt_1} \conv \cdots \conv \intersectprob{\maxShotWeight}{\errWeight_\shots, \subwt_\shots} \right) (j),
    \end{equation}
    being the $\shots$-fold discrete convolution (denoted by $\conv$) of the probability distributions $\intersectprob{\maxShotWeight}{\errWeight_i, \subwt_i}$ evaluated at $j$ for all $i=1,\ldots,\shots$. Here, $\intersectprob{\maxShotWeight}{\errWeight_i, \subwt_i}$ represents the probability distribution of the intersection dimension between the error space and the guessed space in the $i$-th shot.
\end{lemma}
\begin{proof}
Given the error $\e = [\e_1, \ldots, \e_\shots] \in\Fqm^n$ with rank profile $\errWeightVec = [\errWeight_1, \ldots, \errWeight_\shots] \in \NN^\shots$ and the rank profile $\subwtVec = [\subwt_1,\ldots,\subwt_\shots] \in \NN^\shots$ of the guessed support, let $V_i$ be a random variable that corresponds to the dimension of the intersection of the $i$-th guessed space $\GS^{(i)}$ with the $i$-th actual error space $\ES^{(i)}$ for $i\inshots$.
By~\eqref{eq:notation intersect prob}, we have that $\intersectprob{\maxShotWeight}{\errWeight_i, \subwt_i}(j)$ is the probability of that event, i.e.,
\begin{equation}
 \Pr[V_i = j| \e, \subwtVec]=\intersectprob{\maxShotWeight}{\errWeight_i, \subwt_i}(j).   
\end{equation}
Note that the probability $\Pr[V_i = j| \e, \subwtVec]$ depends only on the rank weights $\errWeight_i$ and $\subwt_i$ of the error and guessed support in the $i$-th shot, respectively, and not on the specific error vector $\e$. Thus, we can write
\begin{equation}
 \Pr[V_i = j| \e, \subwtVec] = \Pr[V_i = j| \errWeightVec, \subwtVec] = \intersectprob{\maxShotWeight}{\errWeight_i, \subwt_i}(j).   
\end{equation}
Since we are interested in the probability distribution of the sum of random variables, i.e., $V = \sum_{i=1}^{\shots}V_i$, the resulting probability distribution is given by the $\shots$-fold discrete convolution of the probability distributions of the random variables $V_i$ for $i\inshots$. Thus,
\begin{equation}
\Pr[V = j| \e, \subwtVec] = \Pr[V = j| \errWeightVec, \subwtVec] = \left(\,\Conv_{i=1}^{\shots} \intersectprob{\maxShotWeight}{\errWeight_i, \subwt_i}\right)(j),
\end{equation}
with 
\begin{equation}\label{eq:definition of discrete conv}
    \left(\, \intersectprob{\maxShotWeight}{\errWeight_1, \subwt_1} \Conv \intersectprob{\maxShotWeight}{\errWeight_2, \subwt_2}\right)(j) \defeq \sum_{r=-\infty}^{\infty} \intersectprob{\maxShotWeight}{\errWeight_1, \subwt_1}(r) \, \intersectprob{\maxShotWeight}{\errWeight_2, \subwt_2}(j-r).
\end{equation}
Finally, we have that $S_j$ is the event that $V = j$, which proves the claim.
\end{proof}

The probability of the event $S_j$ given the rank profiles of the error and the guessed support, as stated in Lemma~\ref{lemma:success probability of randomized given w and u}, can be further expanded using the concept of rank profiles. The following proposition expresses this probability. %

\begin{proposition}\label{prop:success probability of randomized given w and u as sum}
    Let $\e = [\e_1, \ldots, \e_\shots] \in\Fqm^n$ be a fixed error with rank profile $\errWeightVec = [\errWeight_1, \ldots, \errWeight_\shots] \in \NN^\shots$, and let $\subwtVec = [\subwt_1,\ldots,\subwt_\shots] \in \NN^\shots$ be a given rank profile. Further, let $\maxShotWeight = \min\{\shotLength, m\}$ and $\sdInter \in \{0,\ldots,\shots\maxShotWeight\}$. Then, we can write~\eqref{eq:success probability of randomized given w and u} from Lemma~\ref{lemma:success probability of randomized given w and u} as 
    \begin{equation}\label{eq:success probability of randomized given w and u as sum}
        \Pr[ S_{\sdInter} | \errWeightVec, \subwtVec] = \sum_{\sdInterVec\in\wdecomp{\sdInter,\shots,\maxShotWeight}} \prod_{i=1}^{\shots} \intersectprob{\maxShotWeight}{\errWeight_i, \subwt_i}(\sdInter_i).
    \end{equation}    
\end{proposition}
\begin{proof}
    Starting from the right-hand side of~\eqref{eq:success probability of randomized given w and u as sum}, we have
    \begin{align}
        \sum_{\sdInterVec\in\wdecomp{\sdInter,\shots,\maxShotWeight}} \prod_{i=1}^{\shots} \intersectprob{\maxShotWeight}{\errWeight_i, \subwt_i}(\sdInter_i) &= \sum_{\substack{\maxShotWeight \geq \sdInter_1,\ldots,\sdInter_\shots \geq 0\\\sdInter_1+\cdots+\sdInter_\shots=\sdInter}} \prod_{i=1}^{\shots} \intersectprob{\maxShotWeight}{\errWeight_i, \subwt_i}(\sdInter_i)\\
        &= \left(\,\Conv_{i=1}^{\shots} \intersectprob{\maxShotWeight}{\errWeight_i, \subwt_i} \right)(\sdInter)\\
        &= \Pr[ S_{\sdInter} | \errWeightVec, \subwtVec],
    \end{align}
    where the first equality follows from the definition of the set of rank profiles, the second equality follows from the definition of discrete convolution as defined in~\eqref{eq:definition of discrete conv}, and the last equality follows from the recursive application of the definition of discrete convolution. This completes the proof.
\end{proof}

The probability of successful decoding is given by the sum of the probabilities of the events $S_{\sdInter}$ over all feasible values of the total intersection dimension $\sdInter$. Since these events are mutually exclusive, we define
\begin{equation}\label{eq:prob_mu_total_rand}
 \probProfRand{\maxShotWeight}{\subwtVec,\errWeightVec} \defeq \sum_{\sdInter=\minEps}^{\min\{\subwt,\errWeight\}} \Pr[S_{\sdInter} \mid \subwtVec, \errWeightVec],
\end{equation}
where $\probProfRand{\maxShotWeight}{\subwtVec,\errWeightVec}$ denotes the probability of successful decoding. The lower bound $\minEps$, defined in~\eqref{eq:minEpsDef}, ensures that the intersection support has enough dimensions for successful decoding. The upper bound $\min\{\subwt,\errWeight\}$ reflects that the intersection cannot have a greater dimension than either of the supports.

Given a probability distribution $\pu{\subwtVec}$ over the rank profiles $\subwtVec$, the overall probability of successful decoding for a fixed error rank profile $\errWeightVec$ is
\begin{equation}\label{eq:errPgivenU}
    \probProfavgRand{\maxShotWeight,\subwt}{\errWeightVec} \defeq \sum_{\subwtVec\in\wdecomp{\subwt,\shots,\maxShotWeight}} \pu{\subwtVec} \cdot \probProfRand{\maxShotWeight}{\subwtVec,\errWeightVec}.
\end{equation}

\subsection{Worst-Case Complexity}

When decoding beyond the unique decoding radius for \ac{LRS} codes (Problem~\ref{prob:LRSDecProblem}), multiple solutions may exist, and the proposed decoder (Algorithm~\ref{alg:randomized_sr_decoder}) lacks a mechanism to identify the original codeword among them. To make the complexity analysis exact, we assume a genie-aided version of the decoder that can identify the correct solution and allows us to stop the decoder at the iteration when the original codeword is found. Consequently, only the upper bound on the expected complexity applies to the non-genie-aided decoder, similar to the bounds derived in~\cite{puchingerGenericDecodingSumRank2022} for the generic decoder.

We analyze the worst-case expected complexity over all possible rank profiles of errors with $\SumRankWeight(\e) = \errWeight$. We derive lower and upper bounds on this complexity and provide algorithms to efficiently compute these bounds. Additionally, we present methods to optimize the guessing support distribution for this worst-case scenario.

Determining the optimal sum-rank weight $\subwt \in \{2\errEx, \ldots, \sdInterMax\}$ for the guessing support to minimize the expected complexity of Algorithm~\ref{alg:randomized_sr_decoder} is not straightforward. The success probability bounds from~\cite{rennerRandomizedDecodingGabidulin2020} for the rank metric are convex in $\subwt$, suggesting that the probability is maximized at $\subwt = 2\errEx$ or $\subwt = \sdInterMax$. However, this does not guarantee that these values always yield the optimal expected complexity. Similarly, for the sum-rank metric, we cannot be certain, though these extreme values for $\subwt$ remain important points of interest.
For these two specific values of $\subwt$, we can express the success probability from~\eqref{eq:prob_mu_total_rand} as follows. When $\subwt = 2\errEx \leq \errWeight$, we have
\begin{align}\label{eq:randomized success prob special case for u = n-k}
    \probProfRand{\maxShotWeight}{\subwtVec,\errWeightVec} = \sum_{\sdInter=\subwt}^{\subwt}\Pr[S_{\sdInter} | \subwtVec, \errWeightVec] = \prod_{i=1}^{\shots} \subspaceprob{\maxShotWeight}{\subwt_i}{\errWeight_i}.
\end{align}
On the other hand, when $\subwt = \sdInterMax \geq \errWeight$, we have
\begin{align}\label{eq:randomzied success prob special case for u = 2*xi}
    \probProfRand{\maxShotWeight}{\subwtVec,\errWeightVec} = \sum_{\sdInter=\errWeight}^{\errWeight}\Pr[S_{\sdInter} | \subwtVec, \errWeightVec] = \prod_{i=1}^{\shots} \subspaceprob{\maxShotWeight}{\errWeight_i}{\subwt_i}.
\end{align}

Notably, choosing either $\subwt = 2\errEx$ or $\subwt = \sdInterMax$ simplifies the expression for the success probability, as the convolution in~\eqref{eq:prob_mu_total_rand} reduces to a simple product over the blocks. This simplification is advantageous for our analysis of the decoder's performance.

When $\subwt = \sdInterMax$, the success probability matches that of the fully generic decoding algorithm, as the expressions in~\eqref{eq:randomized success prob special case for u = n-k} and~\eqref{eq:success probability for generic given profile of v and w} are identical. In this case, the bounds from~\cite{puchingerGenericDecodingSumRank2022} apply directly, with the only difference being the per-iteration complexity, which is replaced by that of the error-and-erasure decoder for \ac{LRS} codes. Therefore, setting $\subwt = \sdInterMax$ offers no improvement in success probability over the fully generic approach.

When $\subwt = 2\errEx$ and $\subwt \leq \errWeight$, the algorithms from~\cite{puchingerGenericDecodingSumRank2022} require adjustments. We present these modifications below, using the notation
\begin{align}\label{eq:phiPrime}
    \probProfRandP{\maxShotWeight}{\subwtVec,\errWeightVec} \defeq \prod_{i=1}^{\shots} \subspaceprob{\maxShotWeight}{\subwt_i}{\errWeight_i}.
\end{align}

From~\eqref{eq:randomzied success prob special case for u = 2*xi} follows the worst-case expected number of iterations of Algorithm~\ref{alg:randomized_sr_decoder} for a given \ac{PMF} of the guessing support, denoted as $\puwc{\subwtVec} \defeq \Pr[\subwtVec]$ by
\begin{equation}\label{eq:maximum expect. iterations equation}
    \max_{\errWeightVec\in\wdecomp{\errWeight,\shots,\maxShotWeight}} \mathbb{E}[\text{\#iterations}] = \max_{\errWeightVec\in\wdecomp{\errWeight,\shots,\maxShotWeight}}{\left( \sum_{\subwtVec \in \wdecomp{\subwt,\shots,\maxShotWeight}} \puwc{\subwtVec}\probProfRandP{\maxShotWeight}{\subwtVec,\errWeightVec} \right)}^{-1}.
\end{equation}

The problem of minimizing the worst-case expected number of iterations over all valid distributions $\puwc{\subwtVec}$ on $\wdecomp{\subwt,\shots,\maxShotWeight}$ can be formulated as a linear program. While this linear program can be solved numerically using standard methods for small values of $\maxShotWeight$, $\shots$, and $\subwt$, the number of unknowns, i.e., $\puwc{\subwtVec} \in [0, 1]$, grows rapidly as these parameters increase. Consequently, solving the linear program directly becomes computationally prohibitive for larger problem instances.

To address this computational challenge, we adopt the approach proposed by Puchinger et al.~\cite{puchingerGenericDecodingSumRank2022}, which introduces a randomized mapping $\ucomp_{\maxShotWeight} \colon \wdecomp{\errWeight,\shots,\maxShotWeight} \times \NN \to \wdecomp{\subwt,\shots,\maxShotWeight}$. This mapping aims to maximize the probability $\probProfRandP{\maxShotWeight}{\ucomp_{\maxShotWeight}(\errWeightVec, \subwt), \errWeightVec}$ for a given $\errWeightVec \in \wdecomp{\errWeight,\shots,\maxShotWeight}$ by randomly selecting an output vector from multiple possible candidates for each input, providing a more computationally tractable approach to the problem.

Rather than directly choosing a vector $\subwtVec \in \wdecomp{\subwt,\shots,\maxShotWeight}$, we first select a vector $\errWeightVec \in \wdecomp{\errWeight,\shots,\maxShotWeight}$ at random according to a designed distribution $\designPr{\errWeightVec}$ on $\wdecomp{\errWeight,\shots,\maxShotWeight}$, and then set $\subwtVec \gets \ucomp_{\maxShotWeight}(\errWeightVec, \subwt)$. For a fixed error $\e$, this allows us to bound the probability as follows
\begin{equation}\label{eq:upper_bound_iterations}
    \Pr(\GSv \subseteq \ESv_{\e}) = \sum_{\subwtVec \in \wdecomp{\subwt,\shots,\maxShotWeight}} \puwc{\subwtVec}\cdot \probProfRandP{\maxShotWeight}{\subwtVec,\errWeightVec_{\e}} \geq \designPr{\errWeightVec_{\e}} \cdot \probProfRandP{\maxShotWeight}{\ucomp_{\maxShotWeight}(\errWeightVec_{\e}, \subwt),\errWeightVec_{\e}}.
\end{equation}

Using this bound, we can minimize the following upper bound on the worst-case expected number of iterations, instead of directly minimizing \eqref{eq:maximum expect. iterations equation}
\begin{equation}
    \max_{\errWeightVec\in\wdecomp{\errWeight,\shots,\maxShotWeight}}  \mathbb{E}[\text{\#iterations}] \leq \max_{\errWeightVec\in\wdecomp{\errWeight,\shots,\maxShotWeight}} {\left(\designPr{\errWeightVec} \cdot \probProfRandP{\maxShotWeight}{\ucomp_{\maxShotWeight}(\errWeightVec, \subwt),\errWeightVec}\right)}^{-1},
\end{equation}
over all valid probability mass functions $\puwc{\subwtVec}$ on $\wdecomp{\subwt,\shots,\maxShotWeight}$.

The randomized mapping $\ucomp_\maxShotWeight$ is formally defined in Appendix~\ref{sec:appendix-2} and its correctness is proofed in Lemma~\ref{lem:ucomp_maximizes_prob}.

We adapt the support-drawing algorithm from~\cite{puchingerGenericDecodingSumRank2022} to handle cases where the sum dimension of the guessed support is smaller than that of the error support. The modified version, shown in Algorithm~\ref{alg:DrawRandomSupport}, retains the structure of the original algorithm but is adjusted to account for this dimension difference.
\begin{algorithm}[htbp]
\caption{$\mathrm{DrawRandomSupport}(\subwt, \errWeight, \maxShotWeight)$}\label{alg:DrawRandomSupport}
\SetKwInOut{Input}{Input}\SetKwInOut{Output}{Output}
\Input{Integers $\subwt, \maxShotWeight, \errWeight \in \NN$ with $\subwt \leq \errWeight$}
\Output{$\GSv$ of sum dimension $\subwt$}
Draw $\errWeightVec \in \wdecomp{\errWeight,\shots,\maxShotWeight}$ according to the distribution $\designPr{\errWeightVec}$ defined in \eqref{eq:design_distribution}. \;
$\subwtVec \gets \ucomp_\maxShotWeight(\errWeightVec, \subwt)$ \;
$\GSv \sample \Xi_{\maxShotWeight,\zeta}(\subwtVec)$ \;
\Return $\GSv$
\end{algorithm}

We define the probability distribution $\designPr{\errWeightVec}$ as follows
\begin{equation}
    \designPr{\errWeightVec} \coloneqq {\left(\probProfRandP{\maxShotWeight}{\ucomp_\maxShotWeight(\errWeightVec, \subwt), \errWeightVec}\cdot \qFactorSymRand_{\ell,\errWeight,\maxShotWeight}\right)}^{-1} \quad \forall \errWeightVec \in \wdecomp{\errWeight,\shots,\maxShotWeight},
    \label{eq:design_distribution}
\end{equation}
where $\qFactorSymRand_{\ell,\errWeight,\maxShotWeight}$ is defined as
\begin{equation}
    \qFactorSymRand_{\ell,\errWeight,\maxShotWeight} \coloneqq \sum_{\errWeightVec \in \wdecomp{\errWeight,\shots,\maxShotWeight}} \probProfRandP{\maxShotWeight}{\ucomp_\maxShotWeight(\errWeightVec, \subwt), \errWeightVec}^{-1}.
    \label{eq:q_factor_def}
\end{equation}

The following proposition presents bounds on the expected number of iterations.

\begin{proposition}\label{prop:expected_iterations_bounds}
Let $\e \in \Fqm^{n}$ be an error of sum-rank weight $\errWeight$ and let $\subwt$ be an integer with $\subwt \leq \errWeight$. If $\GSv$ is a sub-support that is drawn by Algorithm \ref{alg:DrawRandomSupport} with input $\subwt$ and $\errWeight$, then we have
\begin{equation}
    |\wdecomp{\errWeight,\shots,\maxShotWeight}|^{-1}\qFactorSymRand_{\ell,\errWeight,\maxShotWeight} \leq {\Pr(\GSv \subseteq \ESv_{\e})}^{-1} \leq \qFactorSymRand_{\ell,\errWeight,\maxShotWeight},
\end{equation}
where $\qFactorSymRand_{\ell,\errWeight,\maxShotWeight}$ is defined as in \eqref{eq:q_factor_def}, and $\ESv_{\e}$ denotes the error support corresponding to the error vector $\e$.
\end{proposition}
\begin{proof}
Denote by $\designPr{\errWeightVec}$ the distribution of $\errWeightVec = \ucomp_\maxShotWeight(\ESv_{\e})$, where $\errWeightVec$ is a random variable with probability mass function $\designPr{\errWeightVec}$ as defined in \eqref{eq:design_distribution}. By the law of total probability, we have
\begin{align*}
    \Pr( \GSv\subseteq \ESv_{\e} ) &= \sum_{\errWeightVec \in \wdecomp{\errWeight,\shots,\maxShotWeight}} \designPr{\errWeightVec}\cdot \probProfRandP{\maxShotWeight}{\ucomp_\maxShotWeight(\errWeightVec, \subwt), \errWeightVec_{\e}} \\
    &\geq \designPr{\errWeightVec_{\e}}\cdot \probProfRandP{\maxShotWeight}{\ucomp_\maxShotWeight(\errWeightVec_{\e}, \subwt), \errWeightVec_{\e}} \\
    &= \qFactorSymRand_{\ell,\errWeight,\maxShotWeight}^{-1},
\end{align*}
where the last equality follows from the definition of $\designPr{\errWeightVec}$ in \eqref{eq:design_distribution}. This proves the upper bound on $\Pr(\GSv \subseteq \ESv_{\e})^{-1}$. For the lower bound, we observe that for all $\errWeightVec \in \wdecomp{\errWeight,\shots,\maxShotWeight}$
\[
     \probProfRandP{\maxShotWeight}{\ucomp_\maxShotWeight(\errWeightVec, \subwt), \errWeightVec_{\e}} \leq \probProfRandP{\maxShotWeight}{\ucomp_\maxShotWeight(\errWeightVec, \subwt), \errWeightVec},
\]
which yields
\begin{align*}
    \Pr(\GSv \subseteq \ESv_{\e} ) &= \sum_{\errWeightVec \in \wdecomp{\errWeight,\shots,\maxShotWeight}} \designPr{\errWeightVec}\cdot \probProfRandP{\maxShotWeight}{\ucomp_\maxShotWeight(\errWeightVec, \subwt), \errWeightVec_{\e}} \\
    &\leq \sum_{\errWeightVec \in \wdecomp{\errWeight,\shots,\maxShotWeight}} \designPr{\errWeightVec}\cdot \probProfRandP{\maxShotWeight}{\ucomp_\maxShotWeight(\errWeightVec, \subwt), \errWeightVec} \\
    &= \qFactorSymRand_{\ell,\errWeight,\maxShotWeight},
\end{align*}
where the last equality follows from the definitions of $\qFactorSymRand_{\ell,\errWeight,\maxShotWeight}$ in \eqref{eq:q_factor_def} and the design distribution $\designPr{\errWeightVec}$ in \eqref{eq:design_distribution}, which proves the claim.
\end{proof}

Using Proposition \ref{prop:expected_iterations_bounds}, we can formulate the following theorem about the expected runtime of the genie-aided version of Algorithm \ref{alg:randomized_sr_decoder}.

\begin{theorem}\label{thm:expected_runtime_bounds_randomized}
Consider a genie-aided version of Algorithm \ref{alg:randomized_sr_decoder} for an \ac{LRS} code $\LRSCode$ of length $n$, dimension $k$, and length partition $\n$ with constant block length $n_i = \shotLength$ for all $i \inshots$. Let $\e \in \Fqm^{n}$ be an error of sum-rank weight $\UniqueDecRad < \errWeight \leq n-k$, and let $\c \in \LRSCode$ be a codeword. We consider the success event of the algorithm returning the originally transmitted codeword $\c$ when given input $\y = \e + \c$ and parameter $\subwt$ with $\subwt = 2\errEx$.

Each iteration of Algorithm \ref{alg:randomized_sr_decoder} costs $\WiterRand$. By including also the expected number of iterations, we can bound the overall expected runtime $\Wrand$ of the genie-aided version of Algorithm \ref{alg:randomized_sr_decoder} by
\begin{equation*}
    \WrandLB \leq \Wrand \leq \WrandUB,
\end{equation*}
where, for $\maxShotWeight = \min\{\shotLength, m\}$, we define (see \eqref{eq:q_factor_def} for $\qFactorSymRand_{\ell,\errWeight,\maxShotWeight}$)
\begin{align*}
    \WrandLB &\defeq |\wdecomp{\errWeight,\shots,\maxShotWeight}|^{-1}\cdot\qFactorSymRand_{\ell,\errWeight,\maxShotWeight},  \\
    \WrandUB &\defeq \WiterRand\cdot\qFactorSymRand_{\ell,\errWeight,\maxShotWeight}.
\end{align*}
\end{theorem}
\begin{proof}
The bounds follow directly from Proposition \ref{prop:expected_iterations_bounds} by multiplying the cost of a single iteration $\WiterWC$ by the expected number of iterations:
\begin{align*}
    \WrandLB &\defeq |\wdecomp{\errWeight,\shots,\maxShotWeight}|^{-1}\cdot\qFactorSymRand_{\ell,\errWeight,\maxShotWeight} \leq \Wrand \\
    &\leq \WiterWC\cdot\qFactorSymRand_{\ell,\errWeight,\maxShotWeight} \eqdef \WrandUB.
\end{align*}

\end{proof}

\begin{remark}
The complexity of one iteration $\WiterRand$ in Theorem~\ref{thm:expected_runtime_bounds_randomized} is determined by two main components. First, the support drawing algorithm from~\cite{puchingerGenericDecodingSumRank2022} can be easily adapted to our case, which yields a complexity of $\ComplSupportDrawingHeu$ bit operations. 

Second, the overall complexity $\WiterRand$ is then the sum of the complexity of the support drawing algorithm and the complexity of the error and erasure decoder, which is of the order of $\oh{n^2 m^2}$ operations over $\Fq$ (see~\eqref{eq:defWiterEnE}). Similar to the complexity analysis for the worst-case scenario in the generic decoding algorithm, we approximate $\WiterRand$ with the two dominating terms in each of the complexities as $\WiterRand \approx n^3 m^2$ when we plot or evaluate the complexities.
\end{remark}

To evaluate the bounds from Theorem \ref{thm:expected_runtime_bounds_randomized}, we must compute $\qFactorSymRand_{\ell,\errWeight,\maxShotWeight}$. Direct computation using \eqref{eq:q_factor_def} is infeasible, as the number of summands $|\wdecomp{\errWeight,\shots,\maxShotWeight}|$ can grow super-polynomially with $\errWeight$, depending on $\shots$ and $\maxShotWeight$. In Appendix~\ref{sec:appendix-3}, we show, following \cite[Lemma 22]{puchingerGenericDecodingSumRank2022}, how to compute this efficiently in $\softoh{\errWeight\subwt n^3 \maxShotWeight^3 \log_2(q)}$ bit operations.

\subsection{Average Complexity}\label{sec:average complexity for randomized decoder}

We analyze the average complexity of the randomized decoding algorithm over all possible error vectors $\e$ with sum-rank weight $\errWeight$, similar to the analysis for the generic decoder in Section~\ref{sec:generic_decoding_beyond}.
Although Algorithm~\ref{alg:randomized_sr_decoder} operates on an \ac{LRS} code with significant structure, we use random coding arguments akin to the generic decoding approach to estimate the average success probability when decoding beyond the unique decoding radius. This accounts for the additional codeword solutions that may appear in this regime. Note that since \ac{LRS} codes are not random codes, applying random coding arguments only yields an approximation for the lower bound.

Adapting Theorem~\ref{thm:success_probability_bounds}, we replace the expression for the success probability of one iteration of the decoding loop to obtain bounds for our randomized decoder.

\begin{corollary}
\label{cor:success_probability_bounds_rand}
Let $\mycode{C}$ be a random $\Fqm$-linear code of length $n$ and dimension $k$ over $\Fqm$, where each codeword is drawn uniformly at random from $\Fqm^n$. Suppose the received word $\y = \c + \e$, where $\c \in \mycode{C}$ and $\e \in \Fqm^n$ with $\SumRankWeight(\e) = \errWeight$. Assume we have an error-and-erasure decoder that can correct combinations of errors and erasures up to the condition in \eqref{eq:deccond}. Then, the success probability of Algorithm~\ref{alg:randomized_sr_decoder} to output at least one solution satisfies:
\begin{equation}
    \Pr[\text{success}] \geq \frac{1}{|\errorSet_{q,\shotLength,m,\shots}(\errWeight)|} \sum_{\subwt'=\errWeight}^{\smax} \sum_{\errWeightVec \in \wdecomp{\errWeight,\shots,\maxShotWeight}} \probProfavgRand{q, \maxShotWeight,\subwt'}{\errWeightVec} \prod_{i=1}^{\shots} \NMq{q}{m,\shotLength,\errWeight_i},
\end{equation}
and
\begin{equation}
    \Pr[\text{success}] \leq \left( \frac{1}{|\errorSet_{q,\shotLength,m,\shots}(\errWeight)|} + q^{m(k-n)} \right) \sum_{\subwt'=\errWeight}^{\smax} \sum_{\errWeightVec \in \wdecomp{\errWeight,\shots,\maxShotWeight}} \probProfavgRand{q, \maxShotWeight,\subwt'}{\errWeightVec} \prod_{i=1}^{\shots} \NMq{q}{m,\shotLength,\errWeight_i}.
\end{equation}
\end{corollary}
\begin{proof}
    The proof follows the same steps as in Theorem~\ref{thm:success_probability_bounds}, replacing $\probProfavg{q,\maxShotWeight, \supwt'}{\errWeightVec}$ with $\probProfavgRand{q, \maxShotWeight,\subwt'}{\errWeightVec}$.
\end{proof}

\subsection{Optimizing the Support-Drawing Distribution in Randomized Decoding of \ac{LRS} Codes}\label{sec:heuristic approach for randomized decoder}

We propose a heuristic approach to optimize the support-drawing distribution used in Algorithm~\ref{alg:randomized_sr_decoder}. Inspired by the asymptotic analysis in Section~\ref{sec:new heuristic approach for generic decoding for rcu bound}, we aim at maximizing the average intersection between the guessed support and the actual error support, thereby increasing the probability of successful decoding.

To simplify the optimization, we consider an asymptotic setting where the number of blocks $\shots$ tends to infinity. In this context, we approximate the rank weight distributions of the error and the guessed support by their marginal distributions for a single block, assuming they are independently and identically distributed across blocks.

Let $\pum{\subwt} \defeq \Pr[\dim(\GS^{(i)}) = \subwt]$ denote the marginal distribution of the rank weight of the guessed support $\GS^{(i)}$ for a single block $i$, where $\subwt \in \{0, \dots, \maxShotWeight\}$. The joint distribution of the guessing rank profile $\subwtVec$ is then
\begin{equation}\label{eq:rand_sup_pmf_iid}
    \pu{\subwtVec} = \prod_{i=1}^{\shots} \pum{\subwt_i}.
\end{equation}

We introduce the random variable $Z$, representing the dimension of the intersection between the guessed support and the error support for a single block. Our objective is to maximize the expected value $\mathbb{E}[Z]$, as a larger average intersection increases the probability of successful decoding with an error-and-erasure decoder, which requires a sufficiently large intersection dimension (see~\eqref{eq:minEpsDef}). The expectation $\mathbb{E}[Z]$ can be computed as
\begin{equation}\label{eq:maximize_term_rand}
    \mathbb{E}[Z] = \sum_{\epsilon=0}^{\maxShotWeight} \sum_{\errWeight'=0}^{\maxShotWeight} \sum_{\subwt'=0}^{\maxShotWeight} \pum{\subwt} \cdot \Pr[\errWeight'] \cdot \intersectprob{\maxShotWeight}{\errWeight', \subwt}(\epsilon),
\end{equation}
where $\Pr[\errWeight']$ is the marginal distribution of the error rank weight for a single block.

The distribution $\pumvec$ can be optimized using \ac{LP} methods to maximize $\mathbb{E}[Z]$, similar to the approach in Section \ref{sec:new heuristic approach for generic decoding for rcu bound}. 

The following theorem provides bounds on the expected runtime of the randomized sum-rank decoder (Algorithm~\ref{alg:randomized_sr_decoder}) based on the success probability bounds derived in Corollary~\ref{cor:success_probability_bounds_rand}.
\begin{theorem}
\label{thm:complexity_bounds_rand}
Under the same assumptions as in Corollary~\ref{cor:success_probability_bounds_rand}, the overall expected runtime $\WrandRCU$ of Algorithm~\ref{alg:randomized_sr_decoder} to output at least one solution is bounded by
\begin{equation}
    \WrandRCULB \leq \WrandRCU \leq \WrandRCUUB,
\end{equation}
with
\begin{equation}\label{eq:WrandRCULB}
    \WrandRCULB \defeq \WiterEnE \left(\left( \frac{1}{|\errorSet_{q,\shotLength,m,\shots}(\errWeight)|} + q^{m(k-n)} \right) \sum_{\subwt=\errWeight}^{\smax} C_{q, m, \shotLength}(\shots, \errWeight, \subwt)\right)^{-1},
\end{equation}
and
\begin{equation}\label{eq:WrandRCUUB}
    \WrandRCUUB \defeq \WiterEnE \left(\frac{1}{|\errorSet_{q,\shotLength,m,\shots}(\errWeight)|}\sum_{\subwt=\errWeight}^{\smax} C_{q, m, \shotLength}(\shots, \errWeight, \subwt)\right)^{-1},
\end{equation}
where
\begin{equation}\label{eq:defC}
    C_{q, m, \shotLength}(\shots, \errWeight, \subwt) \defeq \sum_{\sdInter=\minEps}^{\min\{\subwt,\errWeight\}} \sum_{\errWeightVec \in \wdecomp{\errWeight,\shots,\maxShotWeight}} \sum_{\subwtVec\in\wdecomp{\subwt,\shots,\maxShotWeight}}\sum_{\sdInterVec\in\wdecomp{\sdInter,\shots,\maxShotWeight}} \prod_{i=1}^{\shots} \pum{\subwt_i} \cdot \intersectprob{\maxShotWeight}{\errWeight_i, \subwt_i}(\sdInter_i) \cdot \NMq{q}{m,\shotLength,\errWeight_i},
\end{equation}
and $\WiterEnE\in\oh{n^2 m^2}$ over $\Fq$ denotes the complexity of the error and erasure decoder from~\cite{hormannErrorErasureDecodingLinearized2022} as defined in Section~\ref{sec:randomized decoding algorithm}.
\end{theorem}
\begin{proof}
    The proof follows similar arguments as in Theorem~\ref{thm:complexity_bounds}. The main difference is that the bounds on the success probability are replaced by the expressions derived in Corollary~\ref{cor:success_probability_bounds_rand}, which involve the definitions from~\eqref{eq:success probability of randomized given w and u as sum}, \eqref{eq:prob_mu_total_rand}, \eqref{eq:errPgivenU}, and \eqref{eq:rand_sup_pmf_iid}. The complexity of one iteration of Algorithm~\ref{alg:randomized_sr_decoder} is given by $\WiterEnE$, which is the complexity of the error and erasure decoder used in the randomized algorithm.
\end{proof}

\begin{remark}
    The function $C_{q, m, \shotLength}(\shots, \errWeight, \subwt)$ plays a crucial role in determining the complexity bounds of the randomized sum-rank syndrome decoder. It can be computed efficiently using a dynamic programming routine similar to Algorithm~\ref{alg:DP_gen} in polynomial time.
\end{remark}

\section{Numerical Results}\label{sec:numerical_results}

We compare the performance of the randomized decoding algorithm for \ac{LRS} codes with the generic decoder . Figures~\ref{fig:complexity_randomized_comparison_1} and~\ref{fig:complexity_randomized_comparison_2} illustrate the expected complexities for both algorithms under two different parameter sets, ensuring that the total number of bits, calculated as $m \log_2(q) = 144$, remains constant.

In both figures, we set $n=48$ and $k=24$, resulting in a minimum sum-rank distance $d_{\text{min}} = 25$, and a unique decoding radius $\UniqueDecRad = \left\lfloor \frac{n - k}{2} \right\rfloor = 12$. We consider errors with sum-rank weight $\errWeight = \UniqueDecRad + \errEx = 13$, where the error excess is $\errEx = 1$.

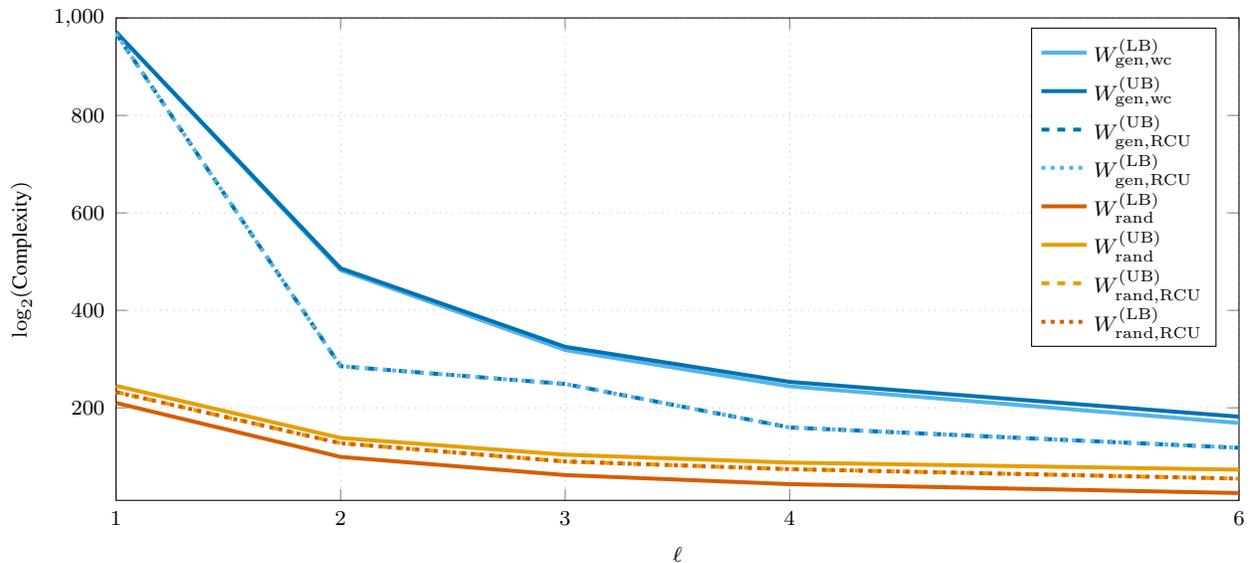
\begin{figure}[htbp]
  \centering
    \begin{tikzpicture}
  \begin{axis}[
      width=\linewidth,
      height=8cm,
      grid=major, 
      grid style={dotted,gray!50},
      xlabel={$\shots$},
      ylabel={$\log_2(\text{Complexity})$},
      xlabel style={font=\scriptsize},
      ylabel style={font=\scriptsize},
      xmin=1,
      xmax=6,
      ymin=10,
      ymax=1000,
      xtick={1,2,3,4,6}, %
      xticklabel style={font=\scriptsize},
      yticklabel style={font=\scriptsize},
      legend style={
          at={(0.98,0.98)}, 
          anchor=north east, 
          legend columns=1, 
          font=\scriptsize
      },
      legend cell align=left,
    ]

    \addplot[WgenLB] 
    table[x=L, y=WF, col sep=comma] {data/results/W_QwcLB_gen_N48_K24_q8_m48_t13_s24.txt};
    \addlegendentry{$\WgenLB$}
    
    \addplot[WgenUB]
    table[x=L, y=WF, col sep=comma] {data/results/W_QwcUB_gen_N48_K24_q8_m48_t13_s24.txt};
    \addlegendentry{$\WgenUB$}
    
    \addplot[WgenRCUUB]
    table[x=L, y=WF, col sep=comma] {data/results/W_heuvecavg_gen_N48_K24_q8_m48_t13_s24.txt};
    \addlegendentry{$\WgenRCUUB$}
        
    \addplot[WgenRCULB]
    table[x=L, y=WF, col sep=comma] {data/results/W_rcveclb_gen_N48_K24_q8_m48_t13_s24.txt};
    \addlegendentry{$\WgenRCULB$}

    \addplot[WrandLB] 
    table[x=L, y=WF, col sep=comma] {data/results/W_QwcLB_rand_N48_K24_q8_m48_t13_u2.txt};
    \addlegendentry{$\WrandLB$}
    
    \addplot[WrandUB]
    table[x=L, y=WF, col sep=comma] {data/results/W_QwcUB_rand_N48_K24_q8_m48_t13_u2.txt};
    \addlegendentry{$\WrandUB$}
    
    \addplot[WrandRCUUB]
    table[x=L, y=WF, col sep=comma] {data/results/WF_heuvecavg_rand_N48_K24_q8_m48_t13_u2.txt};
    \addlegendentry{$\WrandRCUUB$}
        
    \addplot[WrandRCULB]
    table[x=L, y=WF, col sep=comma] {data/results/W_rcveclb_rand_N48_K24_q8_m48_t13_u2.txt};
    \addlegendentry{$\WrandRCULB$}

  \end{axis}    
\end{tikzpicture}
    \vspace{-3em}
    \caption{Complexity comparison of generic decoding vs. randomized decoding beyond the unique decoding radius for parameters: $q=2^3$, $m=48$, $n=48$, $k=24$, $\errWeight=13$, $u=2$, $v=24$.}
    \label{fig:complexity_randomized_comparison_1}
\end{figure}

\begin{figure}[htbp]
  \centering
    \begin{tikzpicture}
  \begin{axis}[
      width=\linewidth,
      height=8cm,
      grid=major, 
      grid style={dotted,gray!50},
      xlabel={$\shots$},
      ylabel={$\log_2(\text{Complexity})$},
      xlabel style={font=\scriptsize},
      ylabel style={font=\scriptsize},
      xmin=2,
      xmax=48,
      ymin=10,
      ymax=1000,
      xtick={1,2,3,4,6,8,12,16,24,48},
      xticklabel style={font=\scriptsize},
      yticklabel style={font=\scriptsize},
      legend style={
          at={(0.99,0.98)}, 
          anchor=north east, 
          legend columns=1, 
          font=\scriptsize
      },
      legend cell align=left,
    ]

    \addplot[WgenLB] 
    table[x=L, y=WF, col sep=comma] {data/results/W_QwcLB_gen_N48_K24_q64_m24_t13_s24.txt};
    \addlegendentry{$\WgenLB$}

    \addplot[WgenUB]
    table[x=L, y=WF, col sep=comma] {data/results/W_QwcUB_gen_N48_K24_q64_m24_t13_s24.txt};
    \addlegendentry{$\WgenUB$}

    \addplot[WgenRCUUB]
    table[x=L, y=WF, col sep=comma] {data/results/W_heuvecavg_gen_N48_K24_q64_m24_t13_s24.txt};
    \addlegendentry{$\WgenRCUUB$}

    \addplot[WgenRCULB]
    table[x=L, y=WF, col sep=comma] {data/results/W_rcveclb_gen_N48_K24_q64_m24_t13_s24.txt};
    \addlegendentry{$\WgenRCULB$}

    \addplot[WrandLB] 
    table[x=L, y=WF, col sep=comma] {data/results/W_QwcLB_rand_N48_K24_q64_m24_t13_u2.txt};
    \addlegendentry{$\WrandLB$}

    \addplot[WrandUB]
    table[x=L, y=WF, col sep=comma] {data/results/W_QwcUB_rand_N48_K24_q64_m24_t13_u2.txt};
    \addlegendentry{$\WrandUB$}

    \addplot[WrandRCUUB]
    table[x=L, y=WF, col sep=comma] {data/results/WF_heuvecavg_rand_N48_K24_q64_m24_t13_u2.txt};
    \addlegendentry{$\WrandRCUUB$}

    \addplot[WrandRCULB]
    table[x=L, y=WF, col sep=comma] {data/results/W_rcveclb_rand_N48_K24_q64_m24_t13_u2.txt};
    \addlegendentry{$\WrandRCULB$}

    \addplot[Wprange]
        coordinates {(48, 46.3327412190881)};
    \addlegendentry{$W_\text{Prange}$}

    \node (tr) at (axis cs:48,140) {}; %
    \node (br) at (axis cs:16,140) {};  %

    \draw[black, thick, dotted] (axis cs:16,10) rectangle (axis cs:48,140);

  \end{axis}

  \begin{scope}[shift={(4cm,2cm)}] %
    \begin{axis}[
        width=7.5cm,
        height=5cm,
        grid=major,
        grid style={dotted,gray!50},
        xlabel={$\shots$},
        ylabel={$\log_2(\text{Complexity})$},
        xlabel style={font=\scriptsize},
        ylabel style={font=\scriptsize},
        xmin=16, xmax=48,
        ymin=0, ymax=140,
        xtick={1,2,3,4,6,8,12,16,24,48},
        ytick={0,20,40,60,80,100, 120, 140},
        xticklabel style={font=\scriptsize},
        yticklabel style={font=\scriptsize},
        legend style={draw=none}, %
        legend cell align=left,
        clip=true, %
    ]

      \addplot[WgenLB] 
      table[x=L, y=WF, col sep=comma] {data/results/W_QwcLB_gen_N48_K24_q64_m24_t13_s24.txt};

      \addplot[WgenUB]
      table[x=L, y=WF, col sep=comma] {data/results/W_QwcUB_gen_N48_K24_q64_m24_t13_s24.txt};

      \addplot[WgenRCUUB]
      table[x=L, y=WF, col sep=comma] {data/results/W_heuvecavg_gen_N48_K24_q64_m24_t13_s24.txt};

      \addplot[WgenRCULB]
      table[x=L, y=WF, col sep=comma] {data/results/W_rcveclb_gen_N48_K24_q64_m24_t13_s24.txt};

      \addplot[WrandLB] 
      table[x=L, y=WF, col sep=comma] {data/results/W_QwcLB_rand_N48_K24_q64_m24_t13_u2.txt};

      \addplot[WrandUB]
      table[x=L, y=WF, col sep=comma] {data/results/W_QwcUB_rand_N48_K24_q64_m24_t13_u2.txt};

      \addplot[WrandRCUUB]
      table[x=L, y=WF, col sep=comma] {data/results/WF_heuvecavg_rand_N48_K24_q64_m24_t13_u2.txt};

      \addplot[WrandRCULB]
      table[x=L, y=WF, col sep=comma] {data/results/W_rcveclb_rand_N48_K24_q64_m24_t13_u2.txt};

        \addplot[Wprange]
        coordinates {(48, 46.3327412190881)};

      \node (inset_tr) at (axis cs:48,0) {}; %
      \node (inset_br) at (axis cs:16,0) {};  %

    \end{axis}
  \end{scope}

  \draw[dotted, thick] 
      (tr) -- (inset_tr); %

  \draw[dotted, thick] 
      (br) -- (inset_br); %

\end{tikzpicture}
    \vspace{-3em}
    \caption{Complexity comparison of generic decoding vs. randomized decoding beyond the unique decoding radius for parameters: $q=2^6$, $m=24$, $n=48$, $k=24$, $\errWeight=13$, $u=2$, $v=24$.}
    \label{fig:complexity_randomized_comparison_2}
\end{figure}
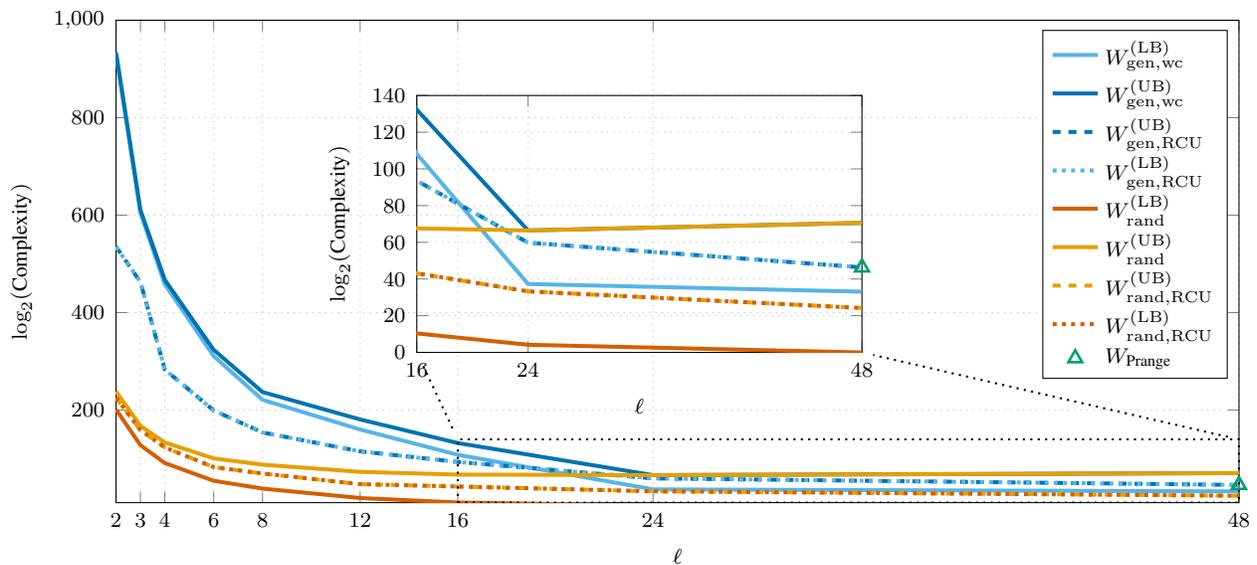

The results show a significant reduction in complexity for the randomized decoding algorithm compared to the generic decoder across both parameter sets. This improvement highlights the advantage of leveraging the structural properties of \ac{LRS} codes. Although the relative gain decreases as the parameters approach those of the Hamming metric, the randomized decoder still maintains a lower complexity in both the worst-case and average-case settings.

\section{Conclusion and Outlook}\label{sec:conclusion}
In this paper, we developed and analyzed algorithms to address the general sum-rank metric decoding problem, with a focus on both worst-case and average-case complexities.
We derived a tighter upper bound for the results in~\cite{puchingerGenericDecodingSumRank2022} and extended their work to the average-case scenario, with a particular focus on decoding beyond the unique decoding radius. Additionally, we improved the randomized decoding algorithm for \ac{LRS} codes, building on our previous work~\cite{jerkovitsRandomizedDecodingLinearized2023}. Furthermore, we introduced a Prange-like algorithm for the sum-rank metric that effectively handles larger error weights in the asymptotic setting, where $\shots \to \infty$.

Future research could adapt techniques that improved Prange's original \ac{ISD} algorithm in the Hamming metric~\cite{sternMethodFindingCodewords1989, bernsteinSmallerDecodingExponents2011, mayDecodingRandomLinear2011, beckerDecodingRandomBinary2012} to the generic and randomized decoding algorithms in the sum-rank metric. Given the hybrid nature of the sum-rank metric, these methods may particularly benefit the Hamming-like error structure and reduce complexity, especially when $\shots$ is large.

Another direction is to apply improvements from generic decoding algorithms in the rank metric, such as those in~\cite{aragonImprovementGenericAttacks2017}, or extend algebraic techniques like~\cite{bardetImprovementsAlgebraicAttacks2020} to the sum-rank metric.

These approaches could yield substantial complexity reductions for specific parameter regimes in the sum-rank metric.

A list decoding algorithm for Gabidulin codes based on Gröbner bases was introduced in~\cite{horlemannModuleMinimizationApproach2017}, enabling error correction beyond the unique decoding radius. This approach can be easily adapted for \ac{LRS} codes. However, as no upper bound on the list size is known, it is difficult to assess the overall complexity of the algorithm, making it challenging to compare with our approach. Establishing tighter complexity bounds for this algorithm remains an open problem and could be a promising direction for future research.

\appendices

\section{Efficient computation of $B_{q, m, \shotLength}(\errWeight, \supwt, \shots)$}\label{sec:appendix-1}

The quantity $B_{q, m, \shotLength}(\errWeight, \supwt, \shots)$ defined in~\eqref{eq:B_def} can also be computed recursively as follows
\begin{equation}
    B_{q, m, \shotLength}(\errWeight, \supwt, \shots) = 
    \begin{cases}
        \psm{\supwt} \NMq{q}{m, \shotLength, \errWeight} \subspaceprob{\maxShotWeight}{\errWeight}{\supwt} & \text{if } \shots = 1 \\
        \\[-10pt]
        \begin{aligned}
            \sum_{\errWeight'=0}^{\min\{\maxShotWeight,\errWeight\}}\sum_{\supwt'=\errWeight'}^{\min\{\maxShotWeight,\supwt\}} 
            & \psm{\supwt'} \NMq{q}{m, \shotLength, \errWeight'} \subspaceprob{\maxShotWeight}{\errWeight'}{\supwt'} \\
            & \cdot B_{q, m, \shotLength}(\errWeight-\errWeight', \supwt-\supwt', \shots-1)
        \end{aligned} & \text{else}
    \end{cases}.
\end{equation}
This expression can be computed in polynomial time using dynamic programming; see Algorithm~\ref{alg:DP_gen}.
\begin{algorithm}
\caption{Compute the term $B_{q, m, \shotLength}(\errWeight, \supwt, \shots)$ in the sum~\eqref{eq:maxP_iid} for a given $\supwt$.}\label{alg:DP_gen}
\Input{Parameters: $q$, $m$, $n$, $k$, $\shots$, $\errWeight$ and $\supwt$ with $\supwt \geq 0$}
\Output{Value of $B_{q, m, \shotLength}(\errWeight, \supwt, \shots)$}
\Initialize{$N(\supwt', \errWeight', \shots') = 0 \quad \forall \errWeight'\in\{0,\ldots,\errWeight\}, \; \supwt'\in\{0,\ldots,\supwt\}, \; \shots'\in\{0,\ldots,\shots\}$}
\If{$\supwt < \errWeight$}{
    \Return $0$
}
    \For{$\errWeight'\in\{0,\ldots,\errWeight\}$}{
        \For{$\supwt'\in\{\errWeight',\ldots,\supwt\}$}{
            \If{$\supwt' \leq \maxShotWeight$}{
                $N(\supwt', \errWeight', 1) \gets \psm{\supwt'} \NMq{q}{m, \shotLength, \errWeight'} \subspaceprob{\maxShotWeight}{\supwt'}{\errWeight'}$
            }
        }
    }
    \For{$\shots'\in\{2,\ldots,\shots\}$}{
        \For{$\errWeight'\in\{0,\ldots,\errWeight\}$}{
            \For{$\supwt'\in\{\errWeight',\ldots,\supwt\}$}{
                $\begin{aligned}
                N(\supwt', \errWeight', \shots') \gets \sum_{\errWeight'' = 0}^{\min\{\maxShotWeight, \errWeight'\}}\sum_{\supwt'' = \errWeight''}^{\min\{\maxShotWeight, \supwt'\}} & N(\supwt'-\supwt'', \errWeight'-\errWeight'', \shots'-1) \\
                & \cdot \psm{\supwt''} \NMq{q}{m, \shotLength, \errWeight''} \subspaceprob{\maxShotWeight}{\supwt''}{\errWeight''}
                \end{aligned}$
            }
        }
    }
\Return $N(\supwt,\errWeight,\shots)$
\end{algorithm}

\section{Appendix for Section~\ref{sec:random_decoding}}
\subsection{Definition of $\ucomp$ and Proof of Correctness}\label{sec:appendix-2}

The computation of $\ucomp_\maxShotWeight(\errWeightVec, \subwt)$ is described in Algorithm~\ref{alg:ucomp}. and Lemma~\ref{lem:ucomp_maximizes_prob} proofs its correctness. The randomization step in Line 6 of Algorithm \ref{alg:ucomp} is crucial to avoid bias towards specific positions, particularly when $\shots$ is large. This contrasts with a deterministic choice, which may lead to suboptimal results. In the Hamming case, where $\shotLength = 1$ and $n = \ell$, such randomization is essential for the effectiveness of Prange's generic decoder. However, our analysis does not explicitly take this randomness property into account and instead relies on $\probProfRandP{\maxShotWeight}{\ucomp_\maxShotWeight(\errWeightVec, \subwt), \errWeightVec}$, which is not randomized, despite $\ucomp_\maxShotWeight$ being a randomized function.

\begin{algorithm}
\caption{$\ucomp_\maxShotWeight(\errWeightVec, \subwt)$}\label{alg:ucomp}
\SetKwInOut{Input}{Input}\SetKwInOut{Output}{Output}
\Input{$\errWeightVec\in\wdecomp{w,\shots,\maxShotWeight}$ and $\subwt\in\NN$ with $\subwt=2\errWeightDiff \leq \errWeight$}
\Output{$\subwtVec\in\wdecomp{\subwt,\shots,\maxShotWeight}$ such that $\subwtVec = \argmax_{\subwtVec' \in \wdecomp{\subwt,\shots,\maxShotWeight}} \probProfRandP{\maxShotWeight}{\subwtVec', \errWeightVec}$}
$\subwtVec = [\subwt_1,\ldots,\subwt_\shots] \gets \errWeightVec$ \;
$\delta \gets \errWeight-\subwt$ \;
\While{$\delta > 0$}{
    $\set{J}_1 \gets \{i\in\{1,\ldots,n\} \st \subwt_i > 0\}$  \;
    $\set{J}_2 \gets \{i\in\set{J}_1 \st \errWeight_i = \min_{j\in\set{J}_1}\{\errWeight_j\}\}$ \;
    $\set{J}_3 \gets \{i\in\set{J}_2 \st \subwt_i = \max_{j\in\set{J}_2}\{\subwt_j\}\}$ \;
    $h \overset{\$}{\gets} \set{J}_3$ \;
    $\subwt_h \gets \subwt_h - 1$ \;
    $\delta \gets \delta - 1$ \;
}
\Return $\subwtVec$
\end{algorithm}

\begin{lemma}\label{lem:ucomp_maximizes_prob}
Let $\errWeightVec \in \wdecomp{\errWeight,\shots,\maxShotWeight}$ and let $\subwt \leq \errWeight$. Then, $\subwtVec = \ucomp_\maxShotWeight(\errWeightVec, \subwt)$, with $\ucomp_\maxShotWeight$ as in Algorithm \ref{alg:ucomp}, maximizes $\probProfRandP{\maxShotWeight}{\subwtVec,\errWeightVec}$, i.e.,
\begin{equation}
    \probProfRandP{\maxShotWeight}{\ucomp_\maxShotWeight(\errWeightVec, \subwt), \errWeightVec} = \max_{\subwtVec \in \wdecomp{\subwt,\shots,\maxShotWeight}} \probProfRandP{\maxShotWeight}{\subwtVec, \errWeightVec}.
\end{equation}
\end{lemma}
\begin{proof}
By~\eqref{eq:phiPrime} we have that
\begin{equation}
    \probProfRandP{\maxShotWeight}{\subwtVec,\errWeightVec} \defeq \prod_{i=1}^{\shots} \subspaceprob{\maxShotWeight}{\subwt_i}{\errWeight_i} = \frac{\quadbinomq{\errWeight_i}{\subwt_i}}{\quadbinomq{\maxShotWeight}{\subwt_i}}.
\end{equation}
The factor by what this expression is increased if we decrease $\subwt_i$ by $1$ is
\begin{align}
    \frac{\quadbinomq{\errWeight_i}{\subwt_i - 1} / \quadbinomq{\maxShotWeight}{\subwt_i - 1}}{\quadbinomq{\errWeight_i}{\subwt_i} / \quadbinomq{\maxShotWeight}{\subwt_i}} &= \prod_{j=1}^{\subwt_i} \frac{q^{\maxShotWeight-\subwt_i+j}-1}{q^{\errWeight_i-\subwt_i+j}-1} \cdot \prod_{j=1}^{u-1} \frac{q^{\errWeight_i - \subwt_i + 1 + j}-1}{q^{\maxShotWeight-\subwt_i+1+j}-1}  \\
    &= \frac{q^{\maxShotWeight}-1}{q^{\errWeight_i}-1} \cdot \underbrace{\prod_{i=j}^{\subwt_i - 1} \frac{(q^{\errWeight_i - \subwt_i +1 + j}-1)(q^{\maxShotWeight - \subwt_i + j}-1)}{(q^{\maxShotWeight-\subwt_i +1 + j}-1)(q^{\errWeight_i - \subwt_i +j}-1)}}_{=\frac{(q^{\errWeight_i}-1)(q^{\maxShotWeight+1}-q^{\subwt_i})}{(q^{\maxShotWeight}-1)(q^{\errWeight_i + 1}-q^{\subwt_i})}} \\
    &= \frac{q^{\maxShotWeight+1}-q^{\subwt_i}}{q^{\errWeight_i + 1}-q^{\subwt_i}}.
    \end{align}
This increase factor is monotonically increasing in $\subwt_i$ for a fixed $\errWeight_i$ and $\maxShotWeight$, and decreasing in $\errWeight_i$ for a fixed $\subwt_i$. Consequently, the maximum increase of \eqref{eq:phiPrime} is obtained by decreasing the largest $\subwt_i$ among the smallest $\errWeight_i$. By adopting a greedy approach and incrementally adjusting such positions, a global maximum can be reached as this strategy ensures optimal increase in subsequent steps. Thus, \eqref{eq:phiPrime} is optimized by incrementally decreasing $\subwt_i$ by one while maintaining $\subwt_i \geq 0$ and ensuring $\sum_{i=1}^{\shots} \subwt_i \geq \subwt$. This method aligns with the operations performed by $\ucomp_\maxShotWeight(\errWeightVec, \subwt)$ as described in Algorithm~\ref{alg:ucomp}.
\end{proof}

\subsection{Efficient Computation of $\qFactorSymRand_{\ell,\errWeight,\maxShotWeight}$}\label{sec:appendix-3}

Fortunately, we can employ a similar approach to \cite[Lemma 22]{puchingerGenericDecodingSumRank2022} and compute $\qFactorSymRand_{\ell,\errWeight,\maxShotWeight}$ as
\begin{equation}
    \qFactorSymRand_{\ell,\errWeight,\maxShotWeight} = \shots! \cdot M(\errWeight, \shots, \maxShotWeight, \subwt),
\end{equation}
where $M(\errWeight, \shots, \maxShotWeight, \subwt)$ can be computed using Algorithm \ref{alg:computeM}. To do so, we first initialize a global table ${\{M(\errWeight', \shots', \maxShotWeight', \subwt')\}}_{\errWeight'\leq\errWeight,\shots'\leq\shots}^{\maxShotWeight'\leq\maxShotWeight,\subwt'\leq\subwt}$ with $M(\errWeight', \shots', \maxShotWeight', \subwt') = -1$ for all entries. Then, we call Algorithm \ref{alg:computeM} with input parameters $\errWeight$, $\shots$, $\maxShotWeight$, and $\subwt$.

By applying arguments similar to those in \cite[Proposition 23]{puchingerGenericDecodingSumRank2022}, we can show that the complexity of computing $\qFactorSymRand_{\ell,\errWeight,\maxShotWeight}$ using this approach is $\softoh{\errWeight\subwt n^3 \maxShotWeight^3 \log_2(q)}$ and thus polynomially bounded.

\begin{algorithm}
\caption{Fill Table ${\{M(\errWeight', \shots', \maxShotWeight', \subwt')\}}_{\errWeight'\leq\errWeight,\shots'\leq\shots}^{\maxShotWeight'\leq\maxShotWeight,\subwt'\leq\subwt}$}\label{alg:computeM}
\SetKwInOut{Input}{Input}\SetKwInOut{Output}{Output}
\Input{
Integers $\errWeight'\leq\errWeight$, $\shots'\leq\shots$, $\maxShotWeight'\leq\maxShotWeight$, $\subwt'\leq\subwt$ \\
Global table ${\{M(\errWeight', \shots', \maxShotWeight', \subwt')\}}_{\errWeight'\leq\errWeight,\shots'\leq\shots}^{\maxShotWeight'\leq\maxShotWeight,\subwt'\leq\subwt}$ \\
Global parameters $q$ and $\maxShotWeight$
}
\Output{$M(\errWeight', \shots', \maxShotWeight', \subwt')$}
\If{$M(\errWeight', \shots', \maxShotWeight', \subwt') = - 1$}{
    \If{$\shots'=\errWeight'=\subwt'=0$}{
        $x \gets 1$
    }\Else{
        \If{$\shots' \geq 1$ \andand $0 \leq \errWeight' \leq \shots' \maxShotWeight'$ \andand $0 \leq \subwt' \leq \min\{\shots'\maxShotWeight, \errWeight'\}$}{
            $x\gets 0$ \;
            \For{$\errWeight_1 \in \{\maxShotWeight,\ldots,\left\lfloor \errWeight / \shots \right\rfloor \}$}{
                $\delta_\mathrm{min} \gets \max\{1, \shots'(\errWeight_1+1) - \errWeight)\}$ \;
                $\delta_\mathrm{max} \gets  \max\{ i \in \NN \st 1 \leq i \leq \shots'+1, \errWeight_1 i \leq \errWeight \}$ \;
                \For{$\delta \in \{\delta_\mathrm{min} , \ldots, \delta_\mathrm{max} \}$}{
                    $\subwt_1 \gets \max\{\subwt' - (\errWeight'-\delta \errWeight_1), 0 \}$
                    $\subwtVec^{(1)} \gets \ucomp_\maxShotWeight([\errWeight_1,\ldots,\errWeight_1], \subwt_1)$ \;
                    $\rho \gets \frac{1}{\delta !} \cdot \prod_{i=1}^{\delta} \left( {\quadbinomq{\maxShotWeight}{\subwt_i^{(1)}}}\cdot{\quadbinomq{\errWeight_1}{\subwt_i^{(1)}}}^{-1} \right) $ \;
                    $x\gets x + \rho \cdot M(\errWeight' - \delta \errWeight_1, \shots'-\delta, \errWeight_1 + 1, \subwt'-\subwt_1)$                
                }
            }
        }\Else{
        $x\gets 0$
        }
    }
}
\Return $M(\errWeight', \shots', \maxShotWeight', \subwt')$
\end{algorithm}
\FloatBarrier

\bibliography{zotereo-references}
\bibliographystyle{IEEEtran}

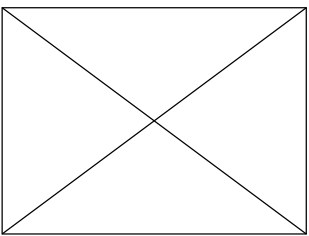
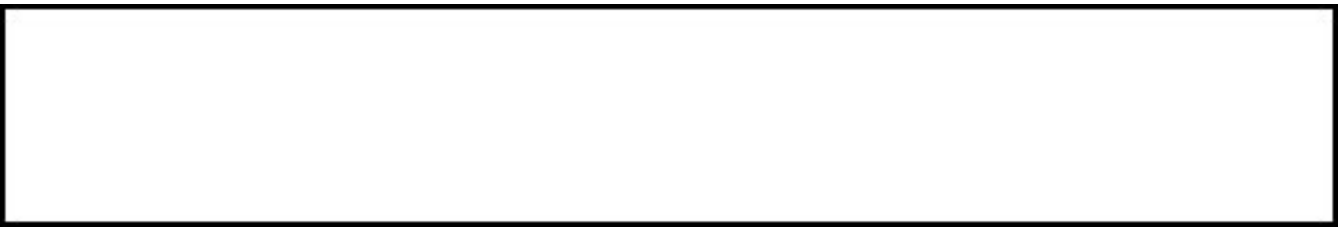
\end{document}